\documentclass[prx,aps,twocolumn,showpacs,floatfix]{revtex4-1}
\usepackage{amsfonts}
\usepackage{amssymb}
\usepackage{hyperref}
\usepackage{graphicx}
\usepackage{dcolumn}
\usepackage{bm,amsmath,verbatim}
\usepackage{mathrsfs}
\usepackage{color}
\usepackage{sidecap}
\allowdisplaybreaks
\usepackage{soul}

\usepackage{amsmath,amssymb,amsthm}
\usepackage{subfigure}
\usepackage{amsfonts,amssymb}
\usepackage{float}

\newtheorem{theorem}{Theorem}[section]

\hypersetup{colorlinks,
linkcolor=blue,          citecolor=blue,        filecolor=blue,      urlcolor=blue           }

\def\be{\begin{equation}}
\def\ee{\end{equation}}
\def\bea{\begin{eqnarray}}
\def\eea{\end{eqnarray}}
\def\bse{\begin{subequations}}
\def\ese{\end{subequations}}

\def\be{\begin{eqnarray}}
\def\ee{\end{eqnarray}}

\newcommand{\ket}[1]{|#1\rangle}
\newcommand{\bra}[1]{\langle #1 |}

\begin{document}

\title{Type-II quadrupole topological insulators}
\author{Yan-Bin Yang$^{1}$}
\author{Kai Li$^{1}$}
\author{L.-M. Duan$^{1}$}
\author{Yong Xu$^{1,2}$}
\email{yongxuphy@tsinghua.edu.cn}
\affiliation{$^{1}$Center for Quantum Information, IIIS, Tsinghua University, Beijing 100084, People's Republic of China}
\affiliation{$^{2}$Shanghai Qi Zhi Institute, Shanghai 200030, People's Republic of China}

\begin{abstract}
Modern theory of electric polarization is formulated by the Berry phase, which,
when quantized, leads to topological phases of matter. Such a formulation has
recently been extended to higher electric multipole moments, through the discovery
of the so-called quadupole topological insulator. It has been established
by a classical electromagnetic theory that in a two-dimensional material the
quantized properties for the quadupole topological insulator should satisfy a
basic relation.
Here we discover a new type of quadrupole topological insulator (dubbed type-II) that violates this relation
due to the breakdown of the correspondence that a Wannier band and an edge energy spectrum close their gaps simultaneously.
We find that, similar to the previously
discovered (referred to as type-I) quadrupole topological insulator, the
type-II hosts topologically protected corner states carrying fractional corner
charges. However, the edge polarizations only occur at a pair of boundaries in the type-II insulating phase, leading to the violation of the classical constraint.
We demonstrate that such new topological phenomena can appear from quench dynamics in non-equilibrium
systems, which can be experimentally observed in ultracold atomic gases.
We also propose an experimental scheme with electric circuits to realize such a new topological phase of
matter. The existence of the new topological insulating phase means that new
multipole topological insulators with distinct properties can exist in broader
contexts beyond classical constraints.
\end{abstract}
\maketitle

\section{Introduction}

Recently, the formulation of electric polarization based on the Berry phase has been extended to higher
electric multipole moments, such as quadrupole moments and octupole moments~\cite{Taylor2017Science,Taylor2017PRB}.
Similar to electric dipole moments, these multipole moments can be quantized due to crystalline symmetries,
such as reflection symmetries, giving rise to multipole topological insulators.
For a quadrupole topological insulator (QTI), besides the quantized quadrupole moment,
the quantized edge polarization and fractional corner charge arise. Such
fractional charges are associated with the appearance of
the topologically protected corner states. The QTI has ignited an intensive study of
higher-order topological insulators with $(n-m)$-dimensional edge states with $m>1$ for a $n$-dimensional system~\cite{Taylor2017Science,Taylor2017PRB,Fritz2012PRL,ZhangFan2013PRL,Slager2015PRB,FangChen2017PRL,Brouwer2017PRL,Bernevig2018SciAdv,Neupert2018NP,
Wan2017arXiv,Ryu2018PRB,Taylor2018PRB,Ezawa2018PRL,Khalaf2018PRB,Brouwer2018PRB,Fulga2018PRB,WangZhong2018PRL,ZhangFan2018PRL,
Roy2019PRB,Brouwer2019PRX,Fengliu2019PRL,SYang2019PRL,Kai2019arxiv},
in stark contrast to the conventional first-order topological insulators with $m=1$. Recently,
the QTI has been experimentally observed~\cite{Huber2018Nature,Bahl2018Nature,Thomale2018NP}.

For a two-dimensional (2D) square classical system with bulk quadrupole moments $q_{xy}$, a classical electromagnetic
theory based on the multipole expansion of an electric potential shows that equal amplitude edge polarizations $p_{x,y}^{\mathrm{edge}~\alpha}$ and corner charges $Q^{\mathrm{corner~}\alpha,\beta}$ can be induced so that $|q_{xy}|=|p_{x}^{\mathrm{edge~} \pm y}|=|p_{y}^{\mathrm{edge~}\pm x}|=|Q^{\mathrm{corner~}\pm x,\pm y}|$~\cite{Taylor2017Science,Taylor2017PRB}. Here,
$p_{x}^{\mathrm{edge~}\beta}$ ($p_{y}^{\mathrm{edge~}\alpha}$) describes the edge polarization per unit length along the $x$ ($y$) direction at the $y$-normal
($x$-normal) boundaries. The edges perpendicular to $y$ ($x$) are labelled by the Greek letters
$\beta=\pm y$ ($\alpha=\pm x$) with the sign denoting their relative positions.
The currently discovered quantum QTIs indeed respect these relations~\cite{Taylor2017Science,Taylor2017PRB} [see Fig.~\ref{fig1}(a)].
Provided that a system has bulk-independent boundary dipole moments besides the bulk quadrupole moments, the relation
is summarized as
$Q^{\mathrm{corner~}+x,+y}=p_y^{\mathrm{edge~}+x}+p_x^{\mathrm{edge~}+y}-q_{xy}$.
Previous research~\cite{Khalaf2019arXiv} has shown that a Wannier band and an edge energy spectrum should close their gaps simultaneously, giving
rise to edge polarizations of equal amplitude at the $x$-normal and $y$-normal boundaries.
The result is consistent with a previously established theory~\cite{Klich2011PRL}. This simultaneous gap closing ensures that
the classical relation has to be respected.

However, we find that the vanishing of the gap of the Wannier band is
not necessarily associated with the vanishing of the gap of the edge energy spectrum.
As a result, the edge polarization along one direction changes while that along the other direction remains,
leading to a novel type of quadrupole topological insulating phase where
$q_{xy}=|Q^{\mathrm{corner~}\pm x,\pm y}|=|p_x^{\mathrm{edge~}\pm y}|=e/2$, but $p_y^{\mathrm{edge~}\pm x}=0$, which is not equal to
$p_x^{\mathrm{edge~}\pm y}$, violating the classical relation [see Fig.~\ref{fig1}(b)].
It is worthwhile to note that the type-II quadrupole insulating phase is
fundamentally different from
the insulator with pure edge polarizations along one direction but without bulk quadrupole moments, where
the classical relation is also respected.
It is also fundamentally different from a topologically trivial bulk insulator attached with a pair of
Su-Schrieffer-Heeger (SSH) models. In this case, the bulk is trivial with zero quadrupole moments
while the type-II quadrupole insulator exhibits nonzero quadrupole moments in the bulk as detailed in
Appendix A.

Furthermore, we find anomalous quadrupole topological phases which have the zero
Berry phase of the Wannier bands (referred to as Wannier-sector polarization) but the nonzero
edge polarization. This tells us that the previously introduced nested Wilson loop formalism~\cite{Taylor2017Science,Taylor2017PRB}
cannot be used to characterize these insulating phases. Such phases arise because
the Wannier Hamiltonian is fundamentally different from a static system Hamiltonian,
given that the energy spectrum of the former is periodic, reminiscent of that of
the effective Hamiltonian in a periodically driven system~\cite{Fulga2018PRB}. This allows the Wannier bands to close their gaps at
either $\nu=0$ or $\nu=\pm 1/2$.
When the Wannier spectrum under open boundary conditions exhibit both
edge states at $\nu=0,\pm1/2$, the Berry phase vanishes; this resembles a
periodically driven system, where although the traditional topological invariant of a Hamiltonian vanishes,
the edge state persists~\cite{Levin2013PRX}. While such anomalous phenomena have also been found in a model with eight
bands, the Berry phase can still be used to characterize the edge polarization~\cite{Fulga2018PRB}. But that
does not work for our minimal quadrupole model. Here we introduce a new topological
invariant for a Wilson line to characterize the edge polarization. The Wannier gap closing at either $\nu=0$
or $\nu=\pm 1/2$ can be reflected by the change of the topological invariants.

While we demonstrate the existence of the type-II QTI based on a particular model,
we further show from a general perspective that it generically occurs in systems
with reflection symmetries and the particle-hole or chiral symmetry in the presence
of long-range hopping; such long-range hopping widely exists in various systems, such as photonic systems~\cite{YDChong2018PRL,Khanikaev2020NatPho},
polar molecules and Rydberg atoms~\cite{Syzranov2014NC,Peter2015PRA,Browaeys2016JPB,Weber2018QST,Leseleuc2019Science}, Shiba lattices~\cite{Menard2015NP,Menard2017NC,Ojanen2015PRL,Franke2018PSS}, and shaken lattices
where the effective hopping strength ratio can be tuned in a wide range~\cite{Holthaus2005PRL,Smith2011PRA,Sengstock2012PRL}.
Particularly,
a new type of corner modes has recently been observed in a photonic kagome crystals with long-range interactions~\cite{Khanikaev2020NatPho}. Based on this understanding, we further construct significantly simplified models that support the type-II QTI.
Furthermore,
the breakdown of the correspondence between the Wannier spectrum gap and the edge spectrum gap suggests that the Wannier band alone can induce
topological phase transitions despite the absence of the energy band gap closing.
Indeed, we find another new topological phase with quantized edge polarizations but without zero-energy corner modes and quadrupole moments.
This transition arises from the relevant Wannier gap closing characterized by the change of the winding number
for a Wilson line.

Similar to the Chern insulator that exhibits a winding of the Berry phase,
Ref.~\cite{Taylor2017Science} introduces a three-dimensional higher-order topological
insulator by breaking the reflection symmetries so that the quadrupole moment,
edge polarizations at all boundaries and corner charges all exhibit a winding, which is
associated with the presence of chiral hinge modes.
Analogously, based on the type-II QTI, we find two types of new three-dimensional higher-order
topological insulators. In one insulator, the winding of the quadrupole moment, corner charges
and edge polarizations at a pair of boundaries appears associated with chiral hinge modes.
In the other one, only the winding for the edge polarizations at a pair of boundaries happens.

Nonequilibrium dynamics under unitary time evolution between distinct topological phases
have been studied in cold atoms from various aspects~\cite{Rigol2015NC,Cooper2015PRL,Vajna2015PRB,Heyl2016PRB,Huang2016PRL,Hu2016PRL,Refael2016PRL,Zhai2017PRL,Weitenberg2018NP,ShuaiChen2018PRL,ShuChen2018PRB,Xiongjun2018scibull,Ueda2018PRL,Yong2018PRB,Cooper2018PRL,Cooper2019PRB,Weitenberg2019NC}. Given that the unitary
evolution does not change the energy spectra of a parent Hamiltonian~\cite{Ueda2018PRL}, the topology
of the evolving states remains unchanged if symmetries of the parent Hamiltonian
are preserved during unitary time evolution~\cite{Rigol2015NC,Cooper2015PRL,Ueda2018PRL}.
Specifically, let us start with a
ground state $|\psi_{\bf k}\rangle$ of an initial Hamiltonian $H_i({\bf k})$ and then suddenly change the
Hamiltonian to $H_f({\bf k})$ by tuning system parameters. The state then evolves under the final Hamiltonian,
i.e., $|\psi_{\bf k}(t)\rangle=e^{-iH_f({\bf k}) t}|\psi_{\bf k}\rangle$.
Since the evolving state is an eigenstate of a
parent Hamiltonian $H_p({\bf k})=e^{-iH_f({\bf k}) t}H_i({\bf k})e^{iH_f({\bf k}) t}$, the topological
properties of the evolving states are dictated by the parent Hamiltonian. If relevant symmetries of $H_p$
do not change as time progresses, it has been shown that $H_i$ and $H_f$ share the same topological property given $\text{det}(H_f)=\text{det}(H_i)$~\cite{Ueda2018PRL}.
On the other hand, a physical quantity,
such as the Berry phase, can change continuously if symmetries of the parent Hamiltonian
are allowed to change during unitary time evolution~\cite{Cooper2018PRL,Cooper2019PRB}.
We here perform an investigation of the quench dynamics across distinct phases in
the Benalcazar-Bernevig-Hughes (BBH) model and find that the type-II QTI phase emerges
as time evolves, even though the symmetries are preserved during the time evolution. This shows that new topological phases can arise due to the Wannier
gap closing in nonequilibrim dynamics.

The paper is organized as follows. In Sec.~\ref{sec2}, we demonstrate the existence of the type-II
anomalous QTI (AQTI) by studying the topological properties of a tight-binding model.
In Sec.~\ref{sec3}, we introduce a topological invariant for a Wilson line to characterize the edge polarization
change due to the Wannier band gap closing and bulk energy gap closing. In Sec.~\ref{sec4}, we
provide a general analysis for the existence of the type-II QTI and show that the type-II phase
generically appears in systems with reflection symmetries and the particle-hole or chiral symmetry
in the presence of long-range hopping. Based on this analysis, we construct several
significantly simplified models supporting the type-II QTI. We also present another novel
topological phase with quantized edge polarizations but without zero-energy corner modes and quadrupole moments.
In Sec.~\ref{sec5},
we study the pumping phenomenon and novel three-dimensional higher-order topological insulators.
In Sec.~\ref{sec6}, we show that the type-II QTI can appear in quench dynamics through unitary time evolution.
In
Sec.~\ref{sec7}, we demonstrate that the quench dynamics can be experimentally realized in cold atoms and
further propose an experimental scheme using electric circuits to realize these new topological phenomena.
Finally, the conclusion is presented in Sec.~\ref{sec8}.

\begin{figure}[t]
  \includegraphics[width=3.3in]{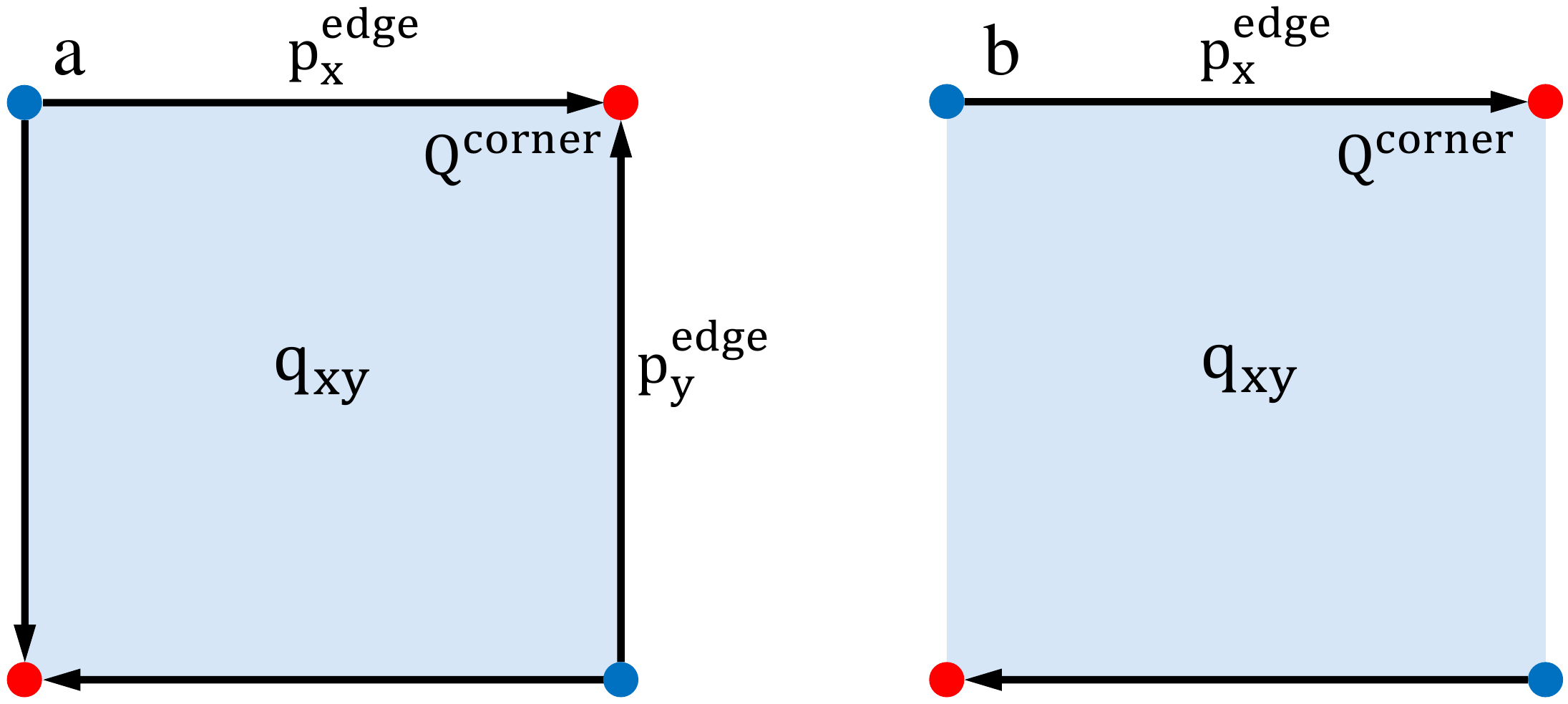}
  \caption{Schematics of edge polarizations and corner charges.
  a, Edge dipole moments exist at all boundaries in
  type-I QTI and b, they exist only
  at the boundaries perpendicular to $y$ in type-II QTIs. Corner
  charges $Q^{\mathrm{corner}}=\pm e/2$ (marked by different colors) appear in both phases.}
\label{fig1}
\end{figure}

\section{The Type-II QTI}
\label{sec2}

To generate the type-II quadrupole topological insulating phase, we consider a 2D
crystal with four sites in each unit cell and long-range hopping between unit cells.
We enforce two reflection symmetries $M_x$:
$x\rightarrow -x$ and $M_y$: $y\rightarrow -y$, in order to maintain
the quantization of the quadrupole moment, corner charges and edge polarizations.
Specifically, the system is described by the following Hamiltonian
\begin{equation}
H=\sum_{\bf R}\sum_{d_x,d_y} \hat{c}^\dagger_{{\bf R}+d_x{\bf e}_x+d_y{\bf e}_y}h_{(d_x d_y)} \hat{c}_{\bf R},
\label{Ham1}
\end{equation}
where $\hat{c}_{\bf R}^\dagger=(\begin{array}{cccc}
 \hat{c}_{{\bf R},1}^\dagger & \hat{c}_{{\bf R},2}^\dagger & \hat{c}_{{\bf R},3}^\dagger & \hat{c}_{{\bf R},4}^\dagger)
\end{array}$ with $\hat{c}_{{\bf R},\alpha}^\dagger$ ($\hat{c}_{{\bf R},\alpha}$) creating (annihilating) an electron
at a sublattice denoted by the index $\alpha=1,2,3,4$ within a unit cell denoted by a lattice vector ${\bf R}=R_x {\bf e}_x+R_y{\bf e}_y$ with $R_x$ and $R_y$ being integers. The sum over $d_x$ and $d_y$ depicts the on-site potential and the
electron tunneling between distinct sublattices within a unit cell when $d_x=d_y=0$ and the tunneling
between neighboring unit cells, otherwise. For simplicity, we choose the lattice constant $a_{x,y}=1$. Here, we consider the case with long-range hopping including up to the hopping
with $(d_x=\pm1,d_y=\pm2)$ and $(d_x=\pm2,d_y=\pm1)$. To be specific,
we choose $h_{(00)}=\gamma(\tau_1\sigma_0+\tau_2\sigma_2)+\Delta \tau_3\sigma_2+\delta \tau_3 \sigma_0$,
$h_{(10)}=t_1 (\tau_1 \sigma_0-i \tau_2 \sigma_3)+t_1^\prime \tau_3 \sigma_2$,
$h_{(01)}=-it_1\tau_2\sigma_-+t_1^\prime \tau_1\sigma_0$, $h_{(11)}=t_2(-i\tau_0\sigma_3-i\tau_3 \sigma_-+\tau_1\sigma_0-i \tau_2 \sigma_3)
-it_2^\prime \tau_2 \sigma_-$, $h_{(20)}=t_2(-i\tau_3 \sigma_3+i\tau_0\sigma_1+\tau_3\sigma_2)$,
$h_{(02)}=-i t_2 \tau_2\sigma_- - it_2^\prime \tau_3 \sigma_-$, $h_{(21)}=t_2^\prime (-\tau_1\sigma_0+i \tau_2 \sigma_3)$
and $h_{(12)}=it_2^\prime \tau_2 \sigma_-$,
where $\sigma,\tau$ denote Pauli matrices for the degrees of freedom within a unit cell and
$\sigma_{\pm}=\sigma_1\pm i \sigma_2$. The other tunnelling matrices
for $d_x\neq 0$ and $d_y\neq 0$ can be obtained by $h_{(-d_x-d_y)}=(h_{(d_xd_y)})^\dagger$, $\hat{m}_x h_{(d_xd_y)}\hat{m}_x^\dagger=h_{(-d_xd_y)}$ and
$\hat{m}_y h_{(d_xd_y)}\hat{m}_y^\dagger=h_{(d_x-d_y)}$ required by the Hermiticity and reflection symmetries of the system
with $\hat{m}_x=\tau_1 \sigma_3$ and $\hat{m}_y=\tau_1 \sigma_1$ for zero $\delta$.
Specifically, we set $\Delta = t_1 = 0.3$, $t_1^\prime=0.2$, $t_2=0.15$, and $t_2^\prime=0.1$.
The Hamiltonian in momentum space can be found in Appendix B. When $\delta=0$, this Hamiltonian also respects
the particle-hole symmetry to protect the zero-energy corner modes, i.e.,
$\mathcal{C}H\mathcal{C}^{-1}=H$ with $\mathcal{C}\hat{c}_{{\bf R},\alpha}\mathcal{C}^{-1}=\sum_{\beta=1}^{4}(\tau_3\sigma_0)_{\alpha\beta}\hat{c}_{{\bf R},\beta}^\dagger$.
The expression for the particle-hole symmetry in momentum space can be found in Appendix B.
Compared with the model in Ref.~\cite{Taylor2017Science}, our model does not preserve
the time-reversal symmetry $\Theta=\kappa$ ($\kappa$ is the complex conjugation operator) and thus the system does not have the $PT=\hat{m}_x \hat{m}_y\kappa$ symmetry, which guarantees the
double degeneracy of the energy bands. Our model breaks this symmetry, lifting the energy degeneracy.

\begin{figure*}[t]
\includegraphics[width=\textwidth]{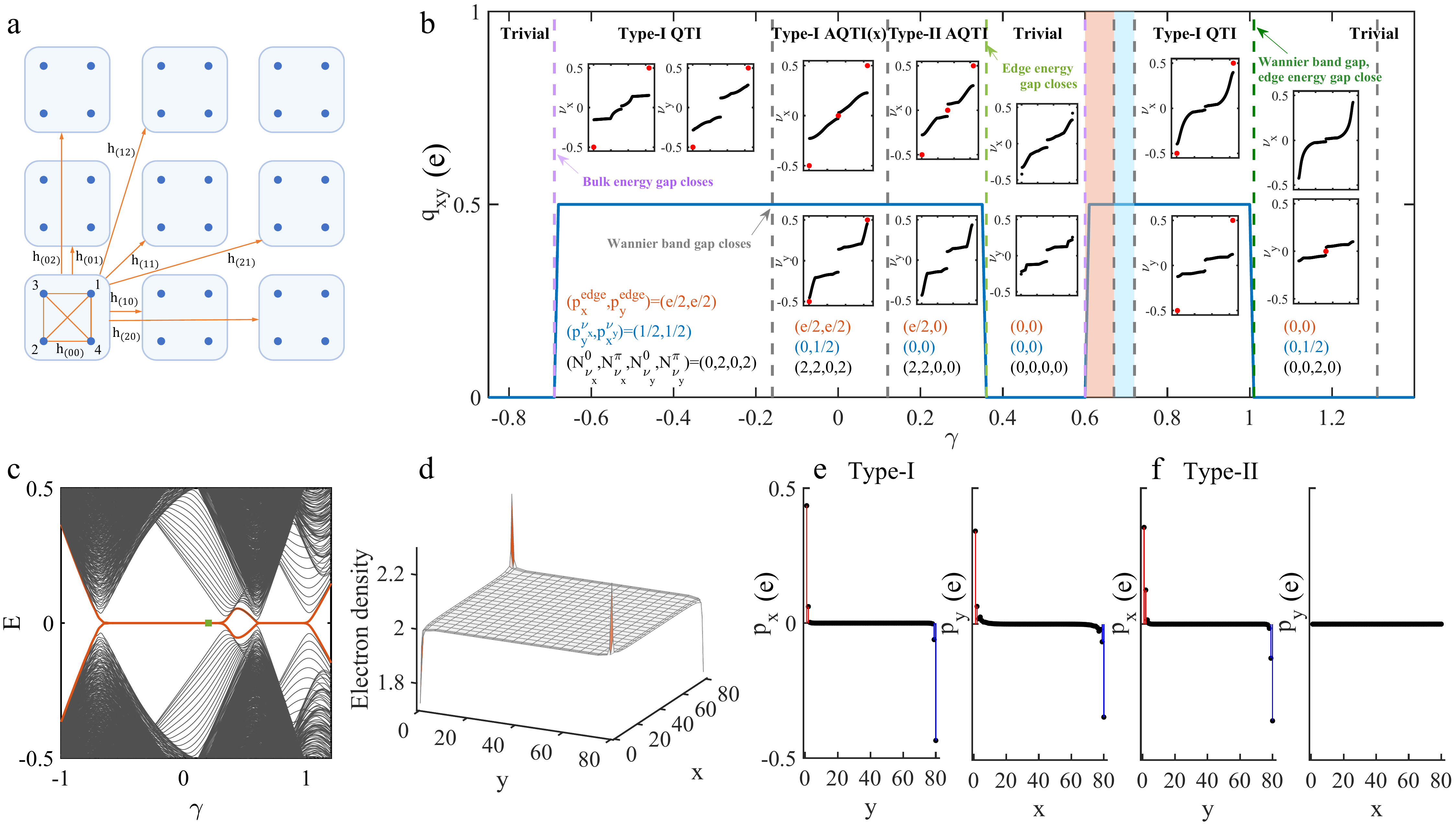}
  \caption{Schematics of our model, phase diagram and topological properties.
  a, Schematics of the tunnelling in our tight-binding model. b, Phase diagram with respect to
  a system parameter $\gamma$. The quadrupole moment evaluated for a $80\times80$ system
  is plotted as a blue line. In the phase diagram, we observe
  the topologically trivial insulator, the type-I QTI, the type-I anomalous
  quadrupole topological insulator (AQTI) (anomalousness exists in the Wannier band $\nu_x$)
  and the type-II AQTI. The subsets display the Wannier spectrum $\nu_x$ ($\nu_y$) in a cylinder
  geometry with periodic boundaries along $x$ ($y$) and open ones along $y$ ($x$) with the isolated Wannier centers
  highlighted by red circles. The edge polarization $(p_x^{\mathrm{edge}},p_y^{\mathrm{edge}})$ calculated using
  the formula~(\ref{edgePEq}), the Wannier-sector polarization
  $(p_y^{\nu_x},p_x^{\nu_y})$, the number of edge states of the Wannier Hamiltonian
  $N_\nu \equiv (N_{\nu_x}^0, N_{\nu_x}^{\pi}, N_{\nu_y}^0, N_{\nu_y}^{\pi})$ are also shown.
  The vertical dashed lines represent the critical points where the bulk energy gap, the edge energy gap, or the
  Wannier spectrum gap vanish. The light red and blue regions denote the type-I AQTI(xy) (anomalousness exists in both Wannier band $\nu_x$ and $\nu_y$) and type-I AQTI(x). Richer phase diagram for the topologically trivial phase can be found in Appendix C. c, The energy spectrum as a function of $\gamma$ for open boundary conditions along both $x$ and $y$ directions with zero-energy corner modes being highlighted by a red line.
  d, The electron density distribution in a typical type-II AQTI phase with the zero-energy corner modes marked by the green square in c. Here, a very small $\delta$ is imposed so that two corner states are occupied.
  e,f, The edge polarization profiles for a type-I AQTI state at $\gamma=-0.1$ and a type-II AQTI state at $\gamma=0.2$, respectively. }
\label{fig2}
\end{figure*}

Our numerical computation shows the presence of an energy gap in momentum space energy spectra with
respect to $\gamma$ unless $\gamma=-0.69$ and $\gamma=0.61$, where the energy gap vanishes.
This implies that the bulk of the system exhibits the insulating property in the gapped regions
at half filling. The insulating feature can also be seen in the energy spectra under
open boundary conditions along both $x$ and $y$
directions [see Fig.~\ref{fig2}(c)]. However, for $-0.69<\gamma<0.34$ and $0.61<\gamma<1.03$, imposing open boundaries render the appearance of four zero-energy states
localized at the corners corresponding to a second-order topological insulator, where corner states exist at
the boundaries of boundaries as shown in Fig.~\ref{fig2}(d).
In other parameter regions, we do not find zero-energy corner states.

These corner states give rise to fractional charges $\pm e/2$ localized at the corners.
Such corner charges are numerically calculated by performing the integration
of the charge density over a quadrant of the system~\cite{Taylor2017PRB},
\begin{equation}
Q^{\mathrm{corner}~-x,-y}=\sum_{R_x=1}^{N_x/2}\sum_{R_y=1}^{N_y/2}\rho({\bf R}),
\end{equation}
where $\rho({\bf R})=2e-e\sum_{n=1}^{N_{\mathrm{occ}}}\sum_{\alpha=1}^4|[u^{n}]^{{\bf R},\alpha}|^2$ is the charge density with
the first term contributed by the atomic positive charges and the second term by the electron distribution
described by the $n$th occupied eigenstate $|u^{n}\rangle$ of our Hamiltonian
under open boundary conditions with $[u^{n}]^{{\bf R},\alpha}$
being the component at the site $\bf R$ with orbital index $\alpha$.
To calculate the corner charge, we include a small $\delta$
so that the fourfold degeneracy of the zero-energy states is lifted, leading to two corner states with positive energy
and the other two with negative energy. At half filling, only two corner states are occupied. Suppose the atoms contribute
$+2e$ charge in each unit cell. This gives us the corner-localized fractional charges $\pm e/2$ in the limit
$\delta \rightarrow 0$.

To show that the insulator with zero-energy corner states is a QTI, we calculate their quadrupole moments
based on the following formula~\cite{Wheeler2018arXiv,Cho2018arXiv}
\begin{equation}
q_{xy}=\frac{1}{2\pi}\mathrm{Im}[\log \langle \Psi_G|\hat{U}_2|\Psi_G\rangle ],
\end{equation}
where $\hat{U}_2=e^{2\pi i\sum_{\bf r}\hat{q}_{xy}({\bf r})}$ with $\hat{q}_{xy}({\bf r})=xy\hat{n}({\bf r})/(L_xL_y)$ being
the quadrupole moment per unit cell measured with respect to $x=y=0$ at the site ${\bf r}$, $L_x$ and $L_y$ are the length of
the system along $x$ and $y$ directions, respectively, the sum is over $(x,y)\in(0,L_x]\times (0,L_y]$, $\hat{n}({\bf r})$
is the number of electrons at the site ${\bf r}$ and $|\Psi_G\rangle$ is the many-body ground state of a system.
Our calculation is performed under periodic boundary conditions. Note that the atomic positive charge contribution
has been deducted.

Our numerical results show that the system has a quantized quadrupole moment
$q_{xy}=e/2$ (protected by the reflection symmetry) in the region where the zero-energy corner modes exist, as shown in Fig.~\ref{fig2}(b). The change of the quadrupole
moment is associated with the vanishing of either a bulk energy gap or an edge energy gap, reflecting
the topological properties of the quadrupole insulating phase. We note that while Ref.~\cite{Watanabe2019PRB} points out
some difficulties for evaluating the quadrupole moment in a generic system using the formula proposed in Ref.~\cite{Wheeler2018arXiv,Cho2018arXiv}, the calculated quadrupole moments in our model are reasonable as verified in Appendix A.

To characterize the edge polarization (for example, the polarization along $x$), we consider
the Wilson loop
\begin{equation}
\mathcal{W}_x=F_{x,k_x+(N_x-1)\delta k_x}\cdots F_{x,k_x}
\end{equation}
and similarly for $\mathcal{W}_y$, where the subscript $x$ ($y$) indicates that the Wilson loop is
defined following the path along $x$ ($y$).
Here $[F_{x,k_x}]^{mn}=\langle u_{k_x+\delta k_x}^m|u_{k_x}^n\rangle$, where $\delta k_x=2\pi/N_x$ with $N_x$ being the number of unit cells along
$x$, and $k_x$ is the quasimomentum along $x$ due to the imposed periodic boundary condition along that direction~\cite{Taylor2017Science,Taylor2017PRB}.
For a system with periodic boundaries along $y$, $|u_{k_x}^n\rangle=|u_{k_x,k_y}^n\rangle$ refers to
the occupied eigenstate
of a system Hamiltonian in momentum space with $n$ being the band index.
Yet, for a system with open boundaries, $|u_{k_x}^n\rangle$
refers to the occupied eigenstate of the Hamiltonian with open boundaries along $y$.
The Wannier Hamiltonian $H_{\mathcal{W}_x}$ is defined by $\mathcal{W}_x\equiv e^{iH_{\mathcal{W}_x}}$ with
its eigenvalues $2\pi\nu_x$ referred to as the Wannier spectrum, where $\nu_x$ is the Wannier center that determines the polarization that each state contributes. Here, the reflection symmetry maintains the vanishing of the total polarization in the bulk~\cite{Taylor2017Science,Taylor2017PRB}. Let us first consider a cylinder geometry with
open boundaries along $y$. In this case, when the Wannier spectrum exhibits isolated eigenvalues at $\nu_x=\pm 1/2$,
with the corresponding eigenstates being localized at two opposite $y$-normal boundaries, the system has
the boundary polarization $p_x^{\mathrm{edge}}$.
The appearance of the edge polarization stems from the topological property of the bulk. There are two routes to the emergence of the edge states in the Wannier spectrum. One is through the change of the
topological property of the Wannier bands, the eigenvalues $\nu_x(k_y)$ of $H_{\mathcal{W}_x}(k_y)$, under periodic
boundary conditions along $y$, by closing the Wannier band gap at $\nu_x=\pm1/2$.
An alternative route is provided by closing either the bulk energy gap or edge energy gap, resulting in an abrupt change of the quadrupole moment.

To distinguish between the type-I and type-II QTIs,
it is necessary to characterize their edge polarization by the Wilson loop.
In a torus geometry, the eigenvalues of the Wilson loop $\mathcal{W}_x(k_y)$ [similarly for $\mathcal{W}_y(k_x)$]
takes the form of $e^{i2\pi \nu_x^j(k_y)}$ with $j=1,\cdots, N_{\mathrm{occ}}$ and $N_{\mathrm{occ}}$
being the number of occupied bands since
$\mathcal{W}_x(k_y)$ is unitary, implying that $H_{\mathcal{W}_x}(k_y)$ has eigenvalues
$2\pi \nu_x^j(k_y)$. Because $e^{i2\pi \nu_x^j(k_y)}$ repeats over intervals of $1$ for $\nu_x^j(k_y)$, we restrict
$\nu_x^j(k_y)$ to $(-0.5,0.5]$. Because of the reflection symmetry $M_x$: $x\rightarrow -x$,
${\pm 2\pi \nu_x^j(k_y)}$ are both eigenvalues of $H_{\mathcal{W}_x}(k_y)$, so that the Wannier centers appear in
pairs $[-\nu_x^j(k_y),\nu_x^j(k_y)]$, maintaining the vanishing of the bulk dipole moments in our model.
The Wannier bands can be gapped with one band $\nu_x^-(k_y)\in(-0.5,0)$ and the other $\nu_x^+(k_y)\in(0,0.5)$
similar to a conventional band. However, it turns out that there are two gaps for the Wannier bands: one is
around $\nu_x=0$ and the other around $\nu_x=\pm 1/2$; the gaps can close at either $\nu_x=0$ or $\nu_x=\pm 1/2$.

Fig.~\ref{fig2} illustrates that in the type-I phase, both Wannier spectra $\nu_x$ and $\nu_y$
under corresponding open boundary conditions
(open along $y$ and $x$, respectively)
exhibit isolated eigenvalues
at $\nu_x=\pm 1/2$, which disappear under periodic boundary conditions,
implying that they are contributed by the boundary states. Their emergence indicates
the presence of the dipole moments at all the four boundaries.
Remarkably, in the type-II phase, only $\nu_x=\pm1/2$ occurs but not for $\nu_y$,
implying that the dipole moments only exist at the $y$-normal edge but not at
the $x$-normal one, as shown in Fig.~\ref{fig1}(b).

To show that the dipole moments are localized at boundaries, we calculate the polarization distribution by
\begin{equation}
p_x(R_y)=\sum_{j=\pm}\rho^j(R_y)\nu_x^j,
\end{equation}
where $\rho^j(R_y)=\frac{1}{N_x}\sum_{k_x,\alpha}|\sum_{n=1}^{N_{\mathrm{occ}}}[u_{k_x}^n]^{R_y,\alpha}[\nu_{k_x}^j]^n|^2$ is the probability density of
the hybrid Wannier functions~\cite{Taylor2017Science,Taylor2017PRB}, $[\nu_{k_x}^j]^n$ is the $n$th entry of the $j$th eigenvector $|\nu_{k_x}^j\rangle$ of the Wannier Hamiltonian corresponding to the Wannier center
$\nu_x^j$ in a cylinder geometry with open boundaries along $y$,
and $[u_{k_x}^n]^{R_y,\alpha}$ describes the collection of entries of the $n$th occupied eigenstate of our Hamiltonian
in the same boundary configuration. The edge polarization is defined as the sum of $p_x(R_y)$ over a half along $y$,
i.e.,
\begin{eqnarray} \label{edgePEq}
p_x^{\mathrm{edge~}-y}&=&\sum_{R_y=1}^{N_y/2}p_x(R_y) \\ \nonumber
&=&-p_x^{\mathrm{edge~}+y}=-\sum_{R_y=N_y/2+1}^{N_y}p_x(R_y),
\end{eqnarray}
which is quantized. The formulation of
the edge polarization along $y$ is similar.

Figure~\ref{fig2}(e,f) show that
the polarization, if exists, is indeed exponentially localized at the boundaries and opposite boundaries have opposite polarizations. While the polarization has a distribution along the direction perpendicular
to a boundary, their total value for an edge is quantized, i.e., $p_{x,y}^{\mathrm{edge}}=\pm e/2$.
Despite the presence of the edge dipole moments, the total polarization vanishes as opposite boundaries have opposite edge polarizations. In the type-II phase, the polarization along $y$ remains zero, in stark contrast to
the corresponding nonzero edge polarization in the type-I phase.

Although the Wannier Hamiltonian can have the edge states at $\nu_x=\pm 1/2$ or $\nu_x=0$,
only the former contributes to the edge polarization. In fact, both of these states
at $\nu_x=\pm 1/2$ or $\nu_x=0$ can appear simultaneously. In that case, we will show that
the Wannier-sector polarization is zero and thus cannot be used to characterize these edge states.
We refer to such a topological insulating phase as an anomalous QTI.

The number of the edge states of the Wannier Hamiltonian $H_{\mathcal{W}_x}$ ($H_{\mathcal{W}_y}$) changes
when the gap of the bulk energy spectrum $E(k_x,k_y)$, edge energy spectrum
$E^{\mathrm{edge},y}(k_x)$ [$E^{\mathrm{edge},x}(k_y)$ ] at the $y$-normal ($x$-normal) edges or Wannier band
$\nu_x(k_y)$ [$\nu_y(k_x)$] closes,
reflecting the topological feature of the edge polarization. Associated with the
vanishing of the bulk energy gap or edge energy gap is
the change of the number of the edge states of the Wannier Hamiltonian
corresponding to both $\nu_x=0$ and $\nu_x=\pm 1/2$.
However,
when the Wannier bands
close their gap at $\nu_x=0$ or $\nu_x=\pm1/2$,
only the number of the edge states with the same eigenvalue as that where
the gap vanishes changes.
Specifically,
the bulk energy gap closes at $\gamma=-0.69$ and $\gamma=0.61$, leading to the phase transition
between the topologically trivial phase with $N_\nu=(2,0,2,0)$ and type-I quadrupole insulating phase with $N_\nu=(0,2,0,2)$,
and the transition between a trivial phase with $N_\nu=(0,0,0,0)$ and the type-I anomalous phase with $N_\nu=(2,2,2,2)$, respectively.
Here, $N_\nu \equiv (N_{\nu_x}^0, N_{\nu_x}^{\pi}, N_{\nu_y}^0, N_{\nu_y}^{\pi})
$ with $N_{\nu_\lambda}^{\epsilon}$ ($\lambda=x,y$) denoting the number of the edge sates of the Wannier
Hamiltonian $\mathcal{W}_{\lambda}$ corresponding to the eigenvalue $\epsilon=0,\pi$.
The vanishing gap of the Wannier bands divides the quadrupole insulating phase for $-0.69<\gamma<0.34$
into three regions: type-I QTI, type-I AQTI and type-II AQTI; each gap closure gives rise to the change
of the number of the corresponding edge states, as shown in Fig.~\ref{fig2}(b).
The edge energy spectrum at
the $y$-normal boundary closes its gap at $\gamma=0.34$, resulting in the phase transition between
a trivial phase with $N_\nu=(0,0,0,0)$ and the type-II AQTI with $N_\nu=(2,2,0,0)$.

Our results show that the Wannier bands $\nu_x(k_y)$ [$\nu_y(k_x)$] do not necessarily close their gap
at $\nu=\pm 1/2$ at the same time
as the edge energy spectrum localized at the $x$-normal ($y$-normal) boundaries (see Sec. IV for details).
If their gaps always vanish simultaneously, there should be equal number of the edge states of the
Wannier Hamiltonian $\mathcal{W}_x$ and $\mathcal{W}_y$ with eigenvalues $\nu=\pm 1/2$,
giving rise to the same amplitude edge polarization at the $x$-normal and
$y$-normal boundaries [see the case for
$\gamma=1.03$ in Fig.~\ref{fig2}(b)]~\cite{Khalaf2019arXiv}. With this violation in our model, we find the type-II phase where
the dipole moments only exist at the boundaries vertical to $y$.

\section{A topological invariant for a Wilson line}
\label{sec3}

The Wannier-sector polarization for the Wannier band $\nu_x^{\pm}$ (similarly for $\nu_y^{\pm}$) is defined as~\cite{Taylor2017Science,Taylor2017PRB}
\begin{equation}
p_y^{\nu_x^{\pm}}=-\frac{1}{(2\pi)^2}\int_{BZ} d^2{\bf k} \mathcal{A}_{y,{\bf k}}^{\pm},
\end{equation}
where $\mathcal{A}_{y,{\bf k}}=-i\langle w_{x,{\bf k}}^{\pm}|\partial_{k_y}|w_{x,{\bf k}}^{\pm}\rangle$ is
the Berry connection over the Wannier bands $\nu_x^{\pm}$, respectively.
$|w_{x,{\bf k}}^{\pm}\rangle=\sum_{n=1,2}|u_{{\bf k}}^n\rangle [\nu_{x,{\bf k}}^\pm]^n$
with $|u_{{\bf k}}^n\rangle$ being the $n$th occupied eigenstate of our Hamiltonian in momentum space
and $[\nu_{x,{\bf k}}^\pm]^n$ being the $n$th entry of the eigenstate $|\nu_{x,{\bf k}}^\pm\rangle$
of the Wannier Hamiltonian in a torus geometry.

The Wannier-sector polarizations were previously introduced to characterize the edge polarizations of
the type-I QTI. However, when it becomes anomalous, we find that a corresponding Wannier-sector polarization
vanishes [see their values $(p_y^{\nu_x},p_x^{\nu_y})$ in Fig.~\ref{fig2}(b)], suggesting that
the Wannier-sector polarization cannot uniquely identify the edge dipole moments.

\begin{figure}[t]
\includegraphics[width=3.3in]{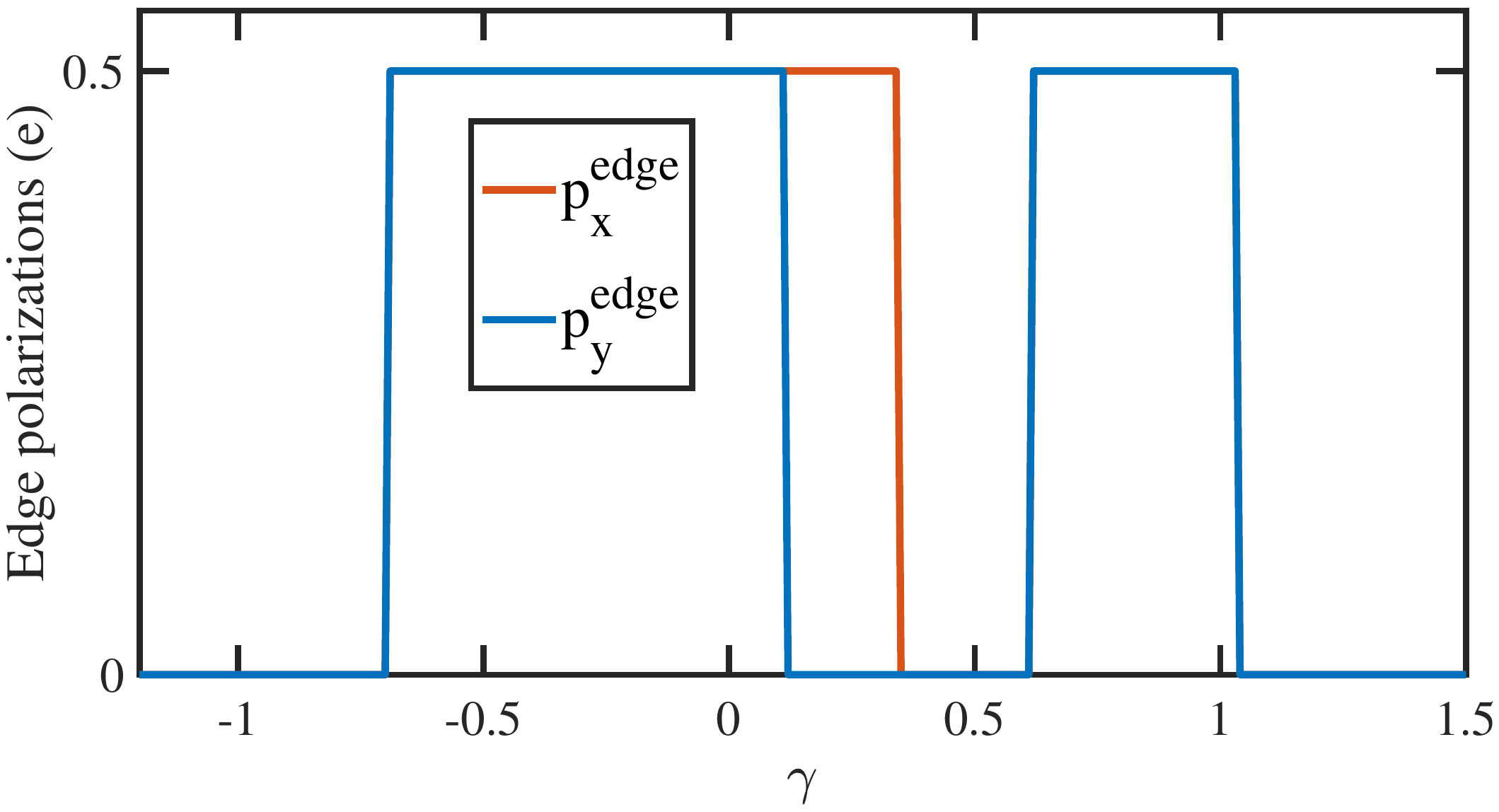}
  \caption{The edge polarizations calculated using the formula~(\ref{windingR})
  by choosing a gauge such that $W_{\nu_x}^{\epsilon=\pi}(\gamma_0=-1)=W_{\nu_y}^{\epsilon=\pi}(\gamma_0=-1)=0$.
  The edge polarization $p_x^{\textrm{edge}}$ and $p_y^{\textrm{edge}}$ in units of $e$ are plotted as the red and blue lines, respectively.
   }
\label{fig3}
\end{figure}

To characterize the edge polarization along $x$ (similarly along $y$), we will introduce a topological invariant based on
the Wilson line with respect to $\epsilon$ defined as
\begin{equation}
\mathcal{W}_{k_x\leftarrow 0}^\epsilon(k_y) \equiv \mathcal{W}_{k_x\leftarrow 0}(k_y)e^{-iH_{W_x}^{\epsilon}
(k_y)k_x/(2\pi)},
\end{equation}
where $\mathcal{W}_{k_x\leftarrow 0}(k_y)=F_{x,(k_x,k_y)}F_{x,(k_x-\delta k_x,k_y)}\cdots F_{x,(0,k_y)}$
is the Wilson line and $H_{W_x}^{\epsilon}(k_y)\equiv -i\log_\epsilon \mathcal{W}_{2\pi\leftarrow 0}(k_y)$ is
the Wannier Hamiltonian with respect to $\epsilon$ with $\log_\epsilon(e^{i\phi})=i\phi$ with
$\epsilon\le \phi<\epsilon+2\pi$. The reflection symmetry leads to (the details are presented in Appendix D):
\begin{eqnarray}
S\mathcal{W}_{\pi\leftarrow 0}^{\epsilon=0}(k_y)S^\dagger&=&-\mathcal{W}_{\pi\leftarrow 0}^{\epsilon=0}(k_y) \\
S\mathcal{W}_{\pi\leftarrow 0}^{\epsilon=\pi}(k_y)S^\dagger&=&\mathcal{W}_{\pi\leftarrow 0}^{\epsilon=\pi}(k_y),
\end{eqnarray}
where $S=\sigma_z$. In the basis consisting of eigenvectors of $S$,
\begin{equation}
\mathcal{W}_{\pi\leftarrow 0}^{\epsilon=0}(k_y)=
\left(
  \begin{array}{cc}
    0 & U^{\epsilon=0}_+(k_y) \\
    U^{\epsilon=0}_-(k_y) & 0 \\
  \end{array}
\right)
\end{equation}
and
\begin{equation}
\mathcal{W}_{\pi\leftarrow 0}^{\epsilon=\pi}(k_y)=
\left(
  \begin{array}{cc}
    U^{\epsilon=\pi}_+(k_y) & 0 \\
    0 & U^{\epsilon=\pi}_-(k_y) \\
  \end{array}
\right).
\end{equation}
Hence, we can define a winding number at $\epsilon=0,\pi$ as
\begin{equation}
W_{\nu_x}^{\epsilon}=\frac{1}{2\pi i}\int_0^{2\pi} dk_y\partial_{k_y}\log U_{+}^{\epsilon}(k_y).
\label{WindFormula}
\end{equation}
As presented in Appendix D, the winding number can change under a gauge transformation for
the occupied eigenstates, since the Wilson line is defined by these eigenstates,
in sharp contrast to the Floquet case, where the winding number is defined for an evolution operator~\cite{Fruchart2016PRB,WangZhong2017PRB}.
It suggests that a definite physical quantity is the change of the winding number as a system parameter
$\gamma$ varies. During this change, the Berry phase of the occupied bands should vary continuously, as discussed
in Appendix D. Fortunately, the quantized dipole moment is a $Z_2$ quantity, so that
the change of the winding number is sufficient to characterize the edge polarization.

The winding number changes as the
bulk energy gap and Wannier band gap close, but does not respond to the closure of the edge energy gap,
which is reasonable as the Wannier bands are constructed from the wave functions without any edge.
However, the closure of the edge energy gap is associated with the change
of the quadrupole moment. Taking into account the edge energy gap closure,
we define
\begin{eqnarray}
p_x^{\mathrm{edge}}(\gamma_1)-p_x^{\mathrm{edge}}(\gamma_0)
=&&e[(W_{\nu_x}^{\epsilon=\pi}(\gamma_1)-W_{\nu_x}^{\epsilon=\pi}(\gamma_0) \nonumber \\
&&-\Delta N_{q,x})/2]\mathrm{mod}(1),
\label{winding}
\end{eqnarray}
where $\Delta N_{q,x}$ represents the number of times that the quadrupole moment changes due to the gap closure of the edge energy spectrum
at the boundaries perpendicular to $y$, when we vary $\gamma$ from $\gamma_0$ to $\gamma_1$.
This shows that the topology of the bulk spectrum dictates the edge polarization as
the right sides are determined by the bulk property.
In fact, $\Delta N_{q,x}$ is also associated with the number of times of the change of a parity (eigenvalue of $\hat{m}_x$)
at the high-symmetric points $k_x=0$ or $k_x=\pi$ for a state localized at one boundary perpendicular to $y$ (see
more detailed discussion in the following section).
Provided that we start from a topologically
trivial phase, i.e., $p_x^{\mathrm{edge}}(\gamma_0)=0$, the formula can be reduced to
\begin{equation}
p_x^{\mathrm{edge}}(\gamma_1)
=e \left[(W_{\nu_x}^{\epsilon=\pi}(\gamma_1)
-\Delta N_{q,x})/2\right]\mathrm{mod}(1)
\label{windingR}
\end{equation}
by choosing a gauge such that $W_{\nu_x}^{\epsilon=\pi}(\gamma_0)=0$.

In Fig.~\ref{fig3}, we plot the edge polarizations $p_{x,y}^{\textrm{edge}}$ as a function of $\gamma$, calculated based on
the formula~(\ref{windingR}) by choosing a gauge such that $W_{\nu_x}^{\epsilon=\pi}(\gamma_0=-1)=W_{\nu_y}^{\epsilon=\pi}(\gamma_0=-1)=0$ given that the phase is
topologically trivial when $\gamma=-1$. The results are consistent with the Wannier spectrum in a cylinder
geometry and the edge polarization calculated using the hybrid Wannier functions.

\section{Inequivalence between Wannier and edge energy spectra due to long-range hopping}
\label{sec4}
\begin{figure}[t]
  \includegraphics[width=3.3in]{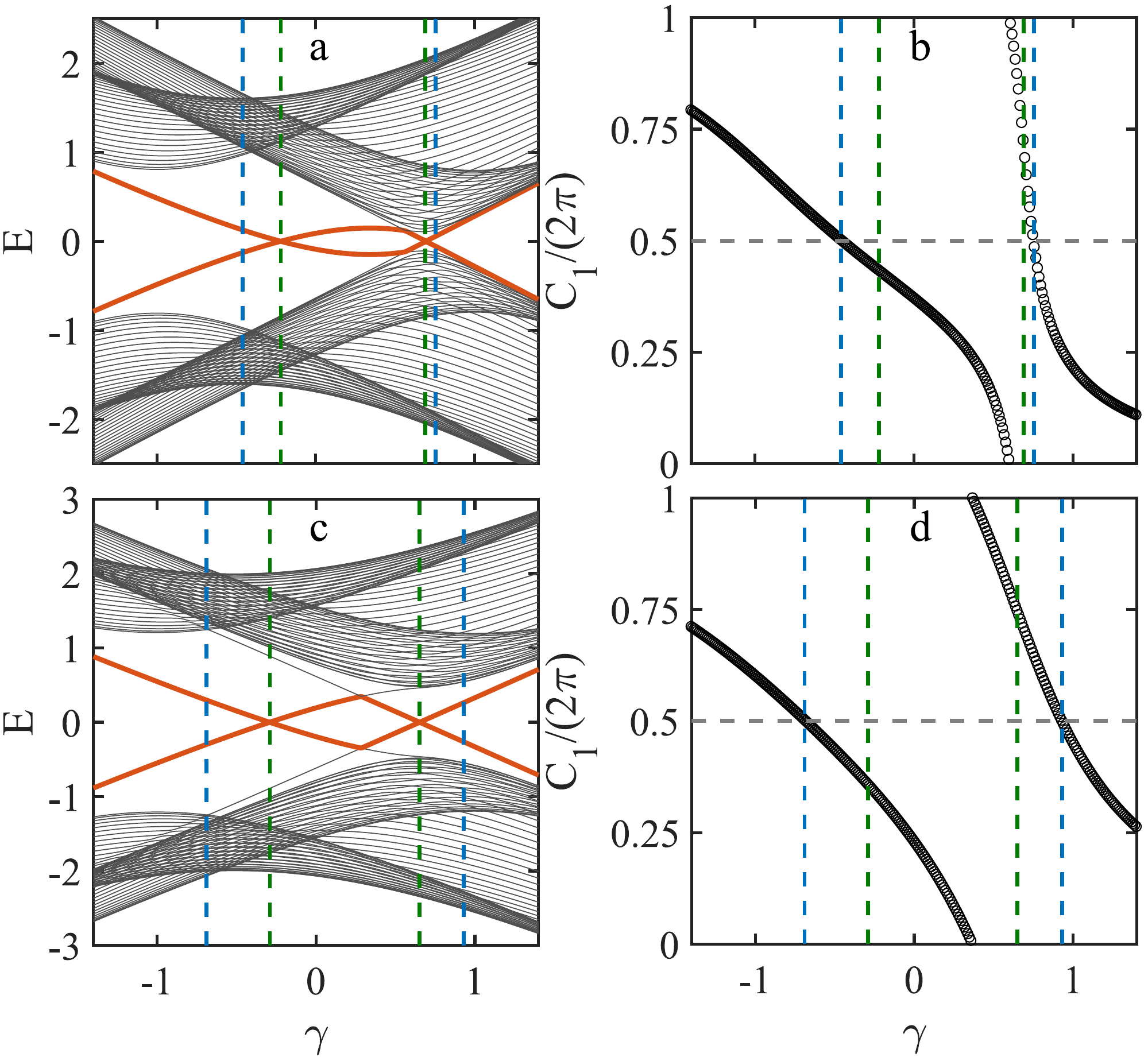}
  \caption{a, c, The energy spectra of the Hamiltonian $H_{\textrm{NNN}}$.
  b, d, The Berry phase of the occupied band for the corresponding Hamiltonian. In a and b,
  $b_2=0.8$ and in c and d, $b_2=1.2$. The green lines show the value of $\gamma$
  where zero-energy edge modes arise, and the blue lines show the value of $\gamma$ where
  the Berry phase is equal to $\pi$. These lines do not coincide, implying that
  the correspondence between $C_1=\pi$ and the existence of zero-energy edge states breaks down.
   }
\label{fig4}
\end{figure}

\subsection{A General Analysis}
As we have already discussed that the type-II QTI arises from the fact that the Wannier band and edge energy gaps
do not vanish simultaneously. In this subsection, we will demonstrate from a general perspective that the breakdown
occurs when the next-nearest-neighbor intercell hopping is appropriately included.

We now consider a generic four band Hamiltonian $H_g(k_x,k_y)$ in momentum space to describe a QTI.
The Hamiltonian is required to respect two reflection symmetries $\hat{m}_x$ and $\hat{m}_y$ to maintain the vanishing of the bulk polarization as well as the
quantization of quadrupole moments, corner charges and edge polarizations. In addition, either the particle-hole or chiral symmetry
is enforced to ensure that the corner modes have zero energy.
We note that the breakdown of the correspondence between Wannier and edge spectra was also found in a
system with neither particle-hole nor chiral symmetry~\cite{Wieder2019}.
To have the gapped Wannier bands, we further require that $\hat{m}_x$ and $\hat{m}_y$ anticommute,
and
at each high-symmetric line, e.g., $k_x=k^*$ ($k_y=k^*$) with $k^*=0,\pi$,
the eigenvalues of $\hat{m}_x$ ($\hat{m}_y$) of two occupied bands should occur in pairs as $(1,-1)$~\cite{Taylor2017Science,Bernevig2014PRB}.
We also require that these parities for two unoccupied bands also occur in pairs as $(1,-1)$.
Due to the reflection symmetry, we can write the Hamiltonian at the high-symmetric lines $k_x=k^*$
(similarly for $k_y=k^*$)
as a direct sum of two submatrices in a basis $\beta=\beta_1\cup\beta_{-1}$ consisting of eigenvectors of $\hat{m}_x$,
\begin{equation}
[H_g(k_x=k^*,k_y)]_\beta= H_1(k_y)\oplus H_2(k_y),
\end{equation}
where $H_1(k_y)$ and $H_2(k_y)$ are the matrix representation of
$H_g$ with respect to the basis $\beta_1$ and $\beta_{-1}$, respectively.
Here $\beta_1$ and $\beta_{-1}$ are the bases of the eigenspaces of $\hat{m}_x$ with the corresponding eigenvalues
being $1$ and $-1$, respectively.
In this basis, the representation of the reflection
symmetry along $x$ can be chosen as $\hat{m}_x=\sigma_z\otimes\sigma_0$.
For convenience, we use $\sigma$ instead of $\tau$ and explicitly write out the tensor product
operation $\otimes$ throughout this section.
To enforce the anti-commutation relation for $\hat{m}_x$ and $\hat{m}_y$, i.e., $\{\hat{m}_x,\hat{m}_y\}=0$,
we require $\hat{m}_y=\sigma_\mu \otimes \sigma_\lambda$ with $\mu=1,2$ and
$\lambda=0,1,2,3$. Now we have $H_2(k_y)=\sigma_\lambda H_1(-k_y)\sigma_\lambda$. Clearly,
the Berry phases of the occupied band of $H_1$ and $H_2$ take the opposite values.
These Berry phases also determine the Wannier spectra $\nu_y(k_x)$ at $k_x=k^*$.
In addition, the particle-hole or chiral symmetry forbids the existence of the terms
$\sigma_0$ and $\cos k_y\sigma_0$ in $H_1$ but allows the existence of the term $\sin k_y\sigma_0$ in it.

\begin{figure*}[t]
\includegraphics[width=6.5in]{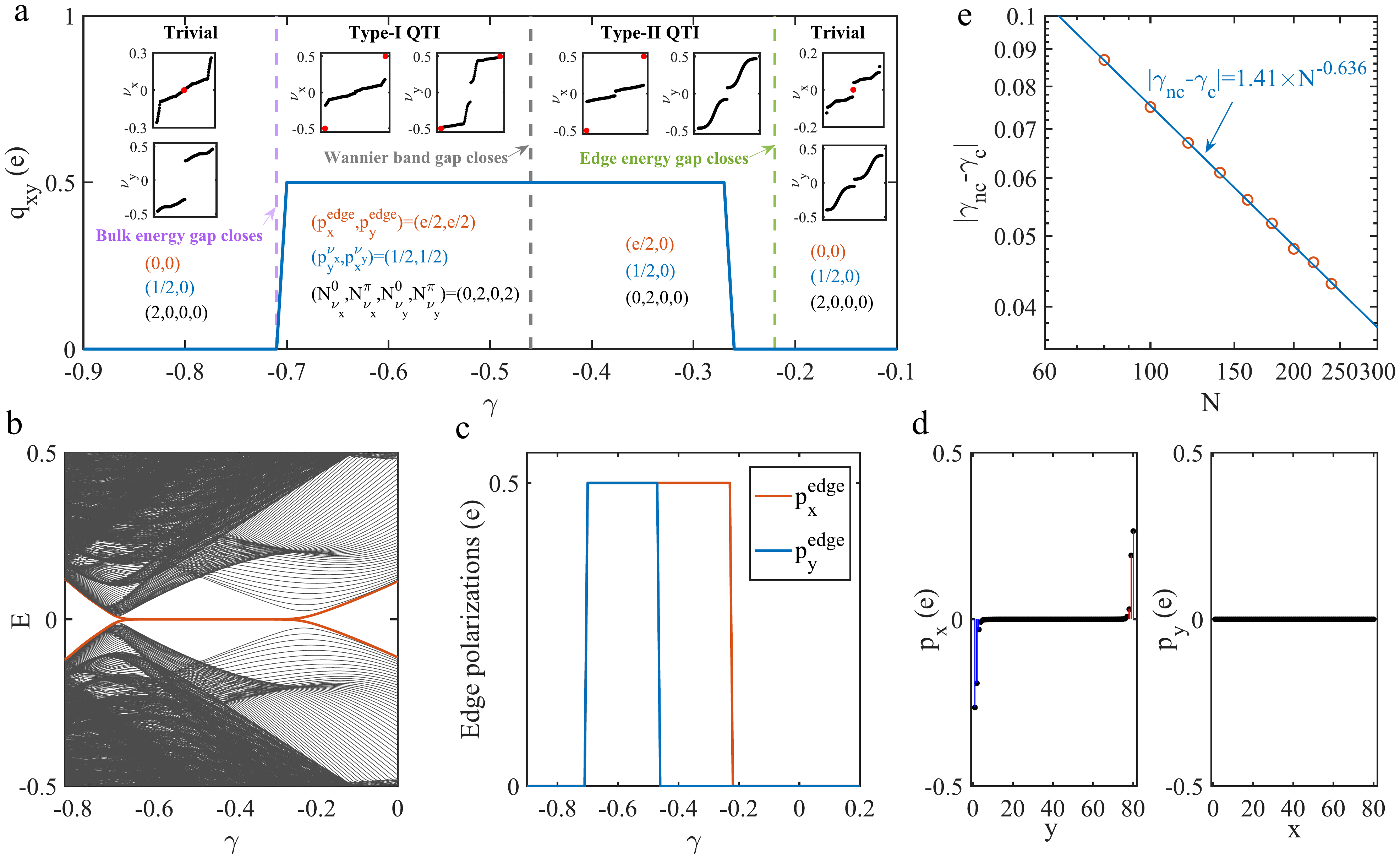}
  \caption{a, Phase diagram for the Hamiltonian $H_{II}$ with respect to
  a system parameter $\gamma$, where the quadrupole moment is plotted as a blue line.
  In the phase diagram, we observe the
  topologically trivial insulator, the type-I QTI and the type-II QTI.
  The subsets display the same quantities as in Fig.~\ref{fig2}.
  b, The energy spectrum as a function of $\gamma$ for open boundary conditions along all directions with zero-energy corner modes being highlighted by a red line.
  c, The edge polarizations calculated based on the formula (\ref{windingR}) by choosing a gauge such that $W_{\nu_x}^{\epsilon=\pi}(\gamma_0=-0.9)=W_{\nu_y}^{\epsilon=\pi}(\gamma_0=-0.9)=0$ given that the phase is topologically
  trivial when $\gamma=-0.9$.
  d, The spatial profiles of the edge polarizations in a type-II QTI with $\gamma=-0.35$.
  e, $|\gamma_{nc}-\gamma_c|$ with respect to the system size $N$ in the logarithmic scale,
  showing $|\gamma_{nc}-\gamma_c|\propto N^{-\delta}$ with $\delta=0.636$. $\gamma_{c}$ and
  $\gamma_{nc}$ denote the point where the edge energy gap closes and the phase transition point determined
  by the quandruple moment for a finite system, respectively. Here $b_2=0.8$. See Appendix E for the phase diagram for $b_2=1.2$.  }
\label{fig5}
\end{figure*}

We now discuss how the parity of an edge state localized at a boundary changes.
Suppose that at a critical point, $H_1$ has zero-energy modes under open boundary conditions but
its energy spectra under periodic boundary conditions remain gapped, implying the existence of
the states localized at two boundaries. Away from this point where the gap is opened, only one
edge state of $H_1$ is occupied, say, without loss of generality, the state localized at the top boundary.
Then, due to the reflection symmetry along $y$, the occupied edge state of $H_2$ should be localized
at the bottom boundary. Evidently, the parity of the state localized at the top (bottom) edge
is $1$ ($-1$). When the edge energy gap for $H_1$ and $H_2$ closes and then reopens,
if the position of the occupied edge state of $H_1$ changes from the top to the bottom,
then the position of the occupied edge state of $H_2$ changes from the bottom to the
top. As a result, the parity of the state localized at the top (bottom) edge changes from $1$ ($-1$)
to $-1$ ($1$), leading to the change of the edge polarization $p_x^{\textrm{edge}}$
and the quadrupole moment.

We are now in a position to ask two questions:
\begin{enumerate}
	\item \textit{If $C_1=\pi$, does $H_1$ always exhibit zero-energy edge modes in a geometry with open boundaries?}
	\item \textit{Conversely, if $H_1$ has zero-energy edge modes under open boundary conditions, does $C_1=\pi$ always hold?}
\end{enumerate}
Here $C_1$ (modulo $2\pi$) denotes the Berry phase of the occupied band of $H_1$.

If the answers to these questions are both yes, then the Wannier gap closing of $H_g$ at $\nu_y=\pm 1/2$
is always associated with the $y$-normal boundary spectra gap closing. For instance,
provided that $H_1$ has either the
particle-hole or chiral symmetry, it is a well-known fact that $H_1$ exhibits zero-energy edge modes under open
boundary conditions if $C_1=\pi$.

To show that both of the two questions are answered yes for the BBH model~\cite{Taylor2017Science}, let us write the model
in the following form
\begin{equation}
H_{\textrm{BBH}}=\sigma_0\otimes H_{\textrm{SSH}}(k_y,t_y)+H_{\textrm{SSH}}(k_x,t_x)\otimes \sigma_3,
\end{equation}
where $H_{\textrm{SSH}}(k,t)=(t+\cos k)\sigma_1+\sin k\sigma_2$ is the SSH Hamiltonian. This model respects the
reflection symmetries
$\hat{m}_x=\sigma_1 \otimes\sigma_0$ and $\hat{m}_y=\sigma_3\otimes\sigma_1$, the chiral symmetry $\Pi=\sigma_3\otimes\sigma_3$,
the particle-hole symmetry $\Xi=\sigma_3\otimes\sigma_3 \kappa$ and the time-reversal symmetry $\Theta=\kappa$.
At the high-symmetric momenta for $k_x$ (similarly for $k_y$), the Hamiltonian can be written in a basis $\beta$
consisting
of eigenvectors of $\hat{m}_x$ as
\begin{equation}
[H_{\textrm{BBH}}(k_x=k^*,k_y)]_\beta= H_1\oplus H_2.
\end{equation}
Specifically, we choose $\beta=\{|\uparrow_1\uparrow_3\rangle,|\uparrow_1\downarrow_3\rangle,|\downarrow_1\uparrow_3\rangle,|\downarrow_1\downarrow_3\rangle\}$
as the basis, where $|\uparrow_\lambda\rangle$ and $|\downarrow_\lambda\rangle$ are eigenvectors of
$\sigma_\lambda$ with $\lambda=1,2,3$.
In this basis, $H_1=H_{\textrm{SSH}}(k_y,t_y)+(t_x+\cos k^*)\sigma_3$ and $H_2=H_{\textrm{SSH}}(k_y,t_y)-(t_x+\cos k^*)\sigma_3$.
Here $H_1$ has zero-energy modes localized at boundaries in a geometry with open boundaries if and only if $C_1=\pi$ for $H_1$. This occurs only when $t_x=-\cos k^*$ so that $H_1=H_{\textrm{SSH}}$
that preserves the time-reversal symmetry, the particle-hole symmetry and the chiral symmetry. In other words, the Wannier band $\nu_y(k_x)$ and the $y$-normal edge energy spectra
close their gaps simultaneously at $k_x=\pi$ ($k_x=0$) when $t_x=1$ ($t_x=-1$) and $|t_y|<1$. In addition, if we
add a term $c \sin k_y\sigma_1\otimes\sigma_0$ with $c$ being a real number, which preserves two reflection symmetries
and the chiral symmetry, it results in a new term $c \sin k_y\sigma_0$ in $H_1$. This term still respects the particle-hole
symmetry and the $PT=\sigma_1\kappa$ symmetry so that the correspondence between the Berry phase of $\pi$ and the
existence of zero-energy edge modes persists~\cite{StructuredYong,ReviewYong}.

However, remarkably, we find that the correspondence does not necessarily hold when the next-nearest-neighbor intercell
hopping is involved. For instance, consider the following model,
\begin{eqnarray}
H_{\textrm{NNN}}(k_y)&=&\left(\gamma+\cos k_y\right)\sigma_3+\left( \sin k_y+b_2\sin 2k_y \right)\sigma_1 \nonumber \\
&&+b_2\cos 2k_y\sigma_2,
\label{HNN}
\end{eqnarray}
where $\gamma$ and $b_2$ are real parameters with $b_2$ denoting the strength of the next-nearest-neighbor
intercell hopping. When $b_2=0$, the Hamiltonian can be obtained by applying a unitary
transformation to $H_{\textrm{SSH}}$, and hence zero-energy edge modes arise if and only if the
Berry phase is equal to $\pi$. However, when $b_2\neq 0$, we find that both of these two questions can
be answered no. Fig.~\ref{fig4} illustrates that when $C_1=\pi$, the energy spectra of the boundary states
remain gapped, and when the edge energy gap vanishes, $C_1\neq \pi$.
This suggests that the correspondence between the Wannier band and edge energy
gaps may break down when the long-range hopping is involved in a QTI,
leading to the emergence of the type-II QTI.

\subsection{Simplified models for the type-II QTI}

By enforcing the reflection symmetries $\hat{m}_x=\sigma_1\otimes\sigma_3$, $\hat{m}_y=\sigma_1\otimes\sigma_1$
and the particle-hole symmetry $\Xi=\sigma_3\otimes\sigma_0\kappa$, we can construct a quadrupole topological
model based on $H_{\textrm{NNN}}$ and a transformed SSH model as
\begin{eqnarray}
H_{II}=H_L
+\sigma_2\otimes H_{\textrm{SSH}}^\prime(k_x),
\label{SimHam}
\end{eqnarray}
where
\begin{eqnarray}
H_L=&&\left(\gamma+\cos k_y\right)\sigma_1\otimes\sigma_0+\left( \sin k_y+b_2\sin 2k_y\right)\sigma_3\otimes\sigma_1 \nonumber \\
&&+b_2\cos 2k_y\sigma_3\otimes\sigma_2,
\end{eqnarray}
and $H_{\textrm{SSH}}^\prime(k_x)=(\gamma+t_x-g_0+t_x\cos k_x)\sigma_2+t_x\sin k_x \sigma_3$ can be obtained
by applying a unitary transformation to $H_{\textrm{SSH}}(k_x)$.
Here we set $t_x=1$.
The Hamiltonian~(\ref{SimHam}) has exactly the same reflection symmetries
and the particle-hole symmetry as the Hamiltonian~(\ref{Ham1}).
Yet, compared to the latter, this Hamiltonian is significantly simplified.

When $k_x=\pi$, the Hamiltonian can be written as
\begin{equation}
[H_{II}]_\beta=\left(
    \begin{array}{cc}
      H_{\textrm{NNN}}(k_y)+t_0\sigma_1 & 0_{2\times2} \\
      0_{2\times2} & \sigma_3 H_{\textrm{NNN}}(-k_y)\sigma_3-t_0\sigma_1 \\
    \end{array}
  \right)
\end{equation}
in a basis consisting of
eigenvectors of $\hat{m}_x$, i.e., $\beta=\{|\uparrow_x\uparrow_z\rangle, |\downarrow_x\downarrow_z\rangle,
|\uparrow_x\downarrow_z\rangle, |\downarrow_x\uparrow_z\rangle\}$.
Here $t_0=\gamma-g_0$.
Adding the term $t_0\sigma_1$ into $H_{\textrm{NNN}}(k_y)$ only shifts the
critical value of $\gamma$ where zero-energy edge modes appear.
For simplicity and clarity, we therefore set $g_0=\gamma_c$ so that $t_0=0$ when $\gamma=\gamma_c$ where
there exist zero-energy edge modes for $H_{\textrm{NNN}}$ under open boundary conditions. For example, when
$b_2=0.8$ ($b_2=1.2$), we set $g_0=-0.22$ ($g_0=-0.29$) as shown in Fig.~\ref{fig4}.
Now we can clearly see that
for this model while the $y$-normal edge energy gap vanishes when $\gamma=\gamma_c$, the Wannier band gap
does not close at $\nu_y=\pm 1/2$.

To see that the type-II QTI indeed emerges in this simplified model, we numerically compute the quadrupole moment,
the energy spectra under open boundary conditions along all directions and
the edge polarizations, displaying them
in Fig.~\ref{fig5}. Evidently, when $\gamma$ decreases across $\gamma_c=-0.22$,
the zero-energy corner modes appear due to the $y$-normal edge energy gap closing,
showing that the system enters into a higher-order topological insulating phase.
Meanwhile, the quadrupole moment suddenly jumps to $0.5e$ at $\gamma_{nc}=-0.26$ for a $N\times N$
system with $N=240$ denoting the number of unit cells along
$x$ (or $y$). The discrepancy from
the edge energy gap closing point $\gamma_c=-0.22$ is attributed to the finite size effects
while calculating the quadrupole moment in real space.
Specifically, in Fig.~\ref{fig5}(e), we display the finite size scaling between $|\gamma_{nc}-\gamma_c|$ and
the system size $N$, where $\gamma_c$ and $\gamma_{nc}$ denote the point where the edge energy gap closes and the phase transition point determined by the quadrupole moment for a finite system, respectively.
The relation exhibits a power law decaying, i.e., $|\gamma_{nc}-\gamma_c|\propto N^{-\delta}$ with
$\delta=0.636$, implying that $\gamma_{nc}=\gamma_c$ in the thermodynamic limit.
This tells us that the higher-order topological insulator is a
quadrupole topological insulating phase. Such an edge energy gap closing also leads to the emergence of the edge polarizations
$p_x^{\textrm{edge}}$ at the boundaries vertical to $y$ [see the inset in Fig.~\ref{fig5}(a) and Fig.~\ref{fig5}(d)]. But the edge polarizations $p_y^{\textrm{edge}}$ at the $x$-normal boundaries remain zero
because the Wannier bands $\nu_y$ remain gapped at $\nu_y=\pm 1/2$. As a consequence, the type-II quadrupole
insulating phase arises. Compared to the type-II phase introduced in Sec.~\ref{sec2},
it is normal in the sense that the Wannier spectra $\nu_x$ have edge states only at $\nu_x=\pm 1/2$.
We also calculate the edge polarizations based on the formula (\ref{windingR})
by choosing a gauge such that $W_{\nu_x}^{\epsilon=\pi}(\gamma_0=-0.8)=W_{\nu_y}^{\epsilon=\pi}(\gamma_0=-0.8)=0$, showing that
the winding number can correctly characterize the edge polarization change due to the Wannier band closing
or the bulk energy gap closing. For the edge polarization $p_x^{\textrm{edge}}$, it changes at $\gamma_c=-0.22$
because of the edge energy gap closing associated with a one time change of a parity (eigenvalue of $\hat{m}_x$) at a $y$-normal boundary and the quadrupole moment.

\begin{figure*}[t]
  \includegraphics[width=7in]{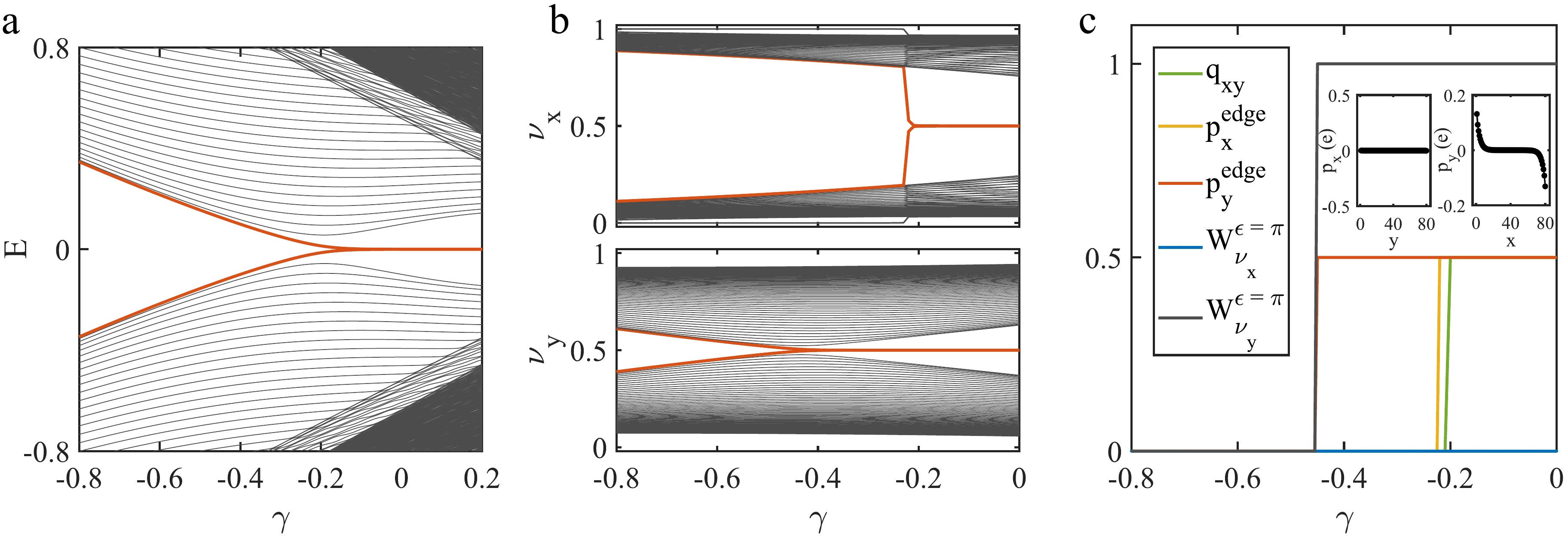}
  \caption{a, The energy spectrum versus $\gamma$ for the Hamiltonian $H_{IV}$ with open boundaries in all directions.
  b, The Wannier band $\nu_x$ and $\nu_y$ versus $\gamma$ under open boundary conditions along $y$ and $x$, respectively.
  The red lines depict the edge states in the Wannier spectrum.
  c, The quadrupole moment $q_{xy}$,
  the edge polarizations $p_{\mu}^{\textrm{edge}}$ ($\mu=x,y$) computed using the formula (\ref{windingR}) and the winding
  numbers $W_{\nu_\mu}^{\epsilon=\pi}$ ($\mu=x,y$) computed using the formula (\ref{WindFormula}) versus $\gamma$.
  When calculating the edge polarizations and the winding number, we choose a gauge
  such that $W_{\nu_x}^{\epsilon=\pi}(\gamma_0=-0.8)=W_{\nu_y}^{\epsilon=\pi}(\gamma_0=-0.8)=0$
  because the phase is topologically trivial when $\gamma=-0.8$. The figure demonstrates the existence of a new
  topological phase with nonzero $p_y^{\textrm{edge}}$ but without quadrupole moments and zero-energy corner modes
  when $-0.45<\gamma<-0.22$ (see the inset for the spatial distribution of the edge polarizations). Note that $p_x^{\textrm{edge}}$ jumps to $0.5$ at $\gamma=-0.22$
  due to the $y$-normal edge gap closing.
  $q_{xy}$ and $p_x^{\textrm{edge}}$ are hidden behind the blue line when $\gamma<-0.22$. The small discrepancy between
  the quadrupole moment transition point ($N=80$ for evaluation of the quadrupole moment) and the edge gap closing point is due to the finite size effects. Here $b_2=1.2$.
  }
\label{fig6}
\end{figure*}

Similarly, we can construct another model
\begin{equation}
H_{III}=H_L-\sigma_3\otimes H_{\textrm{SSH}}^\prime(k_x),
\label{HamSim2}
\end{equation}
which respects the reflection symmetries $\hat{m}_x$ and $\hat{m}_y$, the chiral symmetry $\Pi_2$, i.e., $\Pi_2 H_{III} \Pi_2^{-1}=-H_{III}$ with $\Pi_2=\sigma_2\otimes\sigma_0$
and the particle-hole symmetry $\Xi$.
We also find the type-II quadrupole topological insulating phase in this model (see
the Appendix E
for the phase diagram). If we add a term $\sin k_x \sigma_1\otimes\sigma_1$ which breaks the particle-hole symmetry
but preserves the chiral symmetry, the type-II QTI also arises (see the Appendix E for the phase diagram).
The above results demonstrate that the type-II QTI generically exists in systems with reflection symmetries
as well as the particle-hole or chiral symmetry in the presence of long-range hopping.

\subsection{A new topological phase with nonzero quantized edge polarizations but without zero-energy corner modes}

\begin{figure*}[t]
\includegraphics[width=4.8in]{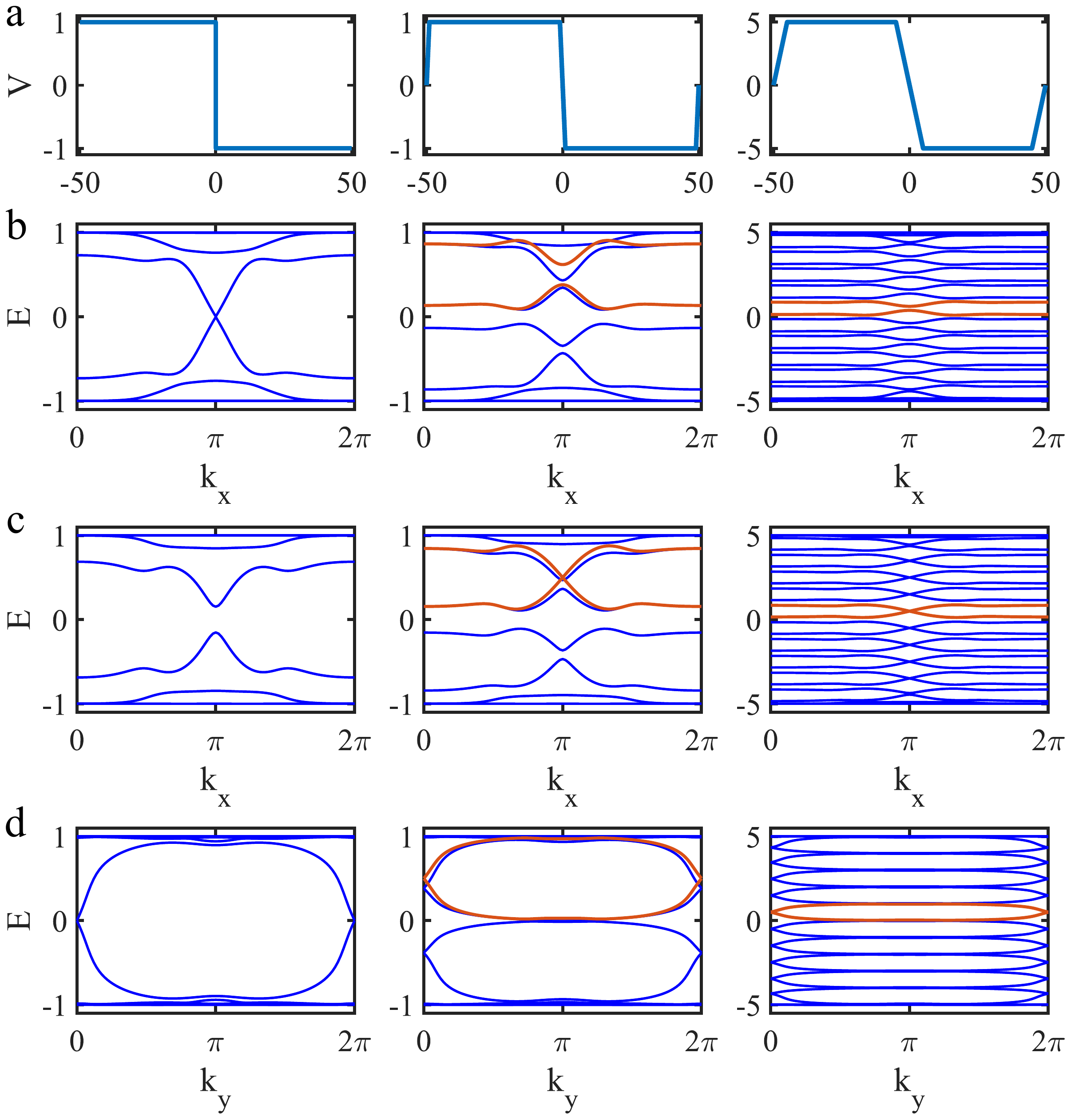}
  \caption{The connection between the edge energy spectrum and Wannier spectrum.
  a, Visualization of the edge potential for $V_0$ in the left column as well as $V_L$ with $M=1$ and $M=5$
  in the middle and right columns, respectively.
  b, The energy bands of $H_e^0$ ($H_e^L$) with $V_0(y)$ ($V_L(y)$) versus $k_x$ in the left column (in the middle and right columns) for $\gamma=0.34$. The edge energy spectrum encoded in $H_e^0$ has a vanishing gap, while the Wannier spectrum
  encoded in $H_e^L$ for a moderate or large $M$ (e.g., $M=5$) has a finite gap.
  c, The same spectrum as b for $\gamma=0.12$. The edge energy spectrum has a finite gap, which gradually vanishes
  as $M$ is increased for $V_L(y)$, leading to a gapless Wannier spectrum.
  d, The energy spectrum of $H_e^0$ ($H_e^L$) with $V_0(x)$ ($V_L(x)$) versus $k_y$ in the left column (in the middle and
  right columns). In this scenario, both the edge energy spectrum and Wannier spectrum are gapless.
  In the middle and right columns for b-d, the Wannier spectra are plotted as the red lines in comparison with the energy
  spectrum of $H_e^L$, showing that the former agrees very well with the latter for $M=5$. For b-d, the middle and right columns correspond to the results for $M=1$ and $M=5$, respectively.
   }
\label{fig7}
\end{figure*}

We have demonstrated that the Wannier gap closing can drive the topological phase transition from the
type-II QTI to the type-I and vice versa, suggesting that the topological property change of the Wannier band
can lead to new topological phases. In contrast to the QTI, one may wonder whether
the Wannier band gap closing can result in a new topological phase with nonzero quantized edge polarizations
but without quadrupole moments and zero-energy corner modes.
In the following, we will show the existence of this phase by considering the Hamiltonian in momentum space
\begin{eqnarray}
&&H_{IV}=(\gamma+\cos k_y)\sigma_1\otimes\sigma_0+(\sin k_y+b_2\sin 2k_y)\sigma_2\otimes\sigma_1 \nonumber \\
&&+\sigma_2\otimes[(1+b_2\cos 2k_y+t_x\cos k_x)\sigma_2+t_x\sin k_x\sigma_3],
\label{H4}
\end{eqnarray}
which respects the reflection symmetries $\hat{m}_x$ and $\hat{m}_y$,
the time-reversal symmetry $\Theta$, the particle-hole symmetry $\Xi$ and the
chiral symmetry $\Pi$. Here we set $t_x=1$.

Figure~\ref{fig6} illustrates that the Wannier spectra $\nu_y$ close their gap at $\nu_y=\pm 1/2$ when $\gamma=-0.45$
and the Wannier spectra $\nu_x$ remain gapped.
This gap closure results in edge states in the Wannier bands
at $\nu_y=\pm 1/2$ under open boundary conditions, giving rise to the quantized edge polarization $p_y^{\textrm{edge}}=\pm e/2$
[their spatial distributions are shown in the inset of Fig.~\ref{fig6}(c)].
However, the energy spectra of the system with open boundaries in all directions remain gapped
until $\gamma=-0.22$. This means that a new topological phase arises when $-0.45<\gamma<-0.22$.
In this phase, despite the vanishing of the quadrupole moment and the absence of the zero-energy corner modes,
nonzero quantized edge polarizations exist. As $\gamma$ increases across $-0.22$, the quadrupole moment and
the corner modes appear, leading to the type-I QTI.
In addition, we
display the edge polarizations in Fig.~\ref{fig6}(c) evaluated based on the formula (\ref{windingR}) by
choosing a gauge such that
$W_{\nu_x}^{\epsilon=\pi}(\gamma_0=-0.8)=W_{\nu_y}^{\epsilon=\pi}(\gamma_0=-0.8)=0$,
implying that the winding number correctly characterizes the change of the edge polarizations
across the topological phase transition at $\gamma=-0.45$.

Interestingly, when we include a term $\sin k_x \sigma_1\otimes\sigma_1$ that breaks the time-reversal symmetry and the particle-hole symmetry but preserves the chiral symmetry, or a term $\sin k_x \sigma_3\otimes\sigma_3$ that breaks the time-reversal symmetry and the chiral symmetry but preserves the particle-hole symmetry, we can still observe the new topological phase transition
(see Appendix E for details).

The discovery of the type-II QTI and the new topological phase discussed above suggests that the gap closing of
the Wannier spectra can induce topological phase transitions without involving any energy gap closing.

\subsection{An alternative approach to show the inequivalence}

Alternatively, we may consider the problem based on the method introduced in Ref.~\cite{Klich2011PRL}. There,
it has been shown that the edge spectrum can be continuously deformed into the Wannier band, implying that
they share the same spectral flow, e.g., a chiral edge mode in a quantum Hall insulator corresponds to
a winding Wannier band. Based on the method, Ref.~\cite{Khalaf2019arXiv} generalizes the correspondence between
the Wannier and edge spectra in the quadrupole topological insulating phase by showing that the Wannier gap
and boundary spectra gap have to close simultaneously. This also indicates that the type-II QTI cannot occur.
However, Ref.~\cite{Khalaf2019arXiv} only studies a specific model with the nearest-neighbor intercell hopping.

In the following, we will show that these two gaps do not necessarily vanish at the same time in our model
Hamiltonian~(\ref{Ham1}) by following the method in Refs.~\cite{Klich2011PRL,Khalaf2019arXiv}.
Let us consider a gapped Hamiltonian $H({\bf k})$ in momentum space. We define the projection operator to the occupied bands
as $\tilde{P}({\bf k})=\sum_{n=1}^{N_{\textrm{occ}}} |u_{\bf k}^n\rangle \langle u_{\bf k}^n|$ with $|u_{\bf k}^n\rangle$ being
an occupied eigenstate of $H({\bf k})$. The representation of this operator in the real space is obtained by
Fourier transformation in the direction perpendicular to an edge,
\begin{equation}
P_{r\alpha,s\beta}({\bf k_{\parallel}})=\int_{BZ} \frac{dk_{\perp}}{2\pi} e^{ik_{\perp}(r-s)}\tilde{P}_{\alpha,\beta}({\bf k})
\end{equation}
where $r,s$ are coordinates of lattices in the edge-normal direction and $\alpha,\beta$ are orbital indices. ${\bf k_{\parallel}}$ and $k_{\perp}$ denote
the momentum parallel and perpendicular to the edge, respectively.

Following Refs.~\cite{Klich2011PRL,Khalaf2019arXiv}, we construct a Hamiltonian from the projection operator,
\begin{equation}
H_e^0=P V_0(x) P+(I-P),
\end{equation}
where
\begin{equation}
V_0(x)=\left\{
\begin{aligned}
1   \qquad \text{for} \ x < 0 \\
-1  \qquad \text{for} \ x \geq 0
\end{aligned}
\right.
\end{equation}
For a tight-binding model, $x$ is the coordinate of discrete lattice sites along the direction perpendicular to the edge
and takes the value of integer numbers from $-N_x/2$ to $(N_x/2-1)$ with $N_x$ being the size along the $x$ direction. Here, $V_0(x)$ introduces two boundaries at $x=0$ and $x=N_x/2$ so that for $-N_x/2 \leq x < 0$, the system is topologically trivial and, otherwise, topologically equivalent
to $H({\bf k})$.
Indeed, we have found that the eigenvalues $E_{V_0}(k_y)$ of $H_{e}^{0}(k_y)$ can be adiabatically mapped to the energy spectrum of the original Hamiltonian under open boundary conditions along $x$.

To see the connection between the edge energy spectrum and the Wannier band, we impose a linear edge so that the Hamiltonian
reads
\begin{equation}
H_e^L=P V_L(x) P + M(I-P),
\end{equation}
where the deformed linear edge potential (see Fig.~\ref{fig7}) is given by
\begin{equation}
V_L(x)=\left\{
\begin{aligned}
&x+N_x/2   &&\text{for} \ -N_x/2 \leq x < -N_x/2+M \\
&M         &&\text{for} \ -N_x/2+M \leq x < -M+1 \\
&-x        &&\text{for} \ -M+1 \leq x < M \\
&-M        &&\text{for} \ M \leq x < N_x/2-M+1 \\
&x-N_x/2   &&\text{for} \ N_x/2-M+1 \leq x < N_x/2
\end{aligned}
\right.
\end{equation}
Here, $M$ is used to control the size of the linear potential region.
When $M=1$, $V_L(x)$ is only slightly deformed from $V_0(x)$. Ref.~\cite{Klich2011PRL} argues that this slight change
should not significantly change the energy spectrum and thus
as one continuously deform the boundary by increasing $M$ for $V_L(x)$, the energy spectrum should
be smoothly deformed into the eigenvalues of $PxP$,
the restriction of which to the interval $[0,1]$ gives exactly the Wannier spectrum.
Ref.~\cite{Khalaf2019arXiv} generalizes this argument to the higher-order topological case by
assuming that this slight change should not open an energy gap if the boundary states using $V_0(x)$ are gapless,
and
concludes that the edge energy spectrum and the Wannier spectrum should close their gaps simultaneously.

However, we find that this generalized argument is not always true. For example, in our model, when $\gamma=0.34$,
the energy spectrum is gapless under the potential $V_0(y)$, which is consistent with our results under open boundary condition along $y$. But, as we slightly deform the boundaries, e.g., taking $M=1$ for $V_L(y)$,
we find that the energy gap opens, as shown in Fig.~\ref{fig7}(b). When we further
increase $M$ to $5$, the gap remains open and the energy spectrum is exactly the same as the Wannier spectrum.

We also show that for $\gamma=0.12$, the energy gap vanishes for $M=5$ while there is a nonzero energy gap
for the energy spectrum under the boundary geometry $V_0(y)$ [see Fig.~\ref{fig7}(c)]. We indeed also find the case where the
energy gap always remains gapless. For example,
for $\gamma=1.03$, the gap remains zero as we continuously deform the boundaries, as shown in Fig.~\ref{fig7}(d).

Therefore, we conclude that the edge energy bands and Wannier bands are not guaranteed to close their gaps
simultaneously and thus can be topologically inequivalent.

\section{Pumping phenomena and novel three-dimensional higher-order topological insulators}
\label{sec5}

\begin{figure}[t]
  \includegraphics[width=3.3in]{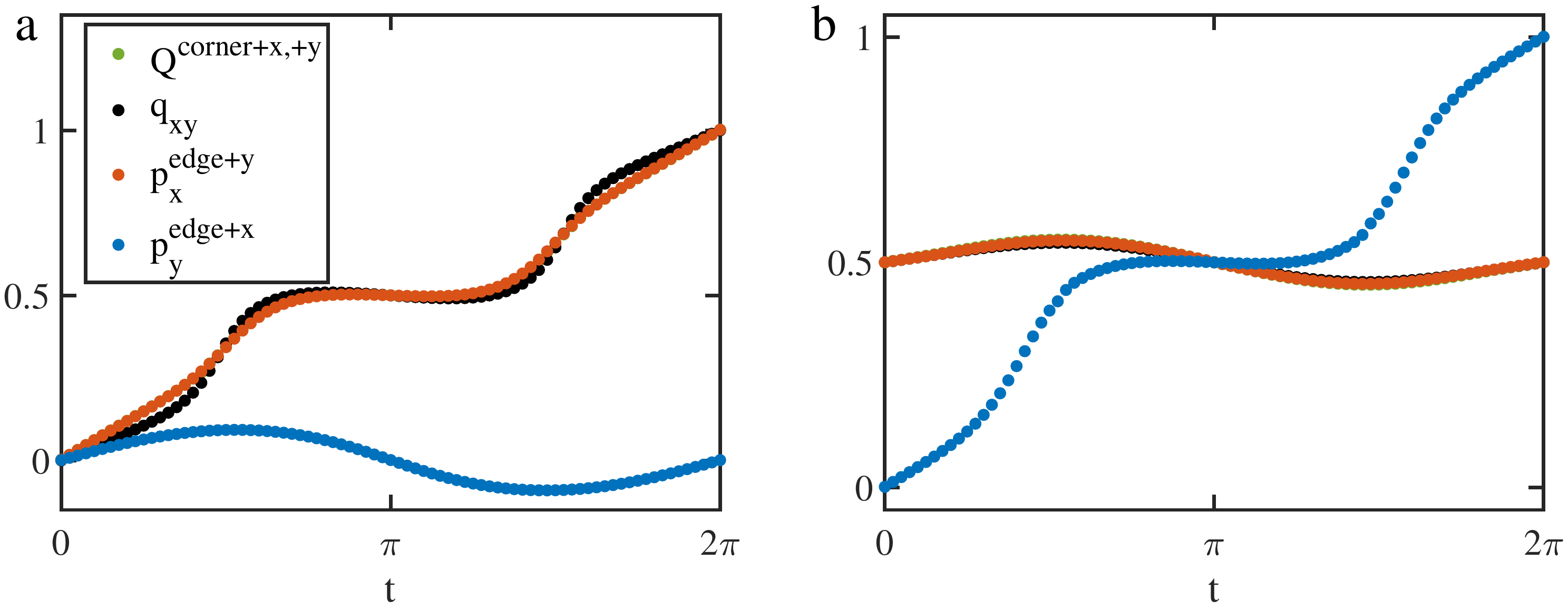}
  \caption{The transport of the quadrupole moment $q_{xy}$,
   corner charge $Q^{\mathrm{corner~+x,+y}}$ and edge polarization $p_{x,y}^{\mathrm{edge~}\alpha }$ over a full cycle.
   a, The cycle refers to the evolution of a system
   from a topologically trivial phase to the type-II AQTI and finally return.
   Note that the green line is hidden behind the red one.
   b, The cycle refers to the evolution of a system
   from the type-II phase to the type-I AQTI and then return.
   Note that the green and black lines are hidden behind the red one.
   The units of all the quantities are $e$.
   }
\label{fig8}
\end{figure}

We now discuss the pumping phenomena as a system parameter is slowly tuned. We expect
the existence of anisotropic edge currents during a pumping process. To induce the change of the
edge polarization, we need to break the reflection symmetry by
adding a $\delta \tau_z\sigma_0$ term so that the edge polarization is not locked to be quantized.
However, to maintain the vanishing of the bulk polarization, we still preserve the inversion symmetry.
We also maintain the bulk and edge energy gaps during the entire cycle for adiabaticity.

Specifically, we consider the pumping process across the topologically trivial phase and type-II AQTI in
Fig.~\ref{fig2}.
To achieve the pumping, we choose $\delta=0.1\sin(t)$ and $\gamma=0.35+0.1\cos(t)$ in the Hamiltonian~(\ref{Ham1}). At $t=0,2\pi$, the system is
in the topologically trivial phase without any edge polarization, quadrupole moments and corner charges, while
at $t=\pi$, it is in the type-II AQTI with $q_{xy}=|Q^{\mathrm{corner}}|=|p_x^{\mathrm{edge}}|=1/2$ and
$p_y^{\mathrm{edge}}=0$. As time progresses with changes of $\delta$ and $\gamma$, the system evolves from
the topologically trivial phase to the type-II quadrupole insulating phase and then return to the original trivial phase.
At each time, we evaluate the edge polarization, corner charge and quadrupole moment.
We find that during an entire cycle, the edge polarization at the top boundary increases by one
and at the bottom boundary it decreases by one, as shown in Fig.~\ref{fig8}(a). This also happens for corner charges
and the quadrupole moment.
However, the left and right edges do not exhibit a net transport for the polarization.

\begin{figure*}[t]
\includegraphics[width=\textwidth]{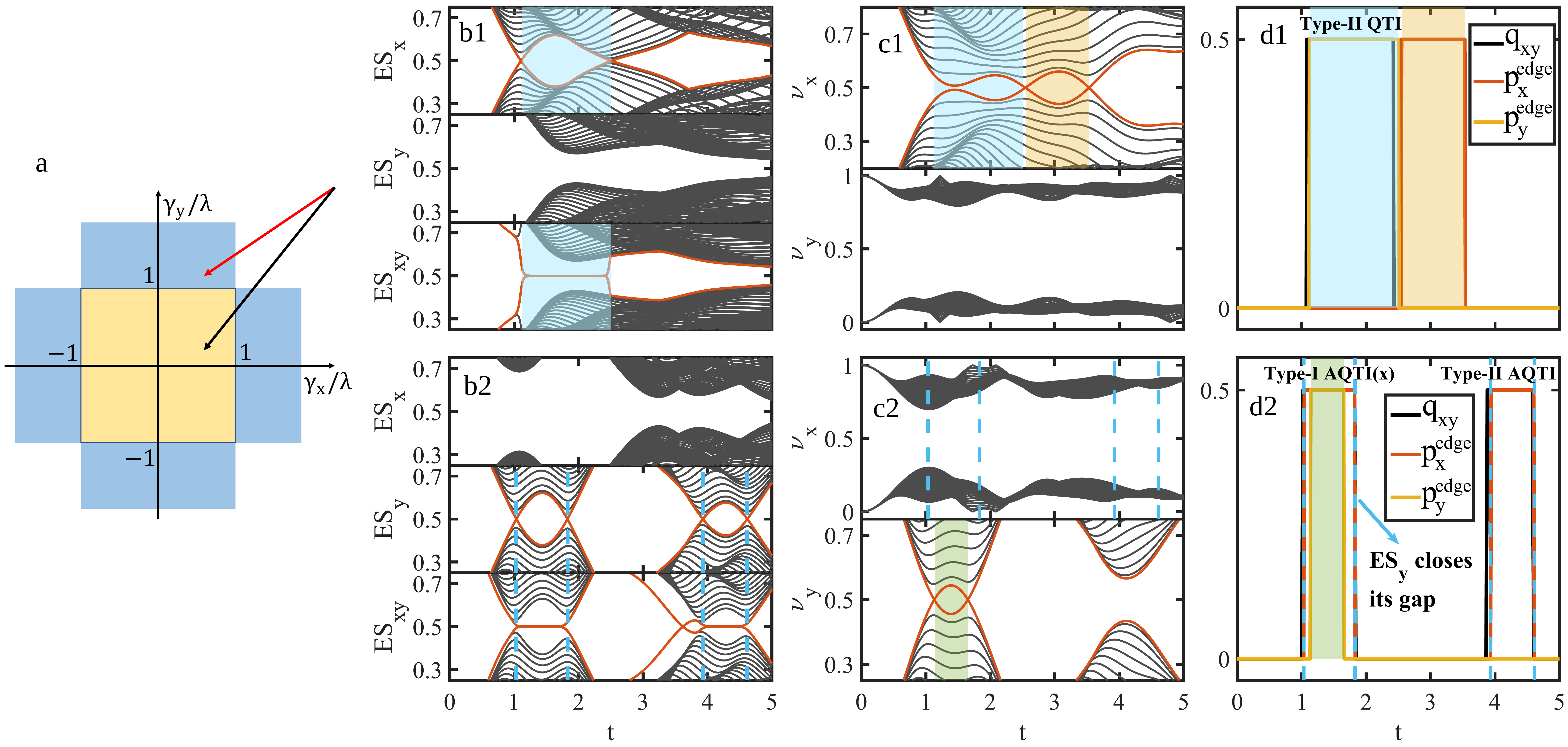}
  \caption{a, Phase diagram of the type-I Hamiltonian, where the yellow region
  corresponds to the type-I QTI and other regions to topologically trivial insulating phases.
  Note that while the Berry phase of the Wannier bands
along one direction in the blue regions is nonzero,
there are no edge polarizations since the Wannier bands close their gaps at either $\nu_x=0$ or $\nu_y=0$.
The quench dynamics are performed by suddenly tuning system parameters from $\gamma_x=\gamma_y,\lambda=0$ to $\gamma_x=0.6, \gamma_y=1.2,\lambda=1$
  and to $\gamma_x=0.6,\gamma_y=0.2,\lambda=1$, respectively (see the red and black arrows). b1, Entanglement spectra,
  c1, Wannier spectra and d1, quadrupole moments and edge polarizations as time evolves, for the quench
  dynamics denoted by the red arrow. b2-d2 display the same physical quantities as b1-d1, but for the quench dynamics denoted by the black arrow. In b1 and b2, $\textrm{ES}_{x}$ ($\textrm{ES}_y$) refers to the entanglement spectrum in the left (bottom) subsystem and
  $\textrm{ES}_{xy}$ the nested entanglement spectrum (see Appendix G). The edge polarizations in d1 and d2 are calculated based on the formula~(\ref{windingR}) by choosing a gauge such that $W_{\nu_x}^{\epsilon=\pi}(t=0)=W_{\nu_y}^{\epsilon=\pi}(t=0)=0$ because the initial states are topologically trivial. In b1-d1, the light blue region displays the
  type-II QTI (specifically, in this region, when $t<1.22$, it corresponds to the type-II AQTI, and when $t>1.22$,
  it corresponds to the type-II normal QTI) and the light red region displays the new topological phase with only quantized $p_x^{\textrm{edge}}$. In b2-d2, the light green region shows the type-I AQTI(x) and the dashed blue lines label the time at which the gap of $\textrm{ES}_y$ vanishes. The region between the two dashed blue lines around $t=4$
  corresponds to the type-II AQTI phase. In d1, $q_{xy}=0$ when $t<1.1$ or $t>2.43$, and in d2, $q_{xy}=0$ when $t<1.04$ or $1.84<t<3.94$ or $t>4.6$. In d1, the slight discrepancy between
  $q_{xy}$ and $p_y^{\textrm{edge}}$ around $t=2.5$ is caused by finite size effects while calculating $q_{xy}$ in real space.
  }
\label{fig9}
\end{figure*}

We also consider the pumping process across the type-I AQTI and type-II AQTI described by $\delta=0.1\sin(t)$ and $\gamma=0.1+0.1\cos(t)$. At $t=0,2 \pi$, the system is in the type-II AQTI phase
while at $t=\pi$, the system in the type-I AQTI phase. As time evolves over a full cycle,
it turns out that only the left and right boundaries show a net change of the polarization
but not for the other
quantities such as the edge polarization at the top and bottom boundaries, corner charges and quadrupole moments,
as shown in Fig.~\ref{fig8}(b).
Both of the pumping phenomena contrast with previous works
where all boundaries, corner charges and quadrupole moments exhibit a net change during a cycle~\cite{Taylor2017Science,Taylor2017PRB,Wheeler2018arXiv,Cho2018arXiv}.
These novel pumping phenomena indicates the peculiar features of the type-II AQTI.

If we regard the adiabatic parameter $t$ as momentum $k_z$ in the third direction,
we obtain two novel three-dimensional higher-order topological insulators characterized
by a set of topological invariants consisting of the winding of the quadrupole moment and
edge polarizations along two directions: $(N_{q_{xy}},N_{p_x^{\mathrm{edge}}},N_{p_y^{\mathrm{edge}}})$, where
$N_{O}=\int_0^{2\pi} dk_z \frac{\partial O}{\partial k_z}$ with
$O=q_{xy},p_{x,y}^{\mathrm{edge}}$. The new insulating phases correspond to
$(N_{q_{xy}},N_{p_x^{\mathrm{edge}}},N_{p_y^{\mathrm{edge}}})=(1,1,0)$ and $(0,0,1)$, respectively, which are
fundamentally different from the previously found insulator with $(1,1,1)$.
Although the phase with the winding number being $(0,0,1)$ has a winding for the edge polarization, it does not lead
to the chiral hinge modes beyond the conventional wisdom that the presence of the winding of the polarization
corresponds to a Chern insulator with chiral edge modes. In fact, in our case, it is associated with the presence
of chiral edge modes in the Wannier bands, as presented in Appendix F.

\section{Quench dynamics}
\label{sec6}
Since these new topological phase transitions are driven by the Wannier gap closure,
they may arise from quench dynamics through unitary time evolution.
In this section, we will study the dynamics of states as the Hamiltonian is suddenly
changed from one phase to another.
Specifically, we consider the type-I quadrupole Hamiltonian~\cite{Taylor2017Science,Taylor2017PRB}
\begin{eqnarray}
H_{\textrm{BBH0}}({\bf k})=&&(\gamma_x+\lambda\cos k_x)\Gamma_4+\lambda\sin k_x\Gamma_3 \nonumber \\
&&+(\gamma_y+\lambda\cos k_y)\Gamma_2+\lambda\sin k_y\Gamma_1,
\label{typeImodel}
\end{eqnarray}
where $\Gamma_j=-\tau_2\sigma_j$ ($j=1,2,3$) and $\Gamma_4=\tau_1\sigma_0$.
This Hamiltonian respects the reflection symmetries $\hat{m}_x$ and $\hat{m}_y$,
the time-reversal symmetry $\Theta$, the particle-hole symmetry $\Xi$ and the chiral symmetry $\Pi$.
The phase diagram is shown in Fig.~\ref{fig9} with respect to $\gamma_x/\lambda$ and $\gamma_y/\lambda$ (see also Ref.~\cite{Taylor2017PRB}).
We choose the ground state $|\psi_i\rangle$ of $H_i/\gamma_x=\tau_1\sigma_0-\tau_2\sigma_2$ (i.e., $\gamma_x=\gamma_y$ and $\lambda=0$)
as the initial state and then suddenly tune $\gamma_x$, $\gamma_y$ and $\lambda$ to the values
as shown in Fig.~\ref{fig9}(a) so that the Hamiltonian changes to $H_f({\bf k})$.
The state then evolves under the final Hamiltonian,
i.e., $|\psi_{\bf k}(t)\rangle=e^{-iH_f({\bf k}) t}|\psi_i\rangle$.
Since the evolving state $|\psi_{\bf k}(t)\rangle$ is an eigenstate of a
parent Hamiltonian defined as $H_p({\bf k})=e^{-iH_f({\bf k}) t}H_i({\bf k})e^{iH_f({\bf k}) t}$, the topological
properties of the evolving states are determined by the parent Hamiltonian.
During time evolution, one can easily find that the parent Hamiltonian still preserves
the reflection symmetries, i.e., $\hat{m}_\mu H_p \hat{m}_\mu^\dagger=H_p(k_\mu\rightarrow -k_\mu)$ with $\mu=x,y$,
and the particle-hole symmetry, i.e., $\Xi H_p \Xi^{-1}=-H_p(-{\bf k})$, but usually breaks the time-reversal symmetry
and the chiral symmetry. Without loss of generality, we
study two scenarios: one corresponds
to the final Hamiltonian in the topologically trivial region and the other in the quadrupole insulating
region.

Figure~\ref{fig9}(b1-d1) illustrate the entanglement spectrum, Wannier spectrum and edge polarizations
as a function of time, after the Hamiltonian is quenched into a topologically trivial phase.
At $t=1.1$, the gap of the entanglement spectrum $\textrm{ES}_x$ vanishes, revealing the vanishing of the energy gap
for the parent Hamiltonian under open boundary conditions along $x$~\cite{Fidkowski2010PRL}. This gap closure leads to
the appearance of the quantized quadrupole moment ($q_{xy}=e/2$) and edge polarizations along $y$ ($p_y^{\textrm{edge}}=\pm e/2$).
Here, the edge polarizations are calculated by the
formula~(\ref{windingR}), where $\Delta N_{q\mu}$ ($\mu=x,y$) are evaluated by the number of times
that the gap of $\textrm{ES}_{\bar{\mu}}$ ($\bar{\mu}=y,x$ if $\mu=x,y$) closes.
Here we choose a gauge such that $W_{\nu_x}^{\epsilon=\pi}(t=0)=W_{\nu_y}^{\epsilon=\pi}(t=0)=0$ for calculation in the sense
that the initial states are topologically trivial.
Since there is neither gap closure for $\textrm{ES}_y$ nor Wannier gap
closure for $\nu_x$,
$p_x^{\textrm{edge}}$ remains zero. In addition, the nested entanglement
spectra $\textrm{ES}_{xy}$ exhibit zero-energy modes (see the blue region) (the nested entanglement spectrum
$\textrm{ES}_{xy}=0.5$ corresponds to the entanglement zero mode, see Appendix G), reflecting the existence of fractional corner charges
($Q^{\textrm{corner }}=\pm e/2$)
for the parent Hamiltonian in a geometry with open boundary conditions. This shows that
the system enters into a type-II quadrupole topological insulating region with the basic relation
$(p_y^{\textrm{edge }}+p_x^{\textrm{edge }}-q_{xy}-Q^{\textrm{corner }})\text{mod}(1)=0$ being violated.
At $t=2.5$, $\textrm{ES}_x$ experiences a gap closure, leading to a topologically trivial phase with zero quadrupole moment ($q_{xy}=0$) and edge polarizations
($p_y^{\textrm{edge}}=p_x^{\textrm{edge}}=0$).

Remarkably, shortly afterwards, the gap of the Wannier bands $\nu_x$ vanishes at $t=2.55$ and $\nu_x=0.5$, resulting in nonzero
quantized edge polarizations $p_x^{\textrm{edge}}=e/2$. However, in this phase,
the entanglement spectra do not exhibit any gap closure, reflecting the absence of the edge
energy gap closure for the parent Hamiltonian under open boundary conditions. This accounts for
the absence of the quadrupole moment and zero-energy modes in the nested entanglement spectrum.
When the gap of these Wannier bands $\nu_x$ closes again at $t=3.54$,
the edge polarization $p_x^{\textrm{edge}}$ vanishes so that the topological phase becomes trivial.
The appearance and disappearance of this new topological phase are caused only by
the gap closure of the Wannier bands at $\nu=0.5$ (in other words, topological properties
of the Wannier bands change due to the Wannier gap closure).
This shows that the quench dynamics can produce new topological phases, although the coherent dynamics
do not involve any energy gap closure.

Figure~\ref{fig9}(b2-d2) present the results for the final Hamiltonian in the type-I quadrupole insulating phase.
We find that the states
evolve into the type-II AQTI phase at $t=1.04$ with quantized quadrupole moments ($q_{xy}=e/2$) and quantized edge polarizations $p_x^{\textrm{edge}}=\pm e/2$ but $p_y^{\textrm{edge}}=0$ due to closure of the gap of $\textrm{ES}_y$. The nested entanglement spectrum has zero-energy modes in the region, reflecting the existence of fractional corner charges ($Q^{\textrm{corner}}=\pm e/2$) for
the parent Hamiltonian under open boundary conditions.
At $t=1.14$, the gap of the Wannier bands $\nu_y$ vanishes, yielding quantized edge polarizations
$p_y^{\textrm{edge}}=\pm e/2$, which signals the transition into the type-I AQTI
where the basic relation
$(p_y^{\textrm{edge }}+p_x^{\textrm{edge }}-q_{xy}-Q^{\textrm{corner }})\text{mod}(1)=0$ is satisfied.
This phase remains until the Wannier bands $\nu_y$ close the gap at $t=1.66$, followed by the reemergence of the
type-II QTI. We also observe that the type-II AQTI reappears at $t=3.94$ as a result of the vanishing
of the gap of $\textrm{ES}_y$.

\begin{figure}[t]
  \includegraphics[width=3.3in]{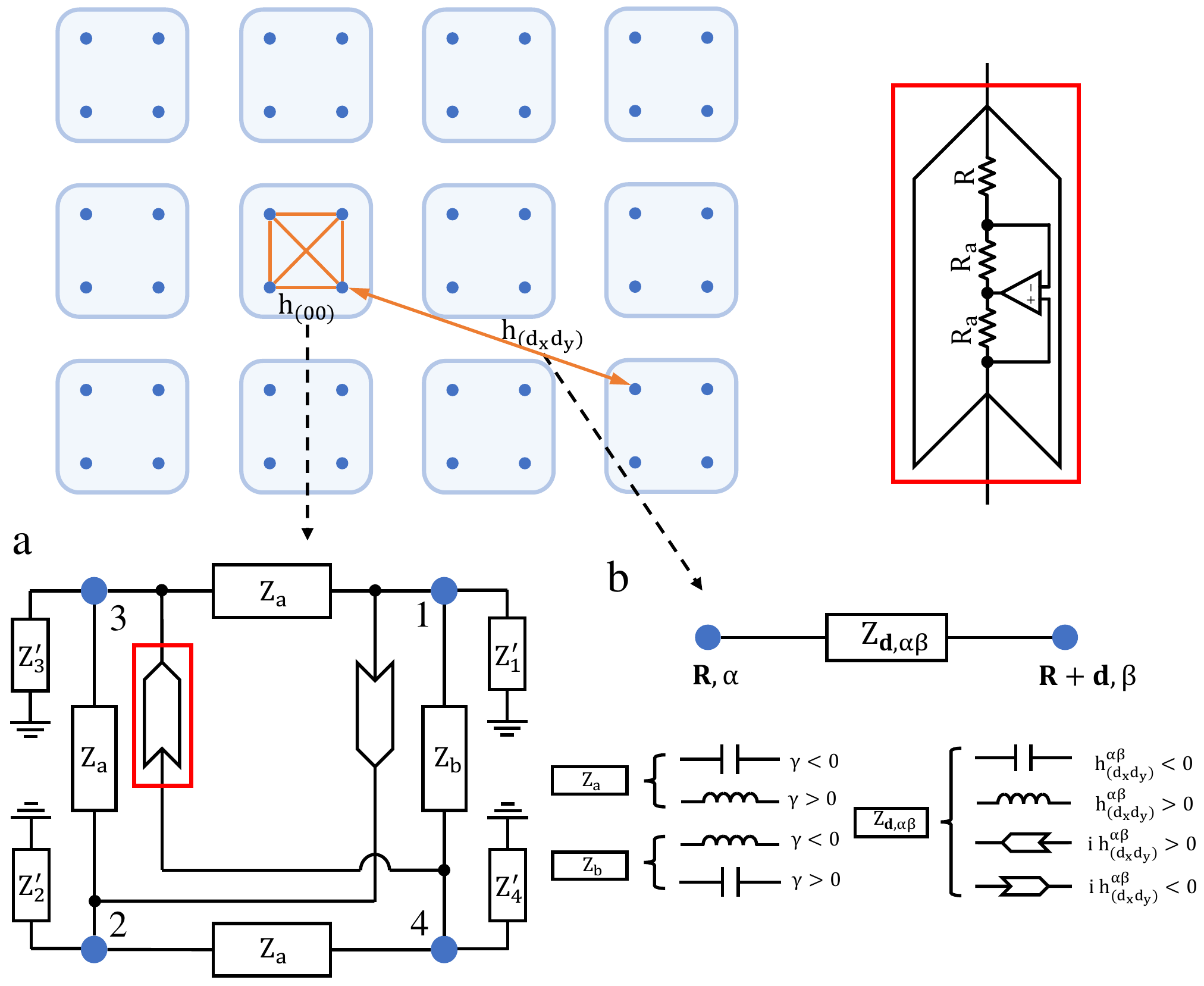}
  \caption{A scheme to realize our Hamiltonian in electric circuits.
  a, The electric device configuration to simulate the term $h_{(00)}$ within a unit cell.
  Here, the box labelled by $Z_a$ or $Z_b$ can be either a capacitor or a inductor depending on
  the sign of $\gamma$ with their impedances being $Z_a=-1/(i\gamma)$ and $Z_b=1/(i\gamma)$, respectively.
  The device in the red box denotes a negative impedance converter with current inversion (INIC),
  the impedance of which depends on the direction of the current. The resistance of the INIC should be
  taken as $R=1/\Delta$. The boxes labelled by $Z_{a=1,2,3,4}'$ represent
  device configurations to eliminate the unnecessary onsite terms (see Appendix H).
  b, A suitable electric device is applied to connect two nodes: $({\bf R},\alpha)$ and $({\bf R}+{\bf d},\beta)$
  for distinct unit cells. The device can be either a capacitor, or an inductor or an INIC with the
  impedance being $Z_{{\bf d},\alpha\beta}=i/h_{(d_xd_y)}^{\alpha\beta}$, as shown
  in the lower right corner figure. Here, ${\bf d}=d_x{\bf e}_x+d_y{\bf e}_y$.
  }
\label{fig10}
\end{figure}

\section{Experimental realization}
\label{sec7}
These new phases can be observed through quench dynamics in cold atom experiments.
In fact, Ref.~\cite{Taylor2017Science} has introduced an experimental scheme to realize the
type-I quadrupole model. In the scheme, laser beams are used to
engineer a superlattice with four sites in each unit cell [see Fig.~\ref{fig1}(a)].
Tunnelling along $y$ is suppressed by a linear potential.
Then, Raman laser beams are applied to restore the hopping with a phase of $\pi$ per plaquette~\cite{Bloch2013PRL,Ketterle2013PRL}.
Initially, we can tune the superlattice to realize large barriers between unit cells,
which suppress the tunnelling between unit cells despite the presence of Raman lasers,
realizing our initial Hamiltonian with almost zero $\lambda$. The cold atoms are prepared in
the ground state of this Hamiltonian. After that,
we suddenly change the model to our final Hamiltonian by tuning the superlattice and
Raman lasers and perform the tomography of the evolving states.
We can achieve the tomography by first removing atoms at two sites in each unit cell
and then performing the tomography of the remaining sites by time-of-flight measurements~\cite{Lewenstein2014PRL,Sengstock2016}.
The atoms at these sites can be kicked out of the trap by
shining resonant laser beams to the sites to excite them
to the $P$ state, which rapidly decays through spontaneous emission and escapes the trap.
This is feasible given that current experiments can realize laser beams with the diameter as small as
$600\textrm{ nm}$~\cite{Bloch2011Nat,Greiner2009Nat}, comparable to
the lattice constant. With measured states, the quadrupole moments, entanglement spectrum, and edge polarization
can be obtained.

In addition, we propose a scheme to simulate the Hamiltonian~(\ref{Ham1})
in electric circuits to realize these new phases, as illustrated in Fig.~\ref{fig10}. In fact, the Hamiltonians
(\ref{SimHam}), (\ref{HamSim2}) and (\ref{H4})
can be experimentally implemented using simplified electric networks.
We note that the type-I QTI has already been observed in electric circuits~\cite{Thomale2018NP}.
In the circuits, a Hamiltonian is simulated by a Laplacian, a matrix connecting
electric potentials with currents, i.e., ${\bf I}=J{\bf V}$,
where $\bf I$ and $\bf V$ are column vectors, each entry of which denotes the corresponding current flowing into the
corresponding node and the electric potential there, respectively~\cite{Thomale2018CP}.
Appropriate electric devices, such as capacitors, inductors and INICs~\cite{Thomale2019PRL,ChenBook}, are applied to connect different nodes to mimic the hopping between different sites
within or outside of a unit cell. The edge polarization and quadrupole moment can be
obtained by measuring the single-point impedances of the circuit, and the existence of corner modes can be
shown
by measuring the resonance of two-point impedances near the corners,
as detailed in Appendix H. Type-II QTIs may also be realized in other systems, such as
solid-state materials, cold atoms and photonic crystals.

\section{Conclusion}
\label{sec8}

In summary, we discover a novel type of QTI violating an established classical
relation. This relation is maintained in a quantum system as the Wannier band and edge energy spectrum close
their gaps simultaneously.
However, the appearance of the type-II QTI indicates that the two gaps do not necessarily vanish at
the same time.
We also find the anomalous quadrupole insulating phases that cannot be characterized by the
previously introduced Wannier-sector polarizations.
We introduce a new topological invariant for a Wilson line to characterize them for a system with reflection symmetries;
such methods to characterize the topological property change for a Wannier band can also be applied to
systems with other symmetries.
In addition, we predict another novel topological phase with quantized edge polarizations but without
zero-energy corner modes and quadrupole moments.
Based on
the type-II insulating phase, we find new pumping phenomena, leading to novel 3D higher-order
topological insulators.
We further show that these new topological phenomena in two dimensions can emerge from quench dynamics,
which can be experimentally achieved in ultracold atomic gases.
We also introduce an experimental proposal with electric
circuits to simulate these new models.
Our results demonstrate that new multipole topological insulators with exotic
properties can exist beyond classical constraints, opening a new direction for exploring multipole
topological insulators.

\begin{acknowledgments}
We thank Q. Zeng, Y.-L. Tao and D.-L. Deng for helpful discussions.
We would also like to thank B. J. Wieder for helpful discussions and
bringing us Ref.~\cite{Wieder2019}, where the breakdown
of the correspondence between Wannier and edge spectra was found in a system with neither particle-hole
nor chiral symmetry.
Y.B.Y., K.L. and Y.X.
are supported by the start-up fund from Tsinghua University,
the National Thousand-Young-Talents Program and the National Natural Science Foundation
of China (11974201).
We acknowledge in addition support from the Frontier Science Center for Quantum Information of the Ministry of Education of China, Tsinghua University Initiative Scientific Research Program, and the National key Research and Development Program of China (2016YFA0301902).
\end{acknowledgments}

\begin{widetext}


\section*{Appendix A: The relationship between two simplest insulators and the type-II AQTI}

\setcounter{equation}{0} \setcounter{figure}{0} \setcounter{table}{0} %
\renewcommand{\theequation}{A\arabic{equation}} \renewcommand{\thefigure}{A%
\arabic{figure}} \renewcommand{\bibnumfmt}[1]{[#1]} \renewcommand{%
\citenumfont}[1]{#1}

In this appendix, by
relating two simplest insulators with zero and nonzero quadrupole moments
to the type-II AQTI, we verify that the quadrupole moments obtained in the type-II AQTI are reasonable.

Let us first consider a topologically trivial model described by the Hamiltonian,
\begin{equation}
H^{\textrm{trivial}}=\sum_{\bf R} \hat{c}^\dagger_{{\bf R}}h \hat{c}_{\bf R},
\end{equation}
where
\begin{equation}
h=\left(
           \begin{array}{cc}
             0 & -\sigma_0+i\sigma_y \\
             H.c. & 0 \\
           \end{array}
         \right).
\end{equation}
This model has four orbital degrees of freedom in each unit cell. There is no tunneling between
distinct unit cells [see Fig.~\ref{figA1}(a)], showing that it is a topologically trivial atomic insulator. Evidently,
this model does not possess quadrupole moment at half filling since the positive charges and the electrons occupy
the same position. We note that this trivial insulator corresponds
to a trivial phase in the model in Ref.~\cite{Taylor2017Science} with $\lambda_x=\lambda_y=0$
and $\gamma_x=-\gamma_y=-1$.

To see the connection between the trivial insulator and the type-II AQTI, we construct a model $H_1(t)$ parameterized by $t\in[0,1]$ as follows,
\begin{equation}
H_1(t)=\left\{
\begin{aligned}
&(1-2t)H^{\textrm{trivial}}+2t H(\gamma=-1)    &0 \leq t < 0.5, \\
&H(\gamma=2.4t-2.2)   &0.5 \leq t \leq 1,
\end{aligned}
\right.
\end{equation}
where $H$ denotes the Hamiltonian (1) in the main text. This Hamiltonian connects the trivial
insulator for $t=0$ with the type-II AQTI for $t=1$ corresponding to the type-II AQTI marked
by the green square in Fig.~\ref{fig2}(c) in the main text.

In Figs.~\ref{figA1}(b,c), we present the energy spectrum (obtained under periodic boundary conditions along both $x$ and $y$ directions) and the quadrupole moment as $t$ varies from $0$ to $1$. We find that at an intermediate point, the bulk energy gap closes and simultaneously the quadrupole moment changes abruptly from $0$ to $1/2$. This suggests that the type-II AQTI is topologically distinct from the trivial insulator with respect to the bulk property and this distinction can be characterized
by the quadrupole moment. This also shows that the type-II AQTI is completely different from the trivial insulator attached with
a pair of SSH models. Even though the edge polarizations are the same for these two models and zero-energy corner modes exist
for both of them, their bulk topological properties are topologically different (the former has nonzero quadrupole moment while
the latter has zero one).

To see the connection between the type-I QTI and the type-II AQTI, we construct another model $H_2(t)$ parameterized by $t\in[0,1]$
as follows,
\begin{equation}
H_2(t)=(1-t) H^{\textrm{typeI}} + t H(\gamma=0.2).
\end{equation}
Here,
\begin{equation}
H^{\textrm{typeI}}=\sum_{\bf R}\left[ \hat{c}^\dagger_{{\bf R}+{\bf e}_x}h_{(1 0)}
+\hat{c}^\dagger_{{\bf R}+{\bf e}_y}h_{(0 1)}
\right]\hat{c}_{\bf R}+H.c.,
\end{equation}
with
\begin{equation}
h_{10}=\left(
         \begin{array}{cccc}
           0 & 0 & 0 & 0 \\
           0 & 0 & 0 & 1 \\
           1 & 0 & 0 & 0 \\
           0 & 0 & 0 & 0 \\
         \end{array}
       \right),
\text{ and }
h_{01}=\left(
         \begin{array}{cccc}
           0 & 0 & 0 & 0 \\
           0 & 0 & -1 & 0 \\
           0 & 0 & 0 & 0 \\
           1 & 0 & 0 & 0 \\
         \end{array}
       \right)
\end{equation}
describes a typical QTI with only intercell hoppings [see Fig.~\ref{figA1}(d)] corresponding to the model in Ref.~\cite{Taylor2017Science} with $\gamma_x=\gamma_y=0$ and $\lambda_x=\lambda_y=1$.
This model has nonzero quadrupole moment, zero-energy corner modes leading to fractional corner charges at half filling
and edge polarizations at all four boundaries~\cite{Taylor2017Science} and thus belongs to the type-I QTI.
$H_2(t)$ connects this type-I QTI with a type-II AQTI.

\begin{figure*}[t]
\includegraphics[width=6.2in]{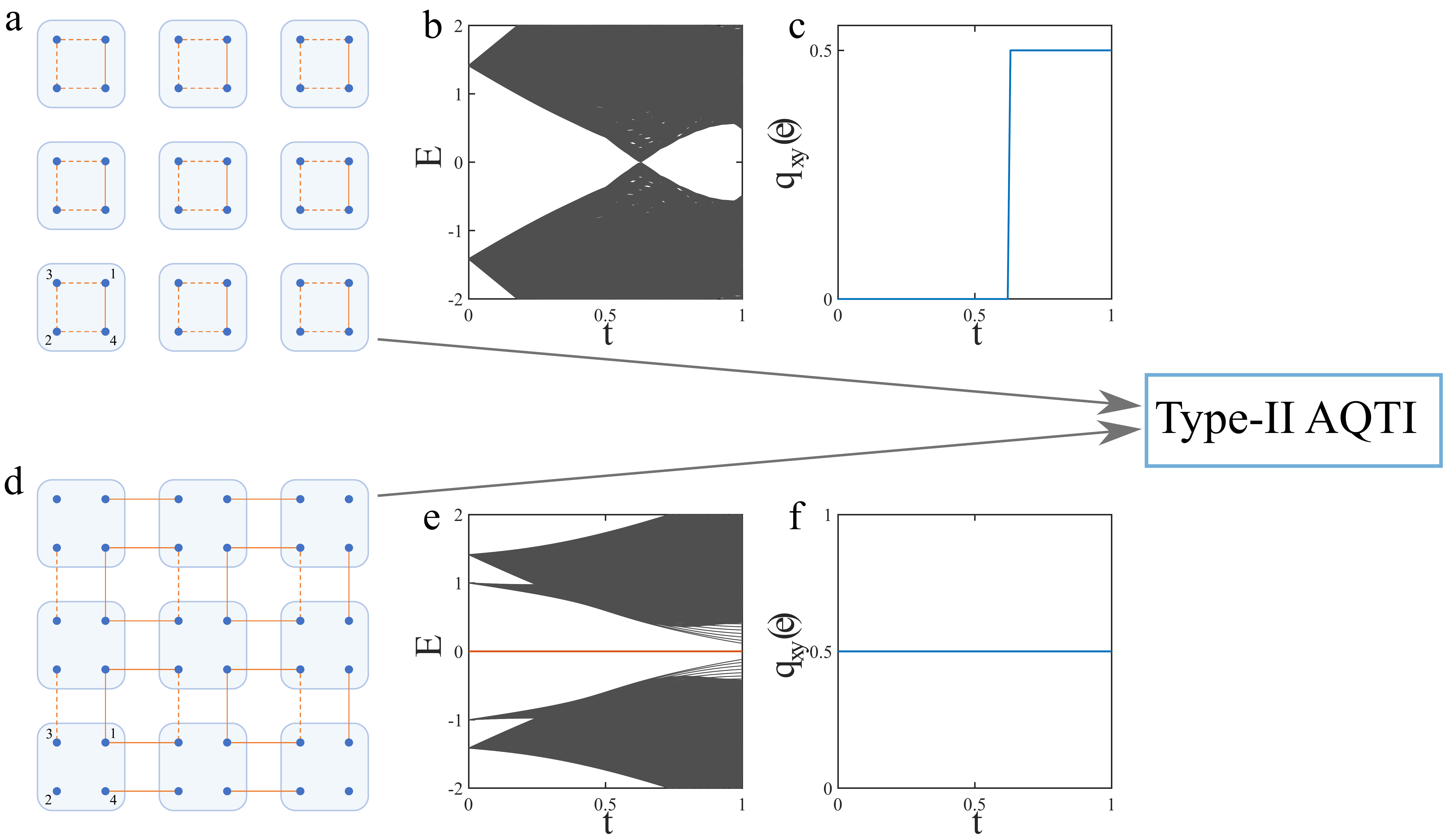}
  \caption{The relationship between the type-II AQTI and two simplest models without and with
  quadruopole moments.
  a, Schematics of a simplest model of a trivial insulator with zero quadrupole moments and
  d, Schematics of a simplest model of type-I QTI with quantized nonzero quadrupole moment.
  In a and d, the solid and dashed red lines represent the hopping terms with positive and negative signs, respectively.
  The index $\alpha=1, 2, 3, 4$ denotes the sites within a unit cell.
  b, The energy spectrum (obtained under periodic boundary conditions along both $x$ and $y$ directions) and c, the quadrupole moment
  as we change a system from the trivial insulating phase in a to a type-II anomalous quadrupole topological insulating phase.
  Here the quadrupole moment changes suddenly from zero to $e/2$ when the bulk energy gap vanishes,
  implying that the type-II AQTI is topologically distinct from the trivial insulator.
  This also distinguishes the type-II AQTI from the trivial insulator attached with a pair of SSH models.
  e, The energy spectrum (obtained under open boundary conditions along both $x$ and $y$ directions) and f, the quadrupole moment
  as we change a system from the type-I QTI in d to a type-II AQTI. Here the quadrupole moment remains
  unchanged and both bulk and edge energy gaps remain open and the zero-energy corner states exist during the whole process.
  It shows that the type-II AQTI possesses the same quadrupole moment as the type-I QTI.
  }
\label{figA1}
\end{figure*}

Fig.~\ref{figA1}(e) illustrates that both bulk and edge energy gaps (see the energy spectrum
obtained under open boundary conditions along both $x$ and $y$ directions) remain open as we vary
$t$ from $0$ to $1$. During this process, the quadrupole moment also remains unchanged (equal to $e/2$) [see Fig.~\ref{figA1}(f)].
This suggests that the type-II AQTI shares the same bulk topological property as the type-I QTI characterized
by the quadrupole moment. Their difference arises from the Wannier band gap closing, leading to different
edge polarizations.

The above discussion suggests the validity of our calculation of the quadrupole moments in our model in the sense
that the quadrupole moments change when the bulk energy gap closes when a system changes from a typical atomic
insulator without quadrupole moments to a type-II AQTI. In addition, this moment remains unchanged when we change a system from
a typical QTI with nonzero quadrupole moments to a type-II AQTI. This confirms that the type-II AQTI is a new type of qudrupole topological insulating phase.

\section*{Appendix B: The Hamiltonian in momentum space}

\setcounter{equation}{0} \setcounter{figure}{0} \setcounter{table}{0} %
\renewcommand{\theequation}{B\arabic{equation}} \renewcommand{\thefigure}{B%
\arabic{figure}} \renewcommand{\bibnumfmt}[1]{[#1]} \renewcommand{%
\citenumfont}[1]{#1}

In this appendix, we write down the explicit form of our Hamiltonian in momentum space as
\begin{equation}
H({\bf k})=\sum_{i,j=0}^{3} g_{ij}({\bf k})\tau_i\otimes\sigma_j,
\end{equation}
where all the nonzero $g_{ij}$s are given by
\begin{align}
&g_{01}=2 t_2 \sin(2k_x) \\
&g_{03}=-4 t_2 \cos(k_x)\sin(k_y) \\
&g_{10}=\gamma+2 t_1 \cos(k_x)+2 t_1' \cos(k_y)+4 t_2 \cos(k_x)\cos(k_y)-4 t_2' \cos(2k_x)\cos(k_y) \\
&g_{21}=-2 t_1 \sin(k_y)-2 t_2 \sin(2k_y)-4 t_2' \cos(k_x)\sin(k_y)+4 t_2' \cos(k_x)\sin(2k_y) \\
&g_{22}=\gamma-2 t_1 \cos(k_y)-2 t_2 \cos(2k_y)-4 t_2' \cos(k_x)\cos(k_y)+4 t_2' \cos(k_x)\cos(2k_y) \\
&g_{23}=-2 t_1 \sin(k_x)-4 t_2 \sin(k_x)\cos(k_y)+4 t_2' \sin(2k_x)\cos(k_y) \\
&g_{31}=-4 t_2 \cos(k_x)\sin(k_y)-2 t_2' \sin(2k_y) \\
&g_{32}=\Delta+2 t_1' \cos(k_x)+2 t_2 \cos(2k_x)-2 t_2' \cos(2k_y)-4 t_2 \cos(k_x)\cos(k_y) \\
&g_{33}=-2 t_2 \sin(2k_x).
\end{align}

One can easily check that this Hamiltonian respects the reflection symmetry: $\hat{m}_xH(k_x,k_y)\hat{m}_x^{-1}=H(-k_x,k_y)$ with $\hat{m}_x=\tau_1\otimes\sigma_3$, $\hat{m}_y H(k_x,k_y) \hat{m}_y^{-1}=H(k_x,-k_y)$ with $\hat{m}_y=\tau_1\otimes\sigma_1$,
and the particle-hole symmetry: $\Xi H(k_x,k_y)\Xi^{-1}=-H(-k_x,-k_y)$ with $\Xi=\tau_3\otimes\sigma_0 \kappa$.

\section*{Appendix C: Topologically trivial phase}

\setcounter{equation}{0} \setcounter{figure}{0} \setcounter{table}{0} %
\renewcommand{\theequation}{C\arabic{equation}} \renewcommand{\thefigure}{C%
\arabic{figure}} \renewcommand{\bibnumfmt}[1]{[#1]} \renewcommand{%
\citenumfont}[1]{#1}

In this appendix, we display richer physics in the topologically trivial phase as shown in Fig.~\ref{figC1}.
We find four regions corresponding to distinct $N_{\nu_\lambda}^0$ ($\lambda=x,y$):
$(N_{\nu_x}^0,N_{\nu_y}^0)=(0,0), (4,0), (2,0), (2,2)$, even though they all
have zero quadrupole moments and edge polarizations. This is reasonable as the edge states of the Wannier
spectrum with zero Wannier centers do not contribute to the dipole moments.
Despite being topologically trivial, these phases can have nonzero Wannier-sector polarization when odd number of pairs of
edge states of the Wannier spectrum exists with zero eigenvalues. In particular,
in the region with $(N_{\nu_x}^0,N_{\nu_y}^0)=(2,2)$, although both Wannier-sector polarizations are nonzero, i.e., $(p_y^{\nu_x},p_x^{\nu_y})=(1/2,1/2)$, the quadrupole moment and edge polarizations do not exist,
implying the breakdown of the Wannier-sector polarizations to characterize the
edge polarization and the quadrupole moments.

\begin{figure*}[t]
\includegraphics[width=4.5in]{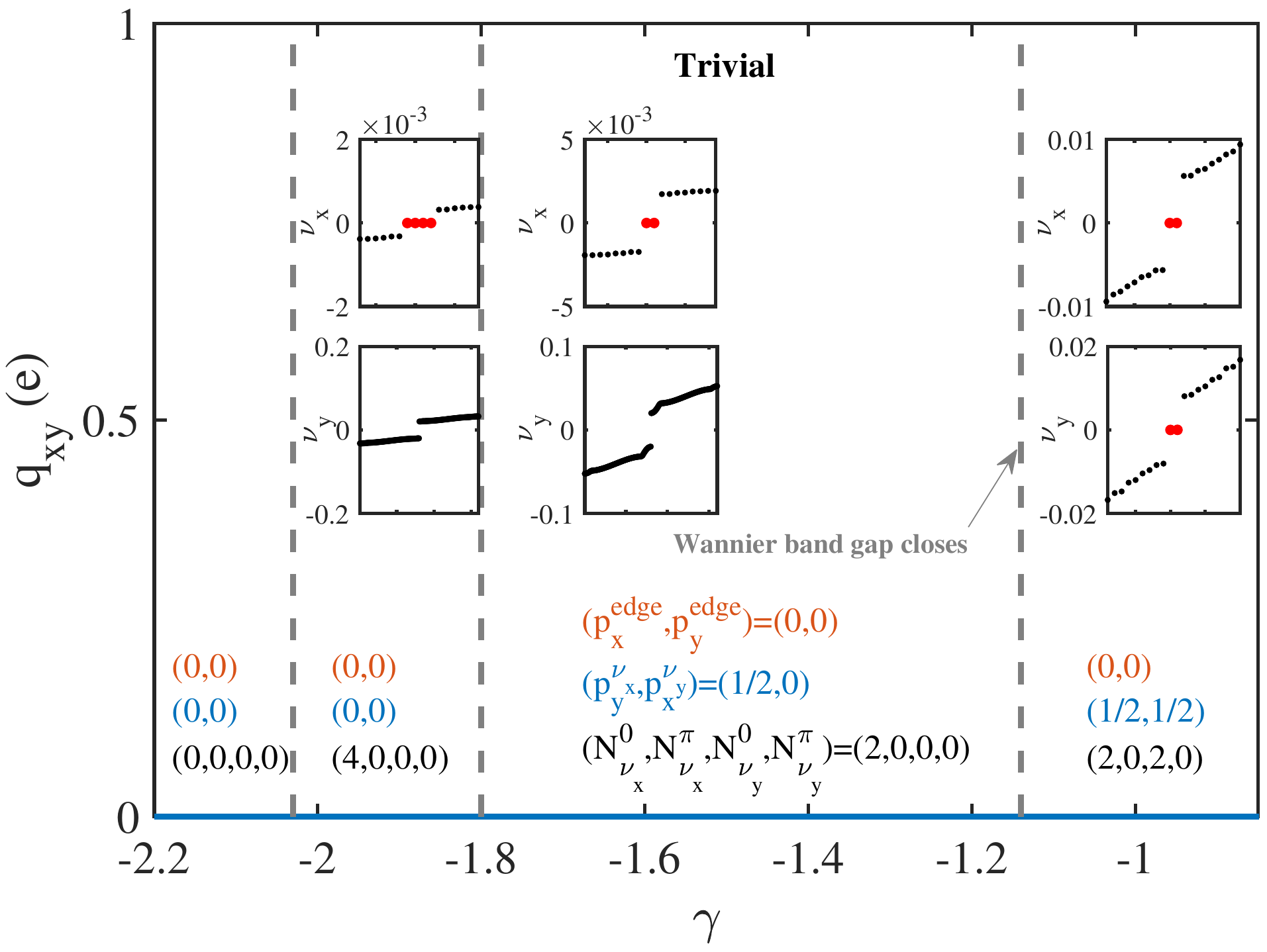}
  \caption{Distinct regions in topologically trivial insulators.
  The subsets display the Wannier spectrum $\nu_x$ ($\nu_y$) for periodic boundary conditions
  along $x$ ($y$) and open ones along $y$ ($x$) with the isolated Wannier centers
  highlighted by red circles. The edge polarization $(p_x^{\textrm{edge}},p_y^{\textrm{edge}})$, the Wannier-sector polarization
  $(p_y^{\nu_x},p_x^{\nu_y})$, the number of edge states of the Wannier Hamiltonian
  $N_\nu \equiv (N_{\nu_x}^0, N_{\nu_x}^{\pi}, N_{\nu_y}^0, N_{\nu_y}^{\pi})$ are also shown.
  The vertical dashed lines represent the critical points where the Wannier spectrum gap closes.
   }
\label{figC1}
\end{figure*}

\section*{Appendix D: A topological invariant for a Wilson line}

\setcounter{equation}{0} \setcounter{figure}{0} \setcounter{table}{0} %
\renewcommand{\theequation}{D\arabic{equation}} \renewcommand{\thefigure}{D%
\arabic{figure}} \renewcommand{\bibnumfmt}[1]{[#1]} \renewcommand{%
\citenumfont}[1]{#1}

In this appendix, we will derive a topological invariant for a Wilson line characterizing the change of the topological
property of the Wannier band.
\subsection*{1. Generic symmetry constraint of a Wilson line}
Let us consider a generic symmorphic lattice symmetry: ${\bf r}\rightarrow D_g{\bf r}$, described by
\begin{equation}
g H({\bf k})g^{-1}=H(D_g {\bf k}),
\end{equation}
where $g$ is a representation of the symmetric operation in momentum space, which is unitary.
If $|u_{\bf k}^n\rangle$ is an eigenstate of $H({\bf k})$ corresponding to the eigenenergy $E_{\bf k}^n$,
then $g|u_{\bf k}^n\rangle$ is an eigenstate of $H(D_g {\bf k})$ corresponding to the same energy.
As a result, we can write $g|u_{\bf k}^n\rangle$
in terms of eigenstates of $H(D_g {\bf k})$,
\begin{equation}
g|u_{\bf k}^n\rangle=\sum_{m=1}^{N_{\textrm{occ}}} |u_{D_g {\bf k}}^m\rangle \langle u_{D_g {\bf k}}^m| g |u_{\bf k}^n\rangle=\sum_{m=1}^{N_{\textrm{occ}}} |u_{D_g {\bf k}}^m\rangle B_{g,{\bf k}}^{mn},
\end{equation}
where $B_{g,{\bf k}}^{mn}=\langle u_{D_g {\bf k}}^m| g |u_{\bf k}^n\rangle$ is a unitary sewing matrix
that connects states at ${\bf k}$ with states at $D_g {\bf k}$. Since
we are interested in the occupied bands, the superscript only enumerates the occupied states from $1$ to the total
number of the occupied states $N_{\textrm{occ}}$.

We now define the Wilson line following a path $C$ in momentum space from ${\bf k}_i$ to ${\bf k}_f$ as
\begin{equation}
\mathcal{W}_{C,{\bf k}_f \leftarrow {\bf k}_i}=F_{{\bf k}_f-\Delta {\bf k}_{N-1}}\cdots F_{{\bf k}_i+\Delta {\bf k}_1}F_{{\bf k}_i},
\end{equation}
where $[F_{{\bf k}_{j^\prime}}]^{mn}=\langle u_{{\bf k}_{j^\prime}+\Delta {\bf k}_{j^\prime +1} }^m|u_{{\bf k}_{j^\prime}}^n\rangle$ with $m$ and $n$ being the indices for the occupied bands and $\Delta {\bf k}_{j}$ dividing the
trajectory into $N$ segments and $j^\prime=0,1,\cdots, N-1$ and
${\bf k}_{j^\prime}={\bf k}_i+\sum_{j=1}^{j^\prime} \Delta {\bf k}_{j}$.
In the limit $N\rightarrow \infty$, we can write the Wilson line
in the following compact form,
\begin{equation}
\mathcal{W}_{C,{\bf k}_f \leftarrow {\bf k}_i}=\lim_{N\rightarrow\infty}\prod_{n=1}^{N-1} (I-i\Delta {\bf k}_n \cdot
{\bf \mathcal{A}}_{{\bf k}_n})
=\overline{\exp}(-i\int_{\bf k_i}^{\bf k_f} {\bf \mathcal{A}}_{\bf k}\cdot d{\bf k}),
\end{equation}
where $\overline{\exp}(\cdots)$ denotes the path-ordered exponential, 
and $[{\bf \mathcal{A}}_{\bf k}]^{mn}=-i\langle u_{{\bf k}}^{m}|\partial_{{\bf k}} u_{{\bf k}}^{n}\rangle$
is the non-Abelian Berry connection. Since $ {\bf \mathcal{A}}_{\bf k}$ is a Hermitian matrix, the Wilson loop is unitary, reminiscent of a time evolution operator.

Since $|u_{\bf k}^n\rangle=g^\dagger \sum_{m}|u_{D_g{\bf k}}^m\rangle B_{g,{\bf k}}^{mn}$, we have
\begin{equation}
[F_{{\bf k}_{j}}]^{mn}=\sum_{m^\prime,n^\prime}(B_{g,{\bf k}_{j+1}}^{\dagger})^{mm^\prime} \langle u_{D_g {\bf k}_{j+1}}^{m^\prime} | u_{D_g {\bf k}_j}^{n^\prime}\rangle B_{g,{\bf k}_j}^{{n^\prime}n},
\end{equation}
leading to
\begin{equation}
 \mathcal{W}_{C,{\bf k}_f \leftarrow {\bf k}_i} =B_{g,{\bf k}_f}^\dagger \mathcal{W}_{D_gC,D_g {\bf k}_f \leftarrow D_g {\bf k}_i}B_{g,{\bf k}_i},
\end{equation}
which can be equivalently expressed as
\begin{equation}
B_{g,{\bf k}_f} \mathcal{W}_{C,{\bf k}_f \leftarrow {\bf k}_i} B_{g,{\bf k}_i}^\dagger=
\mathcal{W}_{D_gC,D_g {\bf k}_f \leftarrow D_g {\bf k}_i}=\mathcal{W}_{D_g {\bf k}_f \leftarrow D_g {\bf k}_i},
\label{SC}
\end{equation}
where $D_gC$ denotes a new path obtained by applying the symmetry operation on the original path $C$ (we will skip
this notation for simplicity).

\subsection*{2. Reflection symmetry constraint and topological invariants}
We now consider the reflection symmetry $M_x$, which, based on Eq.~(\ref{SC}), gives
\begin{equation}
B_{m_x,{\bf k}_f} \mathcal{W}_{x,{\bf k}} B_{m_x,{\bf k}_i}^\dagger=
\mathcal{W}_{x,(-k_x-2\pi,k_y) \leftarrow (-k_x,k_y)}=\mathcal{W}_{-x,M_x{\bf k}},
\end{equation}
where $\mathcal{W}_{x,{\bf k}}$ and $\mathcal{W}_{-x,{\bf k}}$ are defined as
\begin{eqnarray}
\mathcal{W}_{x,{\bf k}}&=&\mathcal{W}_{(k_x+2\pi,k_y)\leftarrow (k_x,k_y)}, \\
\mathcal{W}_{-x,{\bf k}}&=&\mathcal{W}_{(k_x-2\pi,k_y)\leftarrow (k_x,k_y)},
\end{eqnarray}
with $x$ and $-x$ labelling the direction that a Wilson loop is obtained.
Since the Wilson loop is unitary, we can write it as
\begin{equation}
\mathcal{W}_{x,{\bf k}}^{(0)}=\mathcal{W}_{x,(0,k_y)}=e^{iH_{\mathcal{W}_x}^{(0)}(k_y)},
\end{equation}
where $H_{\mathcal{W}_x}^{(0)}(k_y)=-i\log(\mathcal{W}_{x,{\bf k}}^{(0)})$ is the Wannier Hamiltonian for $k_x=0$. Given
that the Wannier Hamiltonian is a multivalued function of the Wilson loop, we redefine it with respect to
a logarithm branch cut $\epsilon$,
\begin{equation}
H_{\mathcal{W}_x}^\epsilon (k_y)=-i\log_{\epsilon}(\mathcal{W}_{x,{\bf k}}^{(0)}),
\end{equation}
where we take $\log_{\epsilon} e^{i\phi}=i\phi$, for $\epsilon\le \phi<\epsilon+2\pi$.

Now we can obtain the symmetry constraint on the Wannier Hamiltonian,
\begin{equation}
B_{m_x,(0,k_y)}H_{\mathcal{W}_x}^\epsilon (k_y)B_{m_x,(0,k_y)}^\dagger=-H_{\mathcal{W}_x}^{-\epsilon} (k_y)+2\pi I_{N_{\textrm{occ}}}=-H_{\mathcal{W}_x}^{-\epsilon+2\pi}(k_y) +4\pi I_{N_{\textrm{occ}}}.
\end{equation}
\begin{proof}
\begin{align}
&B_{m_x,(0,k_y)}H_{\mathcal{W}_x}^\epsilon (k_y)B_{m_x,(0,k_y)}^\dagger  \nonumber \\
=&-iB_{m_x,(0,k_y)}\log_{\epsilon}(\mathcal{W}_{x,{\bf k}}^{(0)})B_{m_x,(0,k_y)}^\dagger \nonumber \\
=&-i\log_{\epsilon}\left[B_{m_x,(0,k_y)}\mathcal{W}_{x,{\bf k}}^{(0)}B_{m_x,(0,k_y)}^\dagger\right] \nonumber \\
=&-i\log_{\epsilon}\left[\mathcal{W}_{-x,{\bf k}}^{(0)}\right] \nonumber \\
=&-i\log_{\epsilon}\left[(\mathcal{W}_{x,{\bf k}}^{(0)})^{-1}\right] \nonumber \\
=& -i\sum_n \log_{\epsilon} (e^{-i\nu_{x}^{(n)}}) |\nu_x^{(n)}\rangle\langle\nu_x^{(n)}| \nonumber \\
=& -i\sum_n \left[-\log_{-\epsilon}(e^{i\nu_{x}^{(n)}}) + 2\pi i \right] |\nu_x^{(n)}\rangle\langle\nu_x^{(n)}| =-H_{\mathcal{W}_x}^{-\epsilon}(k_y) +2\pi I_{N_{\textrm{occ}}}\\
=& -i\sum_n \left[-\log_{-\epsilon+2\pi}(e^{i\nu_{x}^{(n)}}) + 4\pi i \right] |\nu_x^{(n)}\rangle\langle\nu_x^{(n)}| =-H_{\mathcal{W}_x}^{-\epsilon+2\pi}(k_y) +4\pi I_{N_{\textrm{occ}}},
\end{align}
where $|\nu_x^{(n)}\rangle$ is the $n$th eigenvector of the Wilson loop $\mathcal{W}_{x,{\bf k}}^{(0)}$
corresponding to the eigenvalue $e^{i\nu_{x}^{(n)}}$. In the derivation, we have also used the following relations
\begin{align}
&\log_{\epsilon} (e^{-i\phi})= -\log_{-\epsilon} (e^{i\phi}) + 2\pi i \\
&\log_{-\epsilon} (e^{i\phi})= \log_{-\epsilon+2\pi} (e^{i\phi}) - 2\pi i.
\end{align}
\end{proof}

We now define the Wilson line with respect to $\epsilon$ by
\begin{equation}
\mathcal{W}_{k_x\leftarrow 0}^{\epsilon}(k_y) \equiv \mathcal{W}_{k_x\leftarrow 0}(k_y) \exp(-iH_{W_x}^{\epsilon}(k_y)\frac{k_x}{2\pi}),
\end{equation}
where $\mathcal{W}_{k_x\leftarrow 0}(k_y)\equiv \mathcal{W}_{(k_x,k_y)\leftarrow (0,k_y)}$.
It can be easily checked that $\mathcal{W}_{2\pi\leftarrow 0}^{\epsilon}(k_y)=I_{N_{\textrm{occ}}}$,
implying that $\mathcal{W}_{k_x\leftarrow 0}^{\epsilon}(k_y)$ is periodic with respect to $k_x$.

In the following, we will prove the symmetry constraints on the Wilson line with respect to $\epsilon$,
\begin{equation}
B_{m_x,(k_x,k_y)}\mathcal{W}_{k_x \leftarrow 0}^{\epsilon}(k_y) B_{m_x,(0,k_y)}^{\dagger}
= \mathcal{W}_{-k_x+2\pi\leftarrow 0}^{-\epsilon}(k_y) \exp(-ik_x)
= \mathcal{W}_{-k_x+2\pi\leftarrow 0}^{-\epsilon+2\pi}(k_y) \exp(-2ik_x).
\label{EpsilonWL}
\end{equation}
\begin{proof}
We can directly obtain
\begin{align}
B_{m_x,(k_x,k_y)}\mathcal{W}_{k_x \leftarrow 0}^{\epsilon}(k_y) B_{m_x,(0,k_y)}^{\dagger}
&=B_{m_x,(k_x,k_y)}\mathcal{W}_{k_x \leftarrow 0}(k_y)B_{m_x,(0,k_y)}^{\dagger} B_{m_x,(0,k_y)}e^{-iH_{W_x}^{\epsilon}(k_y)\frac{k_x}{2\pi}}B_{m_x,(0,k_y)}^{\dagger} \\
&=\mathcal{W}_{-k_x \leftarrow 0}(k_y) e^{iH_{W_x}^{-\epsilon}(k_y)\frac{k_x}{2\pi}} e^{-ik_x} =\mathcal{W}_{-k_x\leftarrow 0}^{-\epsilon} e^{-ik_x} \\
&= \mathcal{W}_{-k_x \leftarrow 0}(k_y) e^{iH_{W_x}^{-\epsilon+2\pi}(k_y)\frac{k_x}{2\pi}} e^{-2ik_x} =\mathcal{W}_{-k_x\leftarrow 0}^{-\epsilon+2\pi} e^{-2ik_x}.
\end{align}
With the aid of the following equation
\begin{align}
\mathcal{W}_{k_x+2\pi\leftarrow 0}^{\epsilon}(k_y)&=\mathcal{W}_{k_x+2\pi\leftarrow 2\pi} [\mathcal{W}_{2\pi\leftarrow 0}(k_y) \exp(-iH_{W_x}^{\epsilon})] \exp(-iH_{W_x}^{\epsilon}\frac{k_x}{2\pi}) \\
&=\mathcal{W}_{k_x\leftarrow 0}(k_y) \exp(-iH_{W_x}^{\epsilon}\frac{k_x}{2\pi}) \\
&=\mathcal{W}_{k_x\leftarrow 0}^{\epsilon}(k_y),
\end{align}
we obtain Eq.~(\ref{EpsilonWL}).
\end{proof}

At $k_x=\pi$, $\epsilon=0$ or $\epsilon=\pi$, Eq.~(\ref{EpsilonWL}) leads to
\begin{align}
&B_{m_x,(\pi,k_y)}\mathcal{W}_{\pi \leftarrow 0}^{\epsilon=0}(k_y) B_{m_x,(0,k_y)}^{\dagger} = -\mathcal{W}_{\pi\leftarrow 0}^{\epsilon=0}(k_y) \label{BSmmetry1} \\
&B_{m_x,(\pi,k_y)}\mathcal{W}_{\pi \leftarrow 0}^{\epsilon=\pi}(k_y) B_{m_x,(0,k_y)}^{\dagger} = \mathcal{W}_{\pi\leftarrow 0}^{\epsilon=\pi}(k_y).
\label{BSmmetry2}
\end{align}

Let us present a theorem, showing that $B_{m_x,(\pi,k_y)}=B_{m_x,(0,k_y)}$ in a specific basis.
\begin{theorem}
Consider a Hamiltonian $H({\bf k})$ describing an insulator at half filling.
Suppose it respects reflection symmetry such that $\hat{m}_xH({\bf k})\hat{m}_x^\dagger=H(-k_x,k_y)$
and $\hat{m}_yH({\bf k})\hat{m}_y^\dagger=H(k_x,-k_y)$ with $\hat{m}_x$ and $\hat{m}_y$ anticommuting with each other. There exists a basis in which the sewing matrix takes the form
\begin{equation}
B_{m_x,(\pi,k_y)}=B_{m_x,(0,k_y)}=\left(
  \begin{array}{cc}
    I_{N_{\textrm{occ}}/2} & 0 \\
    0 & -I_{N_{\textrm{occ}}/2} \\
  \end{array}
\right).
\end{equation}
\end{theorem}
\begin{proof}
From definition of the sewing matrix, we have
\begin{eqnarray}
B_{m_x,(0,k_y)}^{mn}&=&\langle u_{(0,k_y)}^{m}|\hat{m}_x|u_{(0,k_y)}^{n}\rangle \\
B_{m_x,(\pi,k_y)}^{mn}&=&\langle u_{(\pi,k_y)}^{m}|\hat{m}_x|u_{(\pi,k_y)}^{n}\rangle.
\end{eqnarray}
Because $\hat{m}_x^2=1$, there are two eigenvalues: $\lambda =\pm 1$ for $\hat{m}_x$.
If $|m_{\lambda}\rangle$ is an eigenvector of $\hat{m}_x$ corresponding to an eigenvalue $\lambda$,
then $\hat{m}_y|m_{\lambda}\rangle$ is another eigenvector of $\hat{m}_x$ with an eigenvalue being
$-\lambda$, since $\hat{m}_x \hat{m}_y|m_{\lambda}\rangle=-\hat{m}_y \hat{m}_x|m_{\lambda}\rangle=-\lambda \hat{m}_y|m_{\lambda}\rangle$
arising from the anticommutation relation of $\hat{m}_x$ and $\hat{m}_y$, i.e., $\{\hat{m}_x,\hat{m}_y\}=0$.
In addition, if two eigenvectors $|m_{\lambda}^1\rangle$ and $|m_{\lambda}^2\rangle$
are orthogonal, i.e., $\langle m_{\lambda}^2|m_{\lambda}^1\rangle=0$,
then $\langle m_{\lambda}^2|\hat{m}_y^\dagger \hat{m}_y|m_{\lambda}^1\rangle=0$.
This tells us that the eigenvectors of $\hat{m}_x$ come in pairs with eigenvalues $\pm 1$.

Let us choose a basis $\beta=\beta_1\cup \beta_{-1}$ consisting of eigenvectors of $\hat{m}_x$, $\beta_1=\{|m_{1}^1\rangle,\cdots,|m_{1}^{N_{\textrm{occ}}}\rangle \}$
and $\beta_{-1}=\{\hat{m}_y|m_{1}^1\rangle,\cdots,\hat{m}_y|m_{1}^{N_{\textrm{occ}}}\rangle \}$, where
$N_{\textrm{occ}}$ is the number of occupied states at half filling for a fixed momentum.
Because of the reflection symmetry along $x$ respected by the Hamiltonian,
at $k_x=0$ or $\pi$, the Hamiltonian commutes with $\hat{m}_x$, i.e., $[H(k_x=0,k_y),\hat{m}_x]=0$.
As a result, in the basis $\beta$, $H(k_x=0,k_y)$ takes the following block form,
\begin{equation}
H(k_x=0,k_y)=\left(
    \begin{array}{cc}
      H_1(k_y) & 0 \\
      0 & H_{-1}(k_y) \\
    \end{array}
  \right),
\end{equation}
where $H_\lambda$ with $\lambda=\pm 1$ are $N_{\textrm{occ}} \times N_{\textrm{occ}}$ matrices.

If $|u_1(k_y=0)\rangle$ is an eigenstate of $H_1(k_y=0)$ corresponding to an eigenenergy $E_1(k_y=0)$,
then $\hat{m}_y|u_1(k_y=0)\rangle$ is an eigenstate of $H_{-1}(k_y=0)$ with the same energy because $H\hat{m}_y|u_1(k_y=0)\rangle=E_1(k_y=0)\hat{m}_y|u_1(k_y=0)\rangle$ and $\hat{m}_y|u_1(k_y=0)\rangle\in \textrm{span}\{\beta_{-1}\}$.
This tells us that, at half filling, only half of eigenstates of $H_1$ and $H_{-1}$ are occupied.
Consider an occupied basis
$\beta_{\textrm{occ}}=\beta_{\textrm{occ},{1}}\cup \beta_{\textrm{occ},{-1}}$
with $\beta_{\textrm{occ},{\lambda}}=\{|u_{{\lambda},(k_x=0,k_y=0)}^1\rangle,\cdots, |u_{{\lambda},(k_x=0,k_y=0)}^{N_{\textrm{occ}}/2}\rangle\}$ with $\lambda=\pm 1$, where
$|u_{{\lambda},(k_x=0,k_y=0)}^j\rangle$ with $j=1,\cdots,N_{\textrm{occ}}/2$ being eigenstates of
$H_{\lambda}(k_x=0,k_y=0)$. In this basis,
\begin{equation}
B_{m_x,(0,0)}=\left(
    \begin{array}{cc}
      I_{N_{\textrm{occ}}/2} & 0 \\
      0 & -I_{N_{\textrm{occ}}/2} \\
    \end{array}
  \right),
\label{BMatrix}
\end{equation}
where $I_{N_{\textrm{occ}}/2}$ is a $N_{\textrm{occ}}/2 \times N_{\textrm{occ}}/2$ identity matrix.
For nonzero $k_y$, the number of occupied states in $H_1$ or $H_{-1}$ should remain unchanged; otherwise,
the Hamiltonian would not be an insulator.
For example, if
$|u_{{\lambda},(k_x=0,k_y=0)}^1\rangle$ becomes unoccupied as we continuously vary $k_y$, its energy
must intersect with the Fermi level, leading to a metallic phase instead of an insulating phase.
This tells us that
$B_{m_x,(k_x=0,\pi,k_y)}$ takes the same form as in Eq.~(\ref{BMatrix}) in the basis
$\beta_{\textrm{occ},{\lambda}}=\{|u_{{\lambda},(k_x=0,\pi,k_y)}^1\rangle,\cdots, |u_{{\lambda},(k_x=0,\pi,k_y)}^{N_{\textrm{occ}}/2}\rangle\}$ with $\lambda=\pm 1$, where
$|u_{{\lambda},(k_x=0,\pi,k_y)}^j\rangle$ with $j=1,\cdots,N_{\textrm{occ}}/2$ being eigenstates of
$H_{\lambda}(k_x=0,\pi,k_y)$, respectively.
\end{proof}

Based on the above theorem, we can write Eq.~(\ref{BSmmetry1}) and Eq.~(\ref{BSmmetry2}) into
\begin{align}
&S\mathcal{W}_{\pi \leftarrow 0}^{\epsilon=0}(k_y) S^{\dagger} = -\mathcal{W}_{\pi\leftarrow 0}^{\epsilon=0}(k_y) \label{BSmmetry11} \\
&S\mathcal{W}_{\pi \leftarrow 0}^{\epsilon=\pi}(k_y) S^{\dagger} = \mathcal{W}_{\pi\leftarrow 0}^{\epsilon=\pi}(k_y),
\label{BSmmetry22}
\end{align}
where
\begin{equation}
S=\left(
  \begin{array}{cc}
    I_{N_{\textrm{occ}}/2} & 0 \\
    0 & -I_{N_{\textrm{occ}}/2} \\
  \end{array}
\right).
\end{equation}
As a result,
\begin{eqnarray}
\mathcal{W}_{\pi \leftarrow 0}^{\epsilon=0}(k_y)&=&\left(
  \begin{array}{cc}
    0 & U_+^{\epsilon=0}(k_y) \\
    U_-^{\epsilon=0}(k_y) & 0 \\
  \end{array}
\right)\\
\mathcal{W}_{\pi \leftarrow 0}^{\epsilon=\pi}(k_y)&=&\left(
  \begin{array}{cc}
    U_+^{\epsilon=\pi}(k_y) & 0 \\
    0 & U_-^{\epsilon=\pi}(k_y) \\
  \end{array}
\right).
\end{eqnarray}

In our system, $N_{\textrm{occ}}=2$ and thus $U_+^{\epsilon=0,\pi}(k_y)$ are complex numbers,
we can define a winding number at $\epsilon=0,\pi$ as
\begin{equation}
W_{\nu_x}^{\epsilon}=\frac{1}{2\pi i}\int_0^{2\pi} dk_y\partial_{k_y}\log (U_{+}^{\epsilon}(k_y)),
\end{equation}

\subsection*{3. Gauge transformation}
Since the Wilson line is determined by the occupied eigenstates of a Hamiltonian, the winding number of the Wilson line
introduced in the preceding section may be dependent of the gauge transformation of the occupied eigenstates.
Specifically, if we multiply a global phase to an occupied eigenstate, i.e., $|u_{k_x^{*},k_y}^{m_x}\rangle \rightarrow e^{i\theta_{m_x}(k_x^{*},k_y)}|u_{k_x^{*},k_y}^{m_x}\rangle$ with $k_x^*=0,\pi$ and a reflection eigenvalue $m_x=\pm 1$,
then
\begin{equation}
\mathcal{W}_{\pi \leftarrow 0}^{\epsilon}(k_y)
\rightarrow
L^\dagger(k_x=\pi,k_y) \mathcal{W}_{\pi \leftarrow 0}^{\epsilon}(k_y) L(k_x=0,k_y),
\end{equation}
where
\begin{equation}
L(k_x,k_y)=\left(
             \begin{array}{cc}
               e^{i\theta_{+1}(k_x,k_y)} & 0 \\
               0 & e^{i\theta_{-1}(k_x,k_y)} \\
             \end{array}
           \right).
\end{equation}
This gives us the transformation for the Wilson line for $\epsilon=0,\pi$,
\begin{eqnarray}
\mathcal{W}_{\pi \leftarrow 0}^{\epsilon=0}(k_y)&\rightarrow &\left(
  \begin{array}{cc}
    0 & U_+^{\epsilon=0}(k_y)e^{i[\theta_{-1}(0,k_y)-\theta_{+1}(\pi,k_y)]} \\
    U_-^{\epsilon=0}(k_y)e^{i[\theta_{+1}(0,k_y)-\theta_{-1}(\pi,k_y)]} & 0 \\
  \end{array}
\right)\\
\mathcal{W}_{\pi \leftarrow 0}^{\epsilon=\pi}(k_y)&\rightarrow &\left(
  \begin{array}{cc}
    U_+^{\epsilon=\pi}(k_y)e^{[\theta_{+1}(0,k_y)-\theta_{+1}(\pi,k_y)]} & 0 \\
    0 & U_-^{\epsilon=\pi}(k_y)e^{i[\theta_{-1}(0,k_y)-\theta_{-1}(\pi,k_y)]} \\
  \end{array}
\right).
\end{eqnarray}
Thus, the winding number changes to
\begin{eqnarray}
W_{\nu_x}^{\epsilon}&\rightarrow &\frac{1}{2\pi i}\int_0^{2\pi} dk_y\partial_{k_y}\log (U_{+}^{\epsilon}(k_y)e^{i\Theta^{\epsilon}}(k_y)) \\
&=& N_{\nu_x}^{\epsilon}+\frac{1}{2\pi}[\Theta^{\epsilon}(k_y=2\pi)-\Theta^{\epsilon}(k_y=0)],
\end{eqnarray}
where $\Theta^{\epsilon=0}\equiv \theta_{-1}(0,k_y)-\theta_{+1}(\pi,k_y)$, $\Theta^{\epsilon=\pi}\equiv \theta_{+1}(0,k_y)-\theta_{+1}(\pi,k_y)$. This shows that the global phase at $k_x=0,\pi$ can produce an
unphysical change of the winding number once the global phase exhibits a winding, implying that an isolated value
of the winding number is not physical. However, if we maintain this global phase unchanged as we vary a system parameter
and observe the change of the winding number, this change is physical and tells us that the edge polarization appears
or disappears. This works as the polarization is a $Z_2$ quantity.

Numerically, we need to maintain the continuity of the global phase of the occupied energy states with respect to $k_y$ along the reflection symmetric lines $k_x=0$ and $k_x=\pi$ for a fixed system parameter $\gamma$ and to maintain the continuity of the Berry phase of these states about $k_y$ with respect to $\gamma$.

To maintain the continuity of the global phase with respect to $k_y$, we first make the wave function $|u_{k_x^{*},k_y}^{m_x}\rangle$ continuous
between two neighboring points $(k_x^{*},k_y)$ and $(k_x^{*},k_y+\Delta k_y)$,
where $k_y=n\Delta k_y$ with $\Delta k_y=\frac{2\pi}{N_y}$ and $n=0,1,\cdots,(N_y-1)$. To achieve this,
we calculate the overlap between the two wave functions of neighboring
momenta $\langle u_{k_x^{*},k_y+\Delta k_y}^{m_x}|u_{k_x^{*},k_y}^{m_x}\rangle$, and eliminate the numerical phase difference by the transformation,
\begin{equation}
    |u_{k_x^{*},k_y+\Delta k_y}^{m_x}\rangle\rightarrow e^{i\phi_1}|u_{k_x^{*},k_y+\Delta k_y}^{m_x}\rangle,
    \end{equation}
where $\phi_1=\text{Im}\log \langle u_{k_x^{*},k_y+\Delta k_y}^{m_x}|u_{k_x^{*},k_y}^{m_x}\rangle$.
We repeat this process from $k_y=0$ to $k_y=2\pi-\Delta k_y$.
After that, we make the wave functions continuous
between $k_y=2\pi-\Delta k_y$ and $k_y=0$ by performing the following phase transformations
\begin{equation}
    |u_{k_x^{*},n\Delta k_y}^{m_x}\rangle\rightarrow e^{-i\phi_2n/N_y}|u_{k_x^{*},n\Delta k_y}^{m_x}\rangle,
    \end{equation}
where $\phi_2=\text{Im}\log \langle u_{k_x^{*},0}^{m_x}|u_{k_x^{*},2\pi-\Delta k_y}^{m_x}\rangle$.

To make sure that the Berry phase of each occupied band about $k_y$
along the reflection symmetric lines $k_x^*=0,\pi$ is continuous as $\gamma$ varies,
we numerically compute the Berry phase based on the following formula
\begin{equation}
\nu_y^{\pm1}(k_x^*=0,\pi)=\text{Im}\lim_{N_y\rightarrow\infty}\sum_{n=0}^{N_y-1}\log  \langle u_{k_x^*,(n+1)\Delta k_y}^{\pm1}|u_{k_x^*,n\Delta k_y}^{\pm1}\rangle
\end{equation}
and perform a suitable transformation
$|u_{k_x^{*},k_y=n\Delta k_y}^{m_x}\rangle\rightarrow \exp(-i2\pi p \frac{n}{N_y})|u_{k_x^{*},k_y=n\Delta k_y}^{m_x}\rangle$ with $p$ being an integer provided that the Berry phases between two neighboring $\gamma$ change by $-2 p \pi$.
The procedure is similar for calculation of the winding number $W_{\nu_y}^\epsilon$.

\section*{Appendix E: More phase diagrams for simplified models}

\setcounter{equation}{0} \setcounter{figure}{0} \setcounter{table}{0} %
\renewcommand{\theequation}{E\arabic{equation}} \renewcommand{\thefigure}{E%
\arabic{figure}} \renewcommand{\bibnumfmt}[1]{[#1]} \renewcommand{%
\citenumfont}[1]{#1}

\begin{figure*}[t]
\includegraphics[width=6.5in]{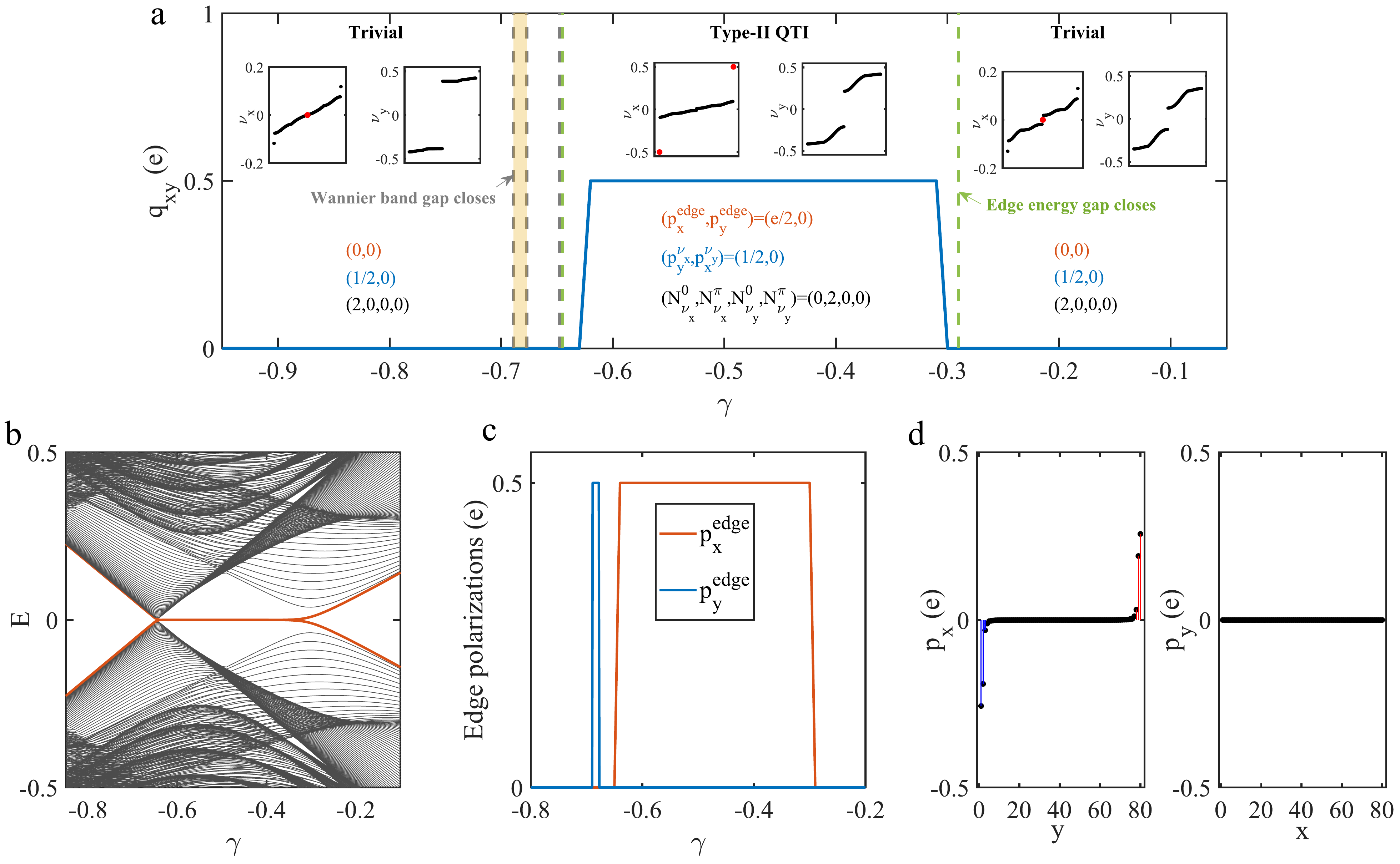}
  \caption{a, Phase diagram for the Hamiltonian $H_{II}$ (with $b_2=1.2$ and $g_0=-0.29$) versus $\gamma$, where the quadrupole moment is plotted as a blue line.
  It displays the topologically trivial insulator, the type-II QTI and the new phase with nonzero edge polarizations but without quadrupole moments
  and zero-energy corner modes (displayed in the light red region).
  The subsets display the same quantities as in Fig.~\ref{fig2}.
  The small discrepancy between the edge gap closing points and the quadrupole moment transition points
  is caused by the finite size effects for calculating the quadrupole moment in a finite size system with $N=240$.
  b, The energy spectrum with respect to $\gamma$ for open boundary conditions along all directions with zero-energy corner modes being highlighted by a red line.
  c, The edge polarizations calculated based on the formula (\ref{windingR}) by choosing a gauge such that $W_{\nu_x}^{\epsilon=\pi}(\gamma_0=-0.8)=W_{\nu_y}^{\epsilon=\pi}(\gamma_0=-0.8)=0$ because it is topologically
  trivial when $\gamma=-0.8$.
  d, The spatial distribution of the edge polarization in a type-II QTI with $\gamma=-0.4$. }
\label{figE1}
\end{figure*}

In the main text, we have demonstrated the existence of these new topological phenomena by studying systems with
the particle-hole symmetry. In this appendix, we will present phase diagrams for more models,
in particular, the models with the chiral symmetry.

We first show the phase diagram and topological properties for the Hamiltonian $H_{II}$
in Eq.~(\ref{SimHam}) with $b_2=1.2$ and $g_0=-0.29$ in Fig.~\ref{figE1}.
Similar to Fig.~\ref{fig5} in the main text, we can clearly see the existence of the type-II QTI in the
phase diagram. This phase is also a normal type-II QTI given that the Wannier spectrum $\nu_x$ has edge
states only at $\nu_x=\pm 1/2$. Besides this type-II phase, we also observe the new topological phase with only nonzero $p_y^{\textrm{edge}}$ in the light red region, which arises due to the Wannier gap closing
at $\nu_y=\pm 1/2$. We note that across $\gamma=-0.645$, the edge polarization $p_x^{\textrm{edge}}$ appears due to the
the $y$-normal edge energy gap closing rather than the Wannier gap closing. The Wannier gap closing happens at $\nu_x=0$ and thus
does not contribute to the edge polarizations.

Similarly, we map out the phase diagram for the Hamiltonian $H_{III}$ in Eq.~(\ref{HamSim2})
with $b_2=1.2$ and $g_0=0.65$ [see Fig.~\ref{figE2}]. This Hamiltonian preserves both the particle-hole
and chiral symmetry. Evidently, the phase diagram exhibits
the type-II QTI as well as the new topological phase with only quantized edge polarizations. By adding a term $\alpha \sin k_x\sigma_1\otimes\sigma_1$ that breaks
the particle-hole symmetry but preserves the chiral symmetry, we can still observe these topological phenomena,
indicating that these phases can arise for a system with the chiral symmetry.

Finally, in Fig.~\ref{figE3}, we provide the energy spectra, the Wannier spectrum, the quadrupole moment,
the edge polarizations and the winding number for the Hamiltonian
$H_{IV}+\alpha \sin k_x \sigma_1\otimes\sigma_1$ with $b_2=1.2$ and $\alpha=0.2$ that preserves the chiral symmetry
and $H_{IV}+\alpha \sin k_x \sigma_3\otimes\sigma_3$ with $b_2=1.2$ and $\alpha=0.2$ that preserves the particle-hole symmetry.
We can see that all these figures are very similar to Fig.~\ref{fig6} in the main text, further reflecting that the new topological phase can appear in a system with either the particle-hole symmetry or the chiral symmetry.

\begin{figure*}[t]
\includegraphics[width=6.5in]{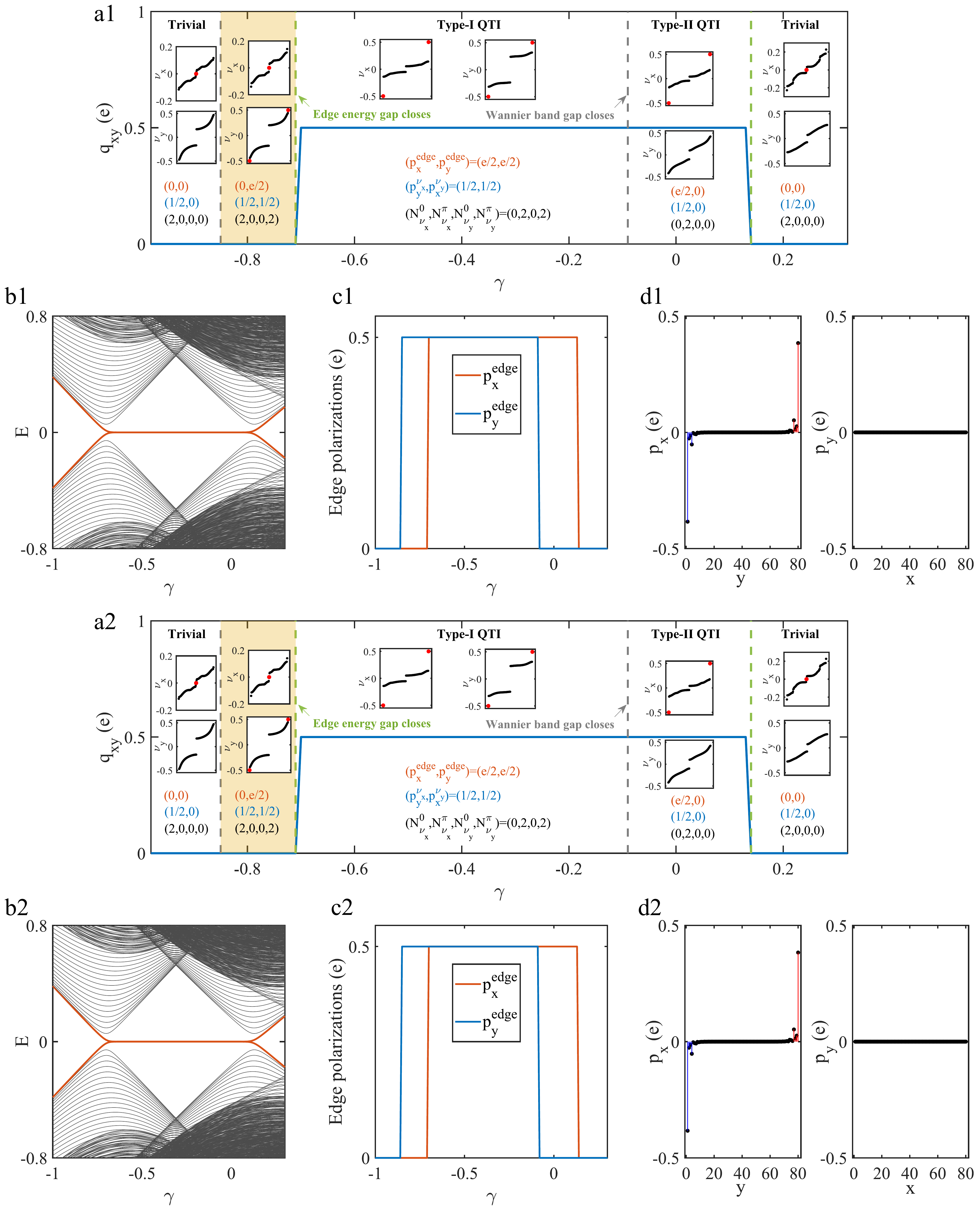}
  \caption{a1, Phase diagram for the Hamiltonian $H_{III}$ (with $b_2=1.2$ and $g_0=0.65$) versus $\gamma$, where the quadrupole moment is plotted as a blue line.
  The light red region shows the new topological phase with nonzero edge polarizations but without quadrupole moments
  and zero-energy corner modes.
  The subsets display the same quantities as in Fig.~\ref{fig2}.
  b1, The energy spectrum with respect to $\gamma$ for open boundary conditions along all directions with zero-energy corner modes being denoted by a red line.
  c1, The edge polarizations calculated based on the formula (\ref{windingR}) by choosing a gauge such that $W_{\nu_x}^{\epsilon=\pi}(\gamma_0=-1)=W_{\nu_y}^{\epsilon=\pi}(\gamma_0=-1)=0$ because it is topologically
  trivial when $\gamma=-1$.
  d1, The spatial profiles of the edge polarization in a type-II QTI.
  a2-d2 display the same quantities as a1-d1 but for the Hamiltonian $H_{III}+\alpha \sin k_x\sigma_1\otimes\sigma_1$ (with $b_2=1.2$,
  $g_0=0.65$ and $\alpha=0.2$). We use $N=80$ for evaluation of the quadrupole moment.
  }
\label{figE2}
\end{figure*}

\begin{figure*}[t]
\includegraphics[width=6.5in]{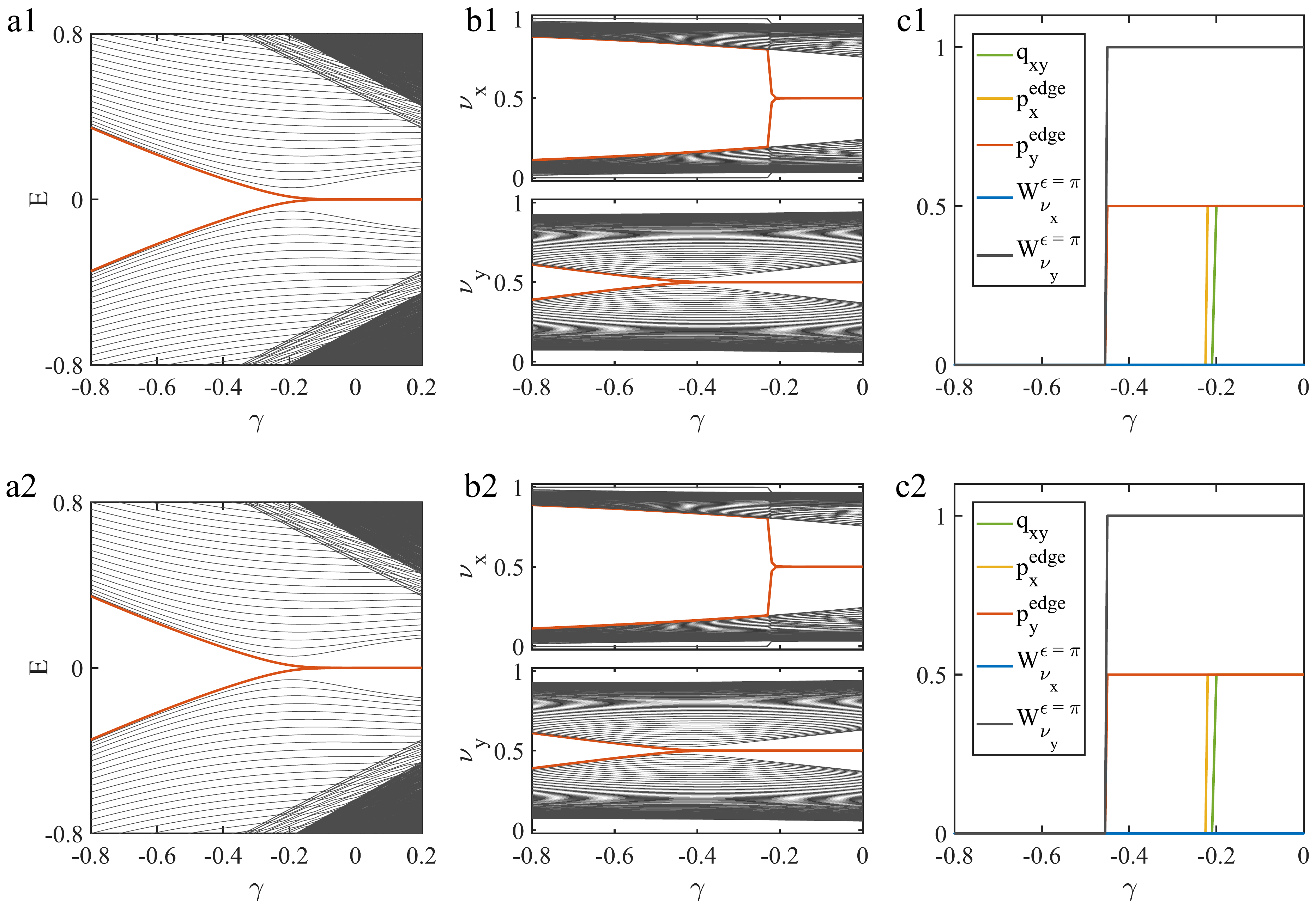}
  \caption{
  a1, a2, The energy spectrum versus $\gamma$ in a geometry with open boundaries along all directions.
  b1, b2, The Wannier band $\nu_x$ and $\nu_y$ versus $\gamma$.
  c1, c2, The quadrupole moment $q_{xy}$,
  the edge polarizations $p_{\mu}^{\textrm{edge}}$ ($\mu=x,y$) computed using the formula (\ref{windingR}) and the winding
  numbers $W_{\nu_\mu}^{\epsilon=\pi}$ ($\mu=x,y$) computed using the formula (\ref{WindFormula}) versus $\gamma$.
  When calculating the edge polarizations and the winding number, we choose a gauge
  such that $W_{\nu_x}^{\epsilon=\pi}(\gamma_0=-0.8)=W_{\nu_y}^{\epsilon=\pi}(\gamma_0=-0.8)=0$
  because the phase is topologically trivial when $\gamma=-0.8$.
  $q_{xy}$ and $p_x^{\textrm{edge}}$ are hidden behind the blue line when $\gamma<-0.22$. The small discrepancy between
  the quadrupole moment transition point and the edge gap closing point is caused by the finite size effects.
  Here we use $N=80$ for evaluation of the quadrupole moment.
  In a1-c1 and a2-c2, we consider the Hamiltonian $H_{IV}+\alpha \sin k_x \sigma_1\otimes\sigma_1$ with $b_2=1.2$ and $\alpha=0.2$,
  and the Hamiltonian $H_{IV}+\alpha \sin k_x \sigma_3\otimes\sigma_3$ with $b_2=1.2$ and $\alpha=0.2$, respectively.
  }
\label{figE3}
\end{figure*}

\section*{Appendix F: Energy and Wannier spectra during a pumping process}

\setcounter{equation}{0} \setcounter{figure}{0} \setcounter{table}{0} %
\renewcommand{\theequation}{F\arabic{equation}} \renewcommand{\thefigure}{F%
\arabic{figure}} \renewcommand{\bibnumfmt}[1]{[#1]} \renewcommand{%
\citenumfont}[1]{#1}

In the main text, we have shown the transport of the corner charge, the quadrupole moment and the
edge polarizations as system parameters vary over an entire cycle. In this appendix, we present the energy spectrum in
an open boundary geometry (see the first column in Fig.~\ref{figF1}) and the Wannier spectrum in a cylinder geometry
(see the last two columns in Fig.~\ref{figF1}) during this full process.
The first row in the figure corresponds to the pump from
a topologically trivial phase to a type-II AQTI and then back to the trivial phase, while the
second row to the pump from a type-II AQTI to a type-I AQTI and then back to the original type-II phase.
It is clear that we can divide the energy spectrum and Wannier spectrum into two bands with gaps between them.
In the former scenario, we see that the corner states connect the lower energy band to the higher
one as time evolves, which shows the chiral property of the corner states if time is regarded as a third momentum;
this agrees well with the quantized transport of corner charges shown in the main text.
For the Wannier spectrum $\nu_x$, the edge states exist inside both the gaps around 0 and $1/2$
and these states connect two bands of the Wannier spectrum as time progresses, similar to the
energy spectrum. However, for the Wannier spectrum $\nu_y$, we do not see the presence
of the edge states connecting two Wannier bands, consistent with the zero net transport for
the edge polarization $p_y^{\textrm{edge}}$. In the second pumping scenario, while the corner states exist
during the full cycle, they do not connect the two energy bands, and thus are not "chiral",
which is consistent with the zero transport of the corner charges. Similar band patterns
occur in the Wannier spectrum $\nu_x$. However, for the Wannier spectrum $\nu_y$, we find the "chiral"-like edge
states, which explains the presence of the net transport of the edge polarization $p_y^{\textrm{edge}}$
shown in the main text.

\begin{figure*}[t]
\includegraphics[width=\textwidth]{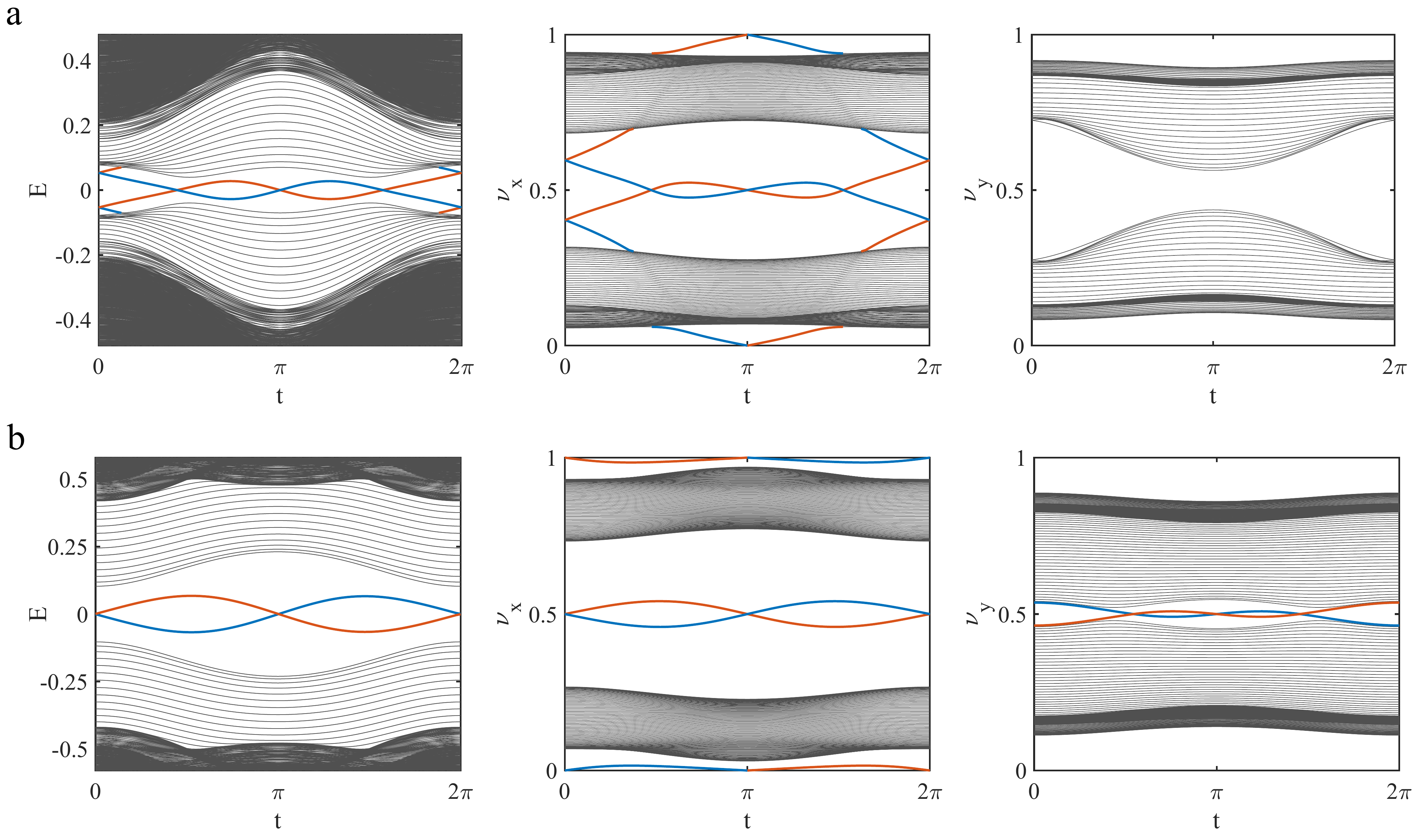}
  \caption{The energy spectrum and Wannier spectrum during the pumping process.
  a, The cycle from a topologically trivial phase to a type-II AQTI and back to the trivial phase.
  b, The cycle from a type-II AQTI to a type-I AQTI and back to the type-II AQTI.
  For both a and b, the first column corresponds to the energy spectrum with open boundary along $x$ and $y$ and
  the second (third) column correspond to the Wannier spectrum $\nu_x$ ($\nu_y$) with open boundary along $y$ ($x$)
  and periodic boundary along $x$ ($y$).
   }
\label{figF1}
\end{figure*}

\section*{Appendix G: Entanglement spectra}

\setcounter{equation}{0} \setcounter{figure}{0} \setcounter{table}{0} %
\renewcommand{\theequation}{G\arabic{equation}} \renewcommand{\thefigure}{G%
\arabic{figure}} \renewcommand{\bibnumfmt}[1]{[#1]} \renewcommand{%
\citenumfont}[1]{#1}

In this appendix, we will define the entanglement spectra and show how to calculate them. Let us first partition a system into two subsystems labelled by $A$ and $B$, respectively.
The reduced density matrix for the subsystem $A$ can be obtained by performing partial trace over the subsystem $B$, that is,
\begin{equation}
\rho_A = {\rm Tr}_B \ket{\Psi} \bra{\Psi} = \frac{e^{-H_A}}{Z_A},
\end{equation}
where $|\Psi\rangle$ is a many-body ground state, $H_A$ is defined as a Hamiltonian corresponding to the reduced density matrix $\rho_A$ and $Z_A={\rm Tr} e^{-H_A}$. The entanglement spectrum refers to the eigenvalues of $\rho_A$~\cite{Haldane2008PRL,Pollmann2010PRB,Fidkowski2010PRL}.

In the single-particle case, the entanglement spectrum can be determined by
diagonalizing the correlation matrix~\cite{Peschel2003}
\begin{equation}
[C_A]_{ij} = \left \langle \hat{c}^{\dagger}_i \hat{c}_j  \right \rangle,
\end{equation}
where $i,j\in A$. If we diagonalize $H_A$ as $H_A=\sum_n \varepsilon_n \hat{a}_n^\dagger \hat{a}_n$,
the single-particle entanglement spectrum $\xi_n$ and $\varepsilon_n$ are related by
\begin{equation}
\xi_n=\frac{1}{e^{\varepsilon_n}+1}.
\end{equation}
Clearly, the entanglement spectrum $\xi_n=0.5$ corresponds to an entanglement zero mode $\varepsilon_n=0$. In the main text,
we show the entanglement spectrum $\textrm{ES}_x$ and $\textrm{ES}_y$ obtained by tracing out the right part and top part
of a system as shown in Fig.~\ref{figG1}(a) and (b), respectively.

To characterize the edge modes of quadrupole topological insulators, which are localized at
the corners, nested entanglement spectra are introduced~\cite{Bernevig2018SciAdv},
as detailed in the following.

For a 2D quadrupole insulator, the Hamiltonian can be diagonalized as
\begin{equation}
H = \sum_{k_x, k_y, n} E_{k_x,k_y}^n \hat{f}^{\dagger}_{k_x ,k_y, n} \hat{f}_{k_x ,k_y ,n},
\end{equation}
where $\hat{f}^{\dagger}_{k_x ,k_y, n} = \sum_{\alpha}[u_{k_x,k_y}^n]^{\alpha}\ \hat{c}^{\dagger}_{k_x,k_y,\alpha}$
and $|u_{k_x,k_y}^n\rangle$ is the $n$th eigenstate of $H({\bf k})$ corresponding to the eigenenergy
$E_{k_x,k_y}^n$. Suppose that the system has $L_x\times L_y$ unit cells with $L_x=2N_x$ and $L_y=2N_y$. We first partition the
system into two subsystems $A$ and $B$ along the $x$ direction:
$A=\{(x,y)|1\leq x \leq N_x , 1\leq y \leq L_y\}$ and
$B=\{(x,y)|N_x < x \leq L_x , 1\leq y \leq L_y\}$.
The correlation matrix in the subsystem $A$ is given by
\begin{equation}
[C_{A,k_y}]_{x \alpha,x'\alpha'} = \left\langle \Psi_G|\hat{c}^{\dagger}_{x\alpha,k_y} \hat{c}_{x'\alpha', k_y}|\Psi_G \right\rangle
=\frac{1}{L_x} \sum_{k_x} e^{i k_x (x-x')} \sum_{n\in occ}
[U_{k_x,k_y}]_{\alpha n} [U_{k_x,k_y}^\dagger]_{n \alpha'},
\end{equation}
where $|\Psi_G\rangle$ is the many-body ground state of $H$ and $U_{k_x,k_y}$ consists of
all occupied eigenstates $|u_{k_x,k_y}^n\rangle$ as column vectors.
Diagonalizing $C_{A,k_y}$ yields the eigenvalues $\xi_{k_y}^m$ and eigenvectors $|v_{k_y}^m\rangle$
of $C_{A,k_y}$. The Hamiltonian $H_A$ of the reduced density matrix $\rho_A$ can be diagonalized as
\begin{equation}
H_A = \sum_{k_y,m}\log(\frac{1}{\xi_{k_y}^m} - 1)\hat{g}^{\dagger}_{k_y,m} \hat{g}_{k_y,m},
\end{equation}
where $g^{\dagger}_{k_y,m} = \sum_{\alpha x} [v_{k_y}^m]_{x\alpha } \  c^{\dagger}_{x\alpha,k_y}$.
The subsystem $A$ is further partitioned into two parts along the $y$ direction [as shown in Fig.~\ref{figG1}(c)]:
$A_1=\{ (x,y)| 1\leq x \leq N_x , 1\leq y \leq N_y \}$ and
$A_2=\{ (x,y)| 1\leq x \leq N_x , N_y < y \leq L_y \}$. The correlation matrix in the region $A_1$ is given by
\begin{equation}
[C_{A_1}]_{ x y\alpha,x^\prime y^\prime \alpha^\prime} = \left\langle \Psi_{A}|c^{\dagger}_{xy\alpha} c_{x^\prime y^\prime\alpha^\prime} |\Psi_{A} \right\rangle
= \frac{1}{L_y} \sum_{k_y} e^{i k_y(y-y')} \sum_{m\in occ}
[V_{k_y}]_{x\alpha,m} [V_{k_y}^{\dagger}]_{m,x'\alpha'},
\end{equation}
where $|\Psi_{A}\rangle$ is the many-body ground state of $H_A$ and the average is over all
occupied states of $H_A$. $V_{k_y}$ is made up of all occupied eigenvectors $|v_{k_y}^m\rangle$ as column vectors.
We can determine the nested entanglement spectrum by calculating the eigenvalues of $C_{A_1}$.

\begin{figure*}[t]
\includegraphics[width=5in]{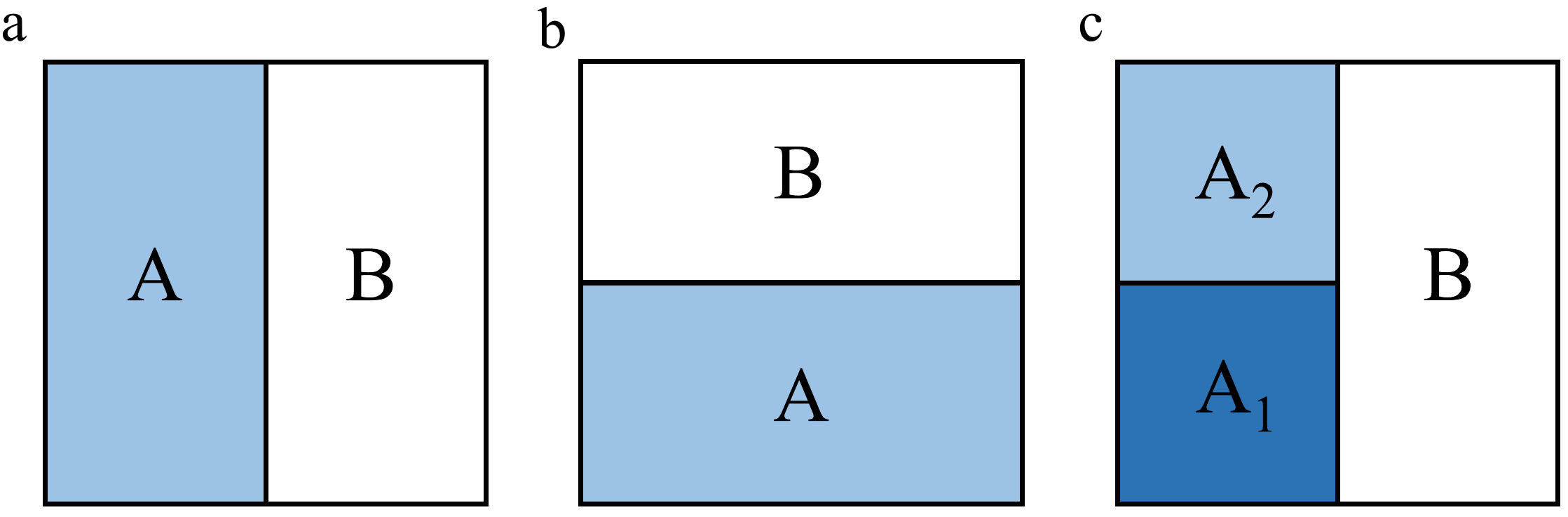}
  \caption{The entanglement spectrum $\textrm{ES}_x$ in a and $\textrm{ES}_y$ in b are evaluated by partitioning a system into
  two subsystems $A$ and $B$ and tracing out the subsystem $B$. c, The nested entanglement
  spectrum ${\textrm{ES}_{xy}}$ is obtained by further tracing out the subsystem $A_2$.
  }
\label{figG1}
\end{figure*}

\section*{Appendix H: Experimental realization}

\setcounter{equation}{0} \setcounter{figure}{0} \setcounter{table}{0} %
\renewcommand{\theequation}{H\arabic{equation}} \renewcommand{\thefigure}{H%
\arabic{figure}} \renewcommand{\bibnumfmt}[1]{[#1]} \renewcommand{%
\citenumfont}[1]{#1}

In this appendix, we discuss how to realize our Hamiltonian in electric circuits, in which the SSH model~\cite{Thomale2018CP},
Weyl semimetal~\cite{Simon2019PRB} and type-I quadrupole topological insulator~\cite{Thomale2018NP} have been experimentally achieved.
Let us consider an electrical network consisting of many nodes simulating sites in a tight-binding model.
For each node $m$ in the circuit, suppose that $I_m$ is the external current flowing into this node and
$V_m$ is the voltage at this node with respect to the ground, according to Kirchhoff's law, we have
\begin{equation}
I_m=\sum_n I_{mn}+I_{m0}=\sum_{n}X_{mn}(V_m-V_n)+X_m V_m,
\end{equation}
where $I_{mn}$ and $I_{m0}$ are the current flowing from node $m$ to $n$ and from node $m$ to the ground, respectively.
$X_{mn}=1/Z_{mn}$ is the admittance between node $m$ and $n$
with $Z_{mn}$ the corresponding impedance, and $X_{m}=1/Z_{m}$ is the admittance between node $m$ and the ground.
Writing $I_{m}$ and $V_{m}$ in the form of column vectors, we have
\begin{equation}
{\bf I} = J(\omega) {\bf V},
\end{equation}
where $J(\omega)$ denotes the circuit Laplacian with $\omega$ being the AC frequency of the input current.

By connecting appropriate capacitors, inductors and negative impedance converters with current inversion (INIC)~\cite{Thomale2019PRL,ChenBook} as shown in Fig.~\ref{fig10} in the main text, we can achieve a Laplacian simulating our
Hamiltonian, i.e., $J=iH$.
The sign of the resistance of a INIC depends on how it is connected. For example,
for the configuration shown in Fig.~\ref{fig10} in the main text, the current from node $1$ to node $2$ within a unit cell
is determined by $I_{12}=(V_1-V_2)/(-R)$ corresponding to a negative resistance while the current from
node $2$ to node $1$ determined by $I_{21}=(V_2-V_1)/R$ corresponding to a positive resistance.
To eliminate unnecessary onsite terms, we also need to add onsite impedances $Z_{m}'$ as shown in Fig.~\ref{fig10} in the main text.
For the system with periodic boundary conditions, the values of $Z_{a=1,2,3,4}'$ are as follows,
\begin{align}
&\frac{1}{Z_1'}=\frac{1}{Z_4'}= (2t_1'-2t_2-2t_2'+\Delta) + i (4t_1+2t_1'+6t_2-2t_2') \\
&\frac{1}{Z_2'}=\frac{1}{Z_3'}= (-2t_1'+2t_2+2t_2'-\Delta) + i (2t_1'+2t_2-6t_2'+2\gamma).
\end{align}

We follow the experimental approach to directly measure the Green's function~\cite{Thomale2019PRB}.
Specifically, we apply an input current $I_n$ at one node $n$ of the circuit and measure the voltage
$V_m^{(n)}$ at node $m$, giving us the single-point impedance,
\begin{equation}
G_{mn}=V_m^{(n)}/I_n,
\end{equation}
which is a Green's function of the system Hamiltonian, i.e.,
\begin{equation}
G_{mn}=(J^{-1})_{mn}.
\end{equation}
All the information, such as the energy spectrum and eigenstates, can be extracted from the Green's function.

To achieve our model, we consider an electric circuit composed of $N_xN_y$ unit cells;
each unit cell contains $4$ nodes and each node is labeled by $({\bf R},\alpha)$.
If we consider a torus geometry with periodic boundaries along both $x$ and $y$,
we only need to apply a current in one node $({\bf 0},\beta)$ and measure the impedance $G_{\alpha\beta}({\bf R})$ between the node $(\bf R,\alpha)$ and node $({\bf 0},\beta)$ for all the different nodes $(\bf R,\alpha)$, leading to the Green's function in momentum space,
\begin{equation}
G_{\alpha\beta}({\bf k})=\sum_{\bf R} G_{\alpha\beta}({\bf R}) \exp(-i{\bf k}\cdot{\bf R}),
\end{equation}
where $\bf k$ is the momentum and $\sum_{\bf R}$ is the sum over all unit cells. Similarly,
in a cylinder geometry, e.g., with open boundaries along $y$ and periodic boundaries along $x$,
only the impedance $Z_{\alpha\beta}((R_x,R_y),(0,R_y^\prime))=G_{\alpha\beta}((R_x,R_y),(0,R_y^\prime))$
between the node $((R_x,R_y),\alpha)$ and node
$((0,R_y^\prime),\beta)$ is required to be probed, yielding the Green's function
\begin{equation}
G_{(R_y,\alpha)(R_y^\prime,\beta)}({\bf k})=\sum_{R_x} G_{\alpha\beta}((R_x,R_y),(0,R_y^\prime)) \exp(-i{k_x}{R_x}).
\end{equation}
Once the Green's functions in the torus and cylinder geometry are measured, the quadrupole moment and edge polarizations can be obtained.

For the system with open boundary conditions, the presence of zero-energy corner states
of the Hamiltonian will give rise to the divergence of the two-point impedance near the corners.
Thus we can measure the corner modes through the measurement of the resonance of the impedance
between two neighbouring nodes near the corners at the resonance frequency~\cite{Thomale2018NP}.

\end{widetext}

\end{document}